\documentclass[draftclsnofoot, onecolumn, 12pt]{IEEEtran} 
\usepackage{graphicx,epsfig}
\usepackage[noadjust]{cite}
\usepackage{mcite}
\usepackage{amsfonts,helvet}
\usepackage{fancyhdr}
\usepackage{threeparttable}
\usepackage{epsf,epsfig}
\usepackage{amsthm}
\usepackage{amsmath}
\usepackage{siunitx}
\usepackage{amssymb}
\usepackage{dsfont}
\usepackage{subfigure}
\usepackage{color}
\usepackage{algorithm}
\usepackage{algpseudocode}
\usepackage{algcompatible}
\usepackage{enumerate}
\usepackage{cancel}
\usepackage{lipsum}

\usepackage{eucal}

\newtheorem{theorem}{Theorem}

\newtheorem{corollary}{Corollary}

\newtheorem{lemma}{Lemma}

\newtheorem{proposition}{Proposition}
\newtheorem{remark}{Remark}

\newcommand\blfootnote[1]{%
  \begingroup
  \renewcommand\thefootnote{}\footnote{#1}%
  \addtocounter{footnote}{-1}%
  \endgroup
}

\begin{document}

\title{Feedback Design for Multi-Antenna $K$-Tier Heterogeneous Downlink Cellular Networks}

\author{Jeonghun~Park, Namyoon~Lee, and Robert~W.~Heath Jr.
%\IEEEauthorblockA{Department of Electrical and Computer Eng.\\
%The University of Texas at Austin\\
%Austin, TX 78712 USA\\
%Email: \{jeonghun\}@utexas.edu}}
\thanks{J. Park and R. W. Heath Jr. are with the Wireless Networking and Communication Group (WNCG), Department of Electrical and Computer Engineering, 
The University of Texas at Austin, TX 78701, USA. (E-mail: $\left\{\right.$jeonghun, rheath$\left\}\right.$@utexas.edu)

N. Lee is with Dept. of Electrical Engineering, POSTECH, 37673, Republic of Korea. (Email:nylee@postech.ac.kr).
}}

\maketitle \setcounter{page}{1} 

\begin{abstract}
%We study feedback partition problems for non-cooperative/cooperative $K$-tier heterogeneous networks. 
% two scenarios, the non-cooperative and the cooperative case.
%We consider two sophisticated cellular systems: a \emph{cooperative} and a \emph{heterogeneous} cellular network. 
We characterize the ergodic spectral efficiency of a non-cooperative and a cooperative type of $K$-tier heterogeneous networks with limited feedback. 
In the non-cooperative case, a multi-antenna base station (BS) serves a single-antenna user using maximum-ratio transmission based on limited feedback. 
In the cooperative case, a BS coordination set is formed by using dynamic clustering across the tiers, wherein the intra-cluster interference is mitigated by using multi-cell zero-forcing also based on limited feedback. 
Modeling the network based on stochastic geometry, we derive analytical expressions for the ergodic spectral efficiency as a function of the system parameters. 
%, such as the feedback bits, the density, the number of antennas, the transmit power and the biasing factor of each tier. 
Leveraging the obtained expressions, we formulate feedback partition problems and obtain solutions to improve the ergodic spectral efficiency. Simulations show the spectral efficiency improvement by using the obtained feedback partitions. 
Our major findings are as follows: 1) In the non-cooperative case, the feedback is only useful in a particular tier if the mean interference is small enough. 2) In the cooperative case, allocating more feedback to stronger intra-cluster BSs is efficient. 3) In both cases, the obtained solutions do not change depending on instantaneous signal-to-interference ratio. 
%The proposed solution provides that how much feedback bits should be allocated to each tier BS under the total feedback bits constraint. Simulations show the performance gains by using the proposed methods. 
%Major distinguishable point is that the proposed feedback partitions are functions of the system parameters that does not change frequently, so that the implementation is easy.
\end{abstract}

\blfootnote{This material is based upon work supported in part by the National Science Foundation under Grant No. NSF-CCF-1319556.}

\blfootnote{A special case of this paper (feedback design in a cooperative single-tier network) appeared in \emph{IEEE Globecom Workshop} 2016 \cite{park:gc:16}.}

\section{Introduction}

\subsection{Motivation}

%Heterogeneous networks (HetNets) have been considered as a promising next-generation wireless network architecture. 
%Heterogeneity across each network tier (e.g., macro, pico, and femto) is a key feature of heterogeneous networks (HetNets). 
In heterogeneous networks (HetNets), different types of base stations (BSs) are densely deployed to aggressively reuse the spectrum. 
One bottleneck in achieving the full gains of HetNets is interference. Compared to conventional single-tier cellular networks, a HetNet has various interference sources including intra-tier BSs and also cross-tier BSs. As a result, small network tiers such as femto (whose transmit power is small) are vulnerable to the interference. 
%For this reason, managing interference well is key to provide high throughput in HetNets. 
%Provided that  HetNets promises large amount of performance improvement
%In HetNets, there exist various interference, which is not 
In dealing with interference in multiple antenna systems, channel state information at transmitters (CSIT) is necessary. 
In frequency division duplex (FDD) downlink cellular systems, limited knowledge of CSI is conventionally given to a base station (BS) by using limited feedback, where its accuracy is determined by the amount of feedback to quantize CSI. For this reason, it is important to use appropriate amount of feedack for providing high throughput. 
%, where the total feedback bits that can be used in a network is limited by the uplink overheads.  
%Since the total feedback bits that can be used in a network is usually limited by the uplink feedback overheads, it is important to determine the appropriate amount of feedback bits for each BS. 
%For this reason, the performance of interference management techniques depends on the amount of used feedback bits, where the total feedback bits is limited by the uplink overheads. 
%Since the total amount of used feedback bits are limited by the uplink feedback overhead, 
%More feedback bits reduces the CSI quantization error, while increasing the uplink feedback overhead. Accordingly, it is essential to determine an appropriate amount of the feedback bits to balance between the performance enhancement and the overheads from a system design perspective. 

Determining the right amount of feedback is not easy, though. The main reason is that the feedback performance depends on the signal-to-interference plus noise ratio (SINR), which is intricately determined by the various system parameters of the HetNet. 
%, resulting in that the feedback performance and the system parameters are tightly interwoven in the SINR.  
For this reason, revealing the relationship between the feedback performance and the system parameters should be preceded prior to determine the amount of feedback.
To clarify this relationship, in this paper, we model a HetNet based on stochastic geometry, allowing us to analyze the rate performance as a function of the key system parameters and the amount of feedback. Leveraging the obtained analytical expressions, we formulate feedback partition problems and propose solutions to maximize the ergodic spectral efficiency. The presented feedback partitions show how the feedback should be allocated depending on the density, the number of antennas, the transmit power, and the biasing factor of each tier in the HetNet.

\subsection{Prior Work}
%There has been extensive prior work on the downlink transmission rate as a function of a finite CSI feedback rate.
%There has been prior work on determining appropriate amount of feedback bits in a simple network setting where a single-tier BS serves a single or multiple users without cooperation. In \cite{jindal:tit:06}, it was shown that the full spatial multiplexing gain is achievable by using zero-forcing (ZF) if the CSI feedback rate logarithmically scales with the SNR. 
%Considering both of limited feedback and channel estimation via downlink training, the achievable rate was characterized  in \cite{caire:tit:10}. 
%In \cite{love:tsp:06, koba:tcom:11}, the cost of the CSI feedback was considered by normalizing the downlink spectral efficiency with the consumed uplink overheads.

There has been some prior work on determining appropriate amount of feedback in various network environments. 
In \cite{bhaga:tsp:11, ny:twc:11_adap}, it was assumed that single-tier BSs form a coordination set and use multi-cell zero-forcing to serve multiple users. In this setting, feedback allocating methods were proposed to obtain constant rate-loss compared to the perfect CSIT case. 
Similarly, in \cite{yuan:twc:13}, an adaptive feedback strategy was proposed for maximizing the achievable rate of coordinated beamforming in a HetNet. 
%In \cite{kim:tvt:15}, limited feedback in time varying channel was considered and a feedback strategy that jointly considers the resolution and period was proposed.
%an adaptive feedback period method, which balances the feedback resolution and the feedback period in time varying channel, was proposed.
Prior work also studied effects of limited feedback on the achievable spatial degrees of freedom (DoF) \cite{kerret:tit:12, rao:tsp:13,niu:wpmc:14}. 
%Specifically, in \cite{kerret:tit:12}, a distributed CSI model was introduced in a network MIMO, where each BS has its own local CSI feedback link and its precoder only depends on this link. In this model, the relation between the DoF and the rate of the feedback link was characterized. 
%In \cite{rao:tsp:13}, an adaptive feedback partition strategy was proposed in the interference channel with heterogeneous path-loss and spatial correlation. 
For example, in \cite{niu:wpmc:14}, the achievable DoF was characterized as a function of the feedback in a HetNet, where a macro BS serves multiple users and pico BSs serve a single user.
A common limitation of the prior work \cite{bhaga:tsp:11, ny:twc:11_adap, yuan:twc:13, kerret:tit:12, rao:tsp:13, niu:wpmc:14} is that a small network model is assumed, where only a few neighboring BSs and users are considered and their locations are deterministic. In this network model, it is hard to evaluate system-level performance obtained by averaging over many user's and BS's locations in a large size cellular network. 

%Considering a multi-antenna HetNet, the prior work \cite{yuan:twc:13,rihan:vtc:15,mokari:tvt:16} studied limited feedback in a HetNet.
%Resolving the first limitation mentioned above, t
%There is some work that studied limited feedback in a HetNet.
%Specifically, \cite{yuan:twc:13} assumed that different types of BSs coordinate to mitigate the cross-tier interference and proposed an adaptive feedback strategy depending on the location of a user associated with a pico cell.
%In \cite{rihan:vtc:15}, when IA is used in two different types BSs, limited feedback/feedforward strategy was proposed. \cite{mokari:tvt:16} addressed a resource allocation to each tier BS relying on limited feedback was proposed. 
%Since the work \cite{yuan:twc:13,rihan:vtc:15,mokari:tvt:16} considered a small deterministic network though, the second limitation mentioned above still remains. 

To overcome the small-network limitation of the prior work \cite{bhaga:tsp:11, ny:twc:11_adap, yuan:twc:13, kerret:tit:12, rao:tsp:13, niu:wpmc:14}, a Poisson point process (PPP) has been used to model a large size network. There is also prior work that explored limited feedback in a random network model based on a PPP. 
In \cite{akoum:tsp:13}, a BS cluster is randomly formed and the feedback is adaptively partitioned for this clustering. In \cite{li:tcom:15}, several closest BSs to a user were included in a cluster and the signal-to-interference plus noise ratio (SINR) performance with limited feedback was analyzed. In \cite{jh:twc:16}, the optimal feedback rate was obtained to maximize the net spectral efficiency, defined as the downlink rate normalized by the uplink feedback overheads. 
In \cite{kountouris:twc:12}, assuming an interference-limited ad-hoc network, the achievable rate was analyzed when spatial division multiple access is used with limited feedback. 
Similar to \cite{kountouris:twc:12}, in \cite{park:wcl:16}, an interference-limited device-to-device network was considered and the ergodic spectral efficiency of single-user maximum ratio transmission (MRT) was characterized as a function of the feedback.
The prior work \cite{akoum:tsp:13, li:tcom:15, jh:twc:16, kountouris:twc:12, park:wcl:16}, however, assumed a single-tier network where all the BSs use the same transmit power and the BSs are distributed by the same density. In HetNets, there are multiple tiers whose transmit power and densities are different, and this heterogeneity changes rate coverage expressions as shown in  \cite{hsjo:2012_twc}. For this reason, our case requires a new approach.

\subsection{Contributions} 
In this paper, we characterize the ergodic spectral efficiency of multi-antenna $K$-tier downlink HetNets with limited feedback. We consider non-cooperative and cooperative HetNet operations. In both cases, the locations of each tier's BS are modeled as mutually independent PPPs. 

In our non-cooperative case, each BS obtains limited feedback sent from an associated single-antenna user. Based on the obtained limited CSIT, a BS uses MRT precoding to serve a single user. 
%We note that this is a typical HetNet setting that has been considered in the prior work \cite{dhil:jsac:12,hsjo:2012_twc}. 
%One way to combat the interference without coordination is maximizing the desired channel gain \cite{jh:wcl:16}. Following this, each BS uses maximum ratio transmission (MRT) precoding based on the limited feedback information. 
We assume that the same amount of feedback is used in the same network tier's BSs. We note that this assumption is only for the non-cooperative case; our cooperative case allows to use different amount of feedback to the same tier BSs. 
In this particular scenario, we derive the signal-to-interference (SIR) complementary cumulative distribution function (CCDF) and the ergodic spectral efficiency as functions of the system parameters, such as the number of antennas, the biasing factor, the transmit power, the density, and the allocated feedback of each tier. 
Leveraging the obtained expressions, we formulate an optimization problem to determine the amount of feedback used for each tier to maximize a lower bound on the sum ergodic spectral efficiency. 
%We assume that total feedback amount used in an unit area is limited. 
%The total feedback constraint is assumed in the optimization problem, which means the maximum feedback bits that can be used on average in an unit area. The total feedback constraint is limited by the uplink overheads. 
Subsequently, we propose a solution of the formulated problem. 
%The total feedback constraint is limited since the uplink overheads increase as the used feedback increases. 
%We assume that the same tier BS uses the same amount of feedback bits under a total feedback bits constraint. 
%The considered constraint limits the feedback bits averagely used in an unit area. For example, if a density of a particular tier is very high, allocating large feedback bits to that tier would take up most of the feedback resource, so that there is no room for other tiers. For this reason, the feedback bits allocation should be carefully designed depending on various system parameters including the densities. 

%occupies large portion of the total feedback constraint.
%To obtain a feedback partition criterion under this constraint, 

In our cooperative case, the BSs form a coordination set by using dynamic clustering, and mitigate the intra-cluster interference by using multi-cell ZF based on the limited feedback. 
%and the connected BSs 
%The user sends the channel feedback to the BSs in the cluster, and these BSs mitigates the interference by using multi-cell ZF based on the limited feedback. 
Dynamic clustering is applied across the tiers in the HetNet, so that a coordination set can include different tiers' BSs. Unfortunately, analyzing the performance of the considered BS coordination is not straightforward since the performance of the cluster can be different depending on the tiers of the BSs included in the cluster. For example, assuming that a cluster has $L$ BSs in a $K$-tier HetNet, there can be $K^L$ possibilities of the cluster's configuration. 
%This is because, each BS's tier can be one of the $K$ tiers, and there are $L$ BSs in the coordination set. 
For this reason, we should consider all the cases to completely characterize the performance of the $L$-size cluster. 
%When the number of tiers increases, too much analytical complexity is caused. 
%An exhaustive way is to consider all the combinations of the intra-cluster BSs' tiers, while it requires too much analytical complexity. 
To resolve this complexity, we derive a lemma showing that the intensity measure of received signal power in a HetNet can be transformed to the intensity measure of signal power in a statistically equivalent single-tier network by rescaling each tier's density. 
By exploiting this lemma, we obtain the SIR CCDF and the ergodic spectral efficiency as a function of the relative system parameters such as the cluster size, the transmit power, the biasing factor, the relative signal power of the intra-cluster BSs, and the used feedback. 
Assuming that each intra-cluster BS uses the same number of antennas, we formulate and solve an optimization problem to partition the feedback. 
%assuming that total feedback amount is given for one cluster. 
%We assume that a total feedback constraint is given for one cluster. 
%In addition, we also investigate a general antenna case where , and also a single-tier cooperative network with limited feedback. 
In addition, to overcome the restricted antenna assumption, we also investigate a general antenna case where each intra-cluster BS uses different number of antennas. As a special case, we study a single-tier cooperative network with limited feedback. 

Numerical results show the spectral efficiency gains obtained by using the proposed feedback partitions compared to the equal feedback partition. Our major findings are summarized as follows: 1) In the non-cooperative case, the feedback is only useful in a particular tier if the mean interference in the corresponding tier is small enough. 2) In the cooperative case, more feedback is allocated to the BSs whose signal powers are larger. 3) If a single-tier network is assumed, the effective cluster size increases as the square root of the total feedback. 4) In both cases, the proposed feedback partitions do not change depending on instantaneous SIR.  

The paper is organized as follows. Section II introduces the system models in non-cooperative and cooperative HetNet operation. In Section III, the performance of a HetNet in the non-cooperative case is characterized and the feedback partitions are obtained based on the ergodic spectral efficiency. In Section IV, the same task is performed in the cooperative HetNet case.
%the performance of a HetNet in the cooperative case is analyzed and the adaptive feedback .
Section V provides numerical results and Section VI concludes the paper.
% In what follows, we characterize the performance and find a lower bound on the optimum number of feedback bits for single-user MRT in Section III and for multi-user ZF in Section IV, respectively. Section V shows simulation results for verifying the obtained results and Section VI concludes the paper.

\section{System Model}

In this section, we introduce the system model assumed in the paper. We first describe the network model using stochastic geometry, and explain how the typical user associated with a BS. Then we illustrate the clustering model that uses dynamic clustering in the considered HetNet. Next, we introduce the feedback model for quantifying the channel quantization error due to limited feedback. Performance metrics are defined in the following subsection.

%the proposed BS cluster and BS cluster pattern are explained by using a toy example, and extend them to a general network model in the next subsection. Further, we introduce an algo- rithm for enhancing the proposed method, called as “edge cutting algorithm.” Signal model and performance metrics are specified in the following subsection.
%The particular model solely used in each scenario will be explained in each section respectively. 

\subsection{Network and Cell Association Model} 
We consider a $K$-tier downlink HetNet. Focusing on the $k$-th tier for $k \in \CMcal{K} = \{1,2, \cdots K\}$, BSs equipped with $N_k$ antennas are spatially distributed according to a homogeneous PPP, $\Phi_k = \left\{ {\bf{d}}_i^k \in \mathbb{R}^2, i\in \mathbb{N}\right\}$ with density $\lambda_{k}$. All the BSs in the $k$-th tier use the same transmit power $P_k$ and biasing factor $S_k$. 
Equivalently, the $k$-th tier network may be represented as a marked PPP, $\Phi_k^{\rm M} = \{{\bf{d}}_{i}^{k}, P_k, S_k, N_k, i \in \mathbb{N}\} $ with density $\lambda_{k}$ where $P_k$, $S_k$, and $N_k$ are the same marks for all the points in $\Phi_{k}$. Without loss of generality, we assume that $\left\| {\bf{d}}_i^k\right\| \le \left\| {\bf{d}}_j^k \right\|$ if $i < j$; thereby ${\bf{d}}_1^k$ indicates the nearest BS location to the origin in the $k$-th tier. 
%This is defined for later use.
Spatial locations of BSs in different tiers are assumed to be mutually independent.
Using the superposition property of independent PPPs, we compactly represent the $K$-tier HetNet as an unified marked PPP $\tilde \Phi^{\rm M} = \sum_{k \in \CMcal{K}} \Phi_k^{\rm M}$. We write $\tilde \Phi^{\rm M} = \{{\bf{d}}_i, \pi(i), P_{\pi(i)}, S_{\pi(i)}, N_{\pi(i)}, i \in \mathbb{N} \}$, where $\pi(i) \in \CMcal{K}$ is an index function indicating the tier of the corresponding point ${\bf{d}}_i$. Assuming that $\left\| {\bf{d}}_i\right\| \le \left\|{\bf{d}}_j \right\|$ if $i<j$, ${\bf{d}}_i$ means the $i$-th nearest BS location to the origin among all the tiers and $\pi(i)$ indicates that the tier of that BS. For example, assuming that the nearest BS to the origin is in the $k$-th tier, i.e., ${\bf{d}}_1 = {\bf{d}}_1^k$, then $\pi(1) = k$. 

Single-antenna users are distributed according to a homogeneous PPP, $\Phi_{\rm U} = \{{\bf{u}}_i, i \in \mathbb{N} \}$, which has density $\lambda_{\rm U} \gg \lambda_{k} $ for $k\in\CMcal{K}$. Since the user density is far larger than the BS density, we assume that there is no empty cell with high probability, so that all the cells are occupied. We note that in HetNets, the BSs can be densely deployed so that empty cells can exist, which is a topic for future work. 
We focus on the typical user located on ${\bf{u}}_1 = {\bf{0}}$ per Slivnyak's theorem \cite{baccelli:inria}.

We consider an open access policy wherein a user is able to communicate with all the BSs in any tier $k$ for $k \in \CMcal{K}$. 
%For cell association, we follow the same rule presented in \cite{hsjo:2012_twc}.
%We follow the same cell association rule in \cite{hsjo:2012_twc}. 
%Open access is assumed, i.e., a user can associate with a BS in any tier $k$ for $k \in \CMcal{K}$. 
%With the biasing factor $S_k$ for $k \in \CMcal{K}$, 
For cell association, the typical user measures the biased average received power and associates with the BS whose the measured power is maximum. For instance, the user associates with the BS located at ${\bf{d}}_{1}^{k}$ if $k = \mathop{\arg \max}_{k' \in \CMcal{K}} P_{k'} S_{k'} \left\| {\bf{d}}_1^{k'}\right\|^{-\beta} $, where $\beta$ is the path-loss exponent. Since we are interested in a HetNet sharing the spectrum among all the tiers, we assume that the path-loss exponent is same in all the tiers. Considering different path-loss exponents in each tier \cite{hsjo:2012_twc} or multi-slope path-loss model \cite{zhang:tcom:15} is future work. 

We note that biasing factor $S_k$ is mainly used for offloading in HetNets \cite{sarabjot:twc:14, hsjo:2012_twc}. For example, as $S_k$ increases, the number of users associated with the $k$-th tier BS also increases, which relieves the number of users associated with the other tiers. This allows other tiers to allocate more resources per one user. Typically, a small network tier such as femto tends to have large biasing factor to save the resources of the macro tier. Jointly considering feedback design and offloading will be interesting future work.

\subsection{Clustering Model}
%We apply a dynamic BS coordination method. 
Dynamic BS coordination is used to form a BS cluster. With the cluster size $L$, the typical user connects to the $L$ BSs that provides $L$ strongest biased average received power. Denoting the BS coordination set $\CMcal{C} = \{i_1,...,i_L\}$, we have
\begin{align} \label{eq:clustering}
P_{\pi(i_1)} S_{\pi(i_1)} \left\|  {\bf{d}}_{i_1}\right\|^{-\beta} \ge P_{\pi(i_2)} S_{\pi(i_2)} \left\|  {\bf{d}}_{i_2}\right\|^{-\beta} \ge ... \ge P_{\pi(i_L)} S_{\pi(i_L)} \left\|  {\bf{d}}_{i_L}\right\|^{-\beta},
\end{align}
where $P_{\pi(i_L)} S_{\pi(i_L)} \left\|  {\bf{d}}_{i_L}\right\|^{-\beta} \ge P_{\pi(j)} S_{\pi(j)} \left\|  {\bf{d}}_{j}\right\|^{-\beta}$ for all $j \in \mathbb{N} \backslash \CMcal{C}$.
% and ${\bf{d}}_j \in \tilde \Phi^{\rm M}$ for $j \in \mathbb{N}$. 
%We assume that the path-loss exponent is same with all the tiers. 
%By the cell association rule considered in this paper, the typical user associates with the BS located at ${\bf{d}}_{i_1}$ and 
According to the association rule, the typical user associates with the BS located at ${\bf{d}}_{i_1}$ and receives the desired signal from it.
%the BS located at ${\bf{d}}_{i_1}$.
%The cluster size $L$ can be made as large as possible that $\mathop {\min}_{i_{\ell} \in \CMcal{C}} N_{\pi(i_{\ell})}$. 
%As a special case, assuming $L = 1$, the considered network becomes the conventional $K$-tier heterogeneous network as in \cite{hsjo:2012_twc, dhil:jsac:12}. 
We note that $L=1$ indicates the non-cooperative case, and $L\ge 2$ is the cooperative-case. To mitigate the intra-cluster interference by using multi-cell ZF in the cooperative case, we assume $L \le \min_{k \in \CMcal{K}} N_k$. 

Using the described dynamic clustering, a cooperative region is mathematically defined by using the notion of the $L$-th order weighted Voronoi region, which is an extended version of the typical Voronoi region [17], [18]. For example, the weighted Voronoi region corresponding to the coordination set $\CMcal{C}= \{i_1,...,i_L\}$ is defined as 
\begin{align} \label{eq:def_voro}
&{\mathcal{V}}^{\rm w}_L({\bf{d}}_{i_1},...,{\bf{d}}_{i_L}) \nonumber \\
&= \left\{ {\bf{d}} \in \mathbb{R}^2 | \cap_{\ell=1}^{L} \left\{ (P_{\pi(i_\ell)} S_{\pi(i_\ell)})^{-\frac{1}{\beta}} \left\| {\bf{d}} - {\bf{d}}_{i_\ell} \right\| < (P_{\pi (j)} S_{\pi (j)})^{-\frac{1}{\beta}}\left\| {\bf{d}} - {\bf{d}}_j\right\| \right\}, j \notin \{i_1,...,i_L\}  \right\}.
\end{align} 
The users located in ${\mathcal{V}}^{\rm w}_L({\bf{d}}_{i_1},...,{\bf{d}}_{i_L})$ are connected to the coordination set $\CMcal{C}$. Naturally, the typical user is also located in ${\mathcal{V}}^{\rm w}_L({\bf{d}}_{i_1},...,{\bf{d}}_{i_L})$, i.e., ${\bf{o}} \in {\mathcal{V}}^{\rm w}_L({\bf{d}}_{i_1},...,{\bf{d}}_{i_L})$. 
By allocating the orthogonal time-frequency resources to adjoint Voronoi regions, a conflict between any two different clusters can be prevented so that each cluster can serve the connected users simultaneously. Optimizing the resources allocated to each Voronoi region is a challenging yet important problem, and will be interesting future work. We note that in a simple case $K=1$ and $L = 2$, this problem can be solved by using cooperative base station coloring \cite{park:tcom:16}.  

\subsection{Signal Model}
%As described in the above subsection, there are $K$ users in the BS cluster $\CMcal{A}$.
%the BS $k$ sends an information symbol $s_k$ to its own user.
%Each BS at ${\bf{d}}_k$ for $k \in \{1,...,K\}$, sends an information symbol $s_k$ for each user. 
%The BS at ${\bf{d}}_1$ sends an information symbol $s_1$ to the typical user.
%, while the BS at ${\bf{d}}_2 \sim {\bf{d}}_K$ nullifies the signal to the typical user. 
%The information symbol for the typical user is, therefore, $s_1$.
We assume a synchronous narrowband signal model. In this setting, each BS encodes the information symbol separately and sends it to the associated user. When transmitting the symbol, the BS at ${\bf{d}}_i$ uses a linear beamforming vector ${\bf{v}}_{i} \in \mathbb{C}^{N_{\pi(i)}}$ and $\left\| {\bf{v}}_i \right\| = 1$. 
Assuming that the typical user is associated with a BS located at ${\bf{d}}_{i_1} \in \tilde \Phi^{\rm M}$, the received signal at the typical user is
%\begin{align}
%y_1 &= P_{\pi(i_1)}\left\|{\bf{d}}_{i_1}^{k} \right\|^{-\beta/2} ({\bf{h}}_{1,1}^{k})^{*} {\bf{v}}_1^{k} s_1^{k} + \sum_{k=1}^{K} \sum_{i \in \Phi_k \backslash {\bf{d}}_1^{k}}  P_{k} \left\| {\bf{d}}_i^{k} \right\|^{-\beta/2} ({\bf{h}}_{1,i}^{k})^*{\bf{v}}_i^{k} s_i^{k} + z_1^{k},
%\end{align}
\begin{align}
y_1 = & \sqrt{P_{\pi(i_1)}}\left\|{\bf{d}}_{i_1} \right\|^{-\beta/2} ({\bf{h}}_{1,i_1})^{*} {\bf{v}}_{i_1} s_{i_1} + 
\sum_{i_{\ell} \in \CMcal{C} \backslash i_1} \sqrt{P_{\pi(i_{\ell})}}\left\|{\bf{d}}_{i_{\ell}} \right\|^{-\beta/2} ({\bf{h}}_{1,i_{\ell}})^{*} {\bf{v}}_{i_{\ell}} s_{i_{\ell}}  \nonumber \\
&+ \sum_{j \in \mathbb{N} \backslash \CMcal{C}} \sqrt{P_{\pi(j)}}\left\|{\bf{d}}_{j} \right\|^{-\beta/2} ({\bf{h}}_{1,j})^{*} {\bf{v}}_{j} s_{j}+ z_1,
\end{align}
where ${\bf{h}}_{1, i_{}} \in \mathbb{C}^{N_{\pi(i)}}$ is the downlink channel coefficient vector from the BS at ${\bf{d}}_i \in \tilde \Phi^{\rm M}$ to the typical user, $\beta>2$ is the path-loss exponent, and $z_1 \sim \CMcal{CN}\left(0,\sigma^2\right)$ is additive Gaussian noise.
The symbol energy is normalized as $\mathbb{E}\left[\left|s_i \right|^2 \right] =1$ for $i \in \mathbb{N}$. Each entry of the channel coefficient vector ${\bf{h}}_{1, i}$ is drawn from independent and identically distributed (IID) complex Gaussian random variables $\CMcal{CN}\left(0,1\right)$ indicating Rayleigh fading.
% $I_{\rm In}$ denotes the intra-cluster interference, defined as
%\begin{align}
%I_{\rm In} = 
%\end{align}
The beamforming vector ${\bf{v}}_i$ is designed based on the CSIT obtained by using limited feedback. 
%The details of the feedback model are explai ned in the next subsection. 

\subsection{Feedback Model}
%Each user estimates the corresponding channel vectors, and sends them back to each of the corresponding BS. 
We explain the feedback process focusing on the typical user. This process is applied to other users equivalently. 
Let us assume that the typical user feeds back the channel information to a BS located at ${\bf{d}}_i$.  
First, the typical user estimates the channel coefficient vector ${\bf{h}}_{1, i}$. 
% by using predefined pilot symbols sent from the BSs, and sends it back to the associated BS via a finite rate feedback link. 
%Since a feedback link has generally finite capacity, the typical user quantizes the es
%To do this, the typical user quantizes the obtained CSI.
To focus on the effect of limited feedback, we assume that the channel estimation is perfect.
%As an example, we assume that the typical user feedback the channel ${\bf{h}}_{1, i_1}$ to the BS at ${\bf{d}}_{i_1}$. 
Once the typical user learns the channel coefficient vector ${\bf{h}}_{1, i}$, it quantizes 
the channel direction information $\tilde {\bf{ h}}_{1, i} = {\bf{h}}_{1, i}/\left\| {\bf{h}}_{1, i} \right\| $ by using a predefined quantization codebook $\CMcal{Q}$. 
The codebook $\CMcal{Q}$ is shared with the BS at ${\bf{d}}_{i}$ and the typical user. 
Assuming that $B$ bits are used for quantizing $\tilde {\bf{ h}}_{1, i}$, the codebook $\CMcal{Q}$ is constructed as $\CMcal{Q} = \left\{{\bf{w}}_1,...,{\bf{w}}_{2^{B}} \right\}$, where each codeword ${\bf{w}}_j$ is a $N_{\pi(i)}$-dimensional unit norm vector, i.e., $\left\| {\bf{w}}_j \right\| = 1$ for $j \in \left\{1,...,2^{B} \right\}$. 
Then, the codeword that has maximum inner product with $\tilde {\bf{ h}}_{1, i}$ is selected, namely
%By measuring the inner products between the channel direction vector $\tilde {\bf{h}}_k = {\bf{h}}_k / \left\| {\bf{h}}_k \right\|$ and the codeword vectors ${\bf w}_i$ for $i \in \{1,...,2^B\}$, a user chooses an index that provides the maximum inner product value, namely, 
\begin{align}
j_{\max} = \mathop {\arg \max} \limits_{j = 1,...,2^{B}} \left|  (\tilde {\bf{h}}_{1,i} )^* {\bf{w}}_j \right|.
\end{align}
The chosen index $j_{\max}$ is sent to the BS at ${\bf{d}}_{i}$ and 
%By using the shared channel codebook $\CMcal{C}_k$, t
the BS recovers the quantized channel direction information from this index. 
We denote the quantized channel direction as $\hat {\bf{h}}_{1,i} = {\bf{w}}_{j_{\max}}$.
%The total number of feedback bits is denoted as $B_{\rm total}$, so that $B_{\rm total} = \sum_{k=2}^{K} B_k$.
%Since the typical user uses $B_k$ for the BS at ${\bf{d}}_k$ and there are $K$ BSs in the cluster $\CMcal{A}$, the total bits used in the channel feedback is 
%Since the BS has the same codebook as user $k$, it acquires quantized channel direction information $\hat {\bf{h}}_k = {\bf{w}}_{i_{\max}}$ from $i_{\max} $. The quantized channel information $\hat {\bf{h}}_k$ is used for designing a precoding matrix ${\bf{V}}$.

%One should note that only channel direction information is delivered to the BS by this feedback method, but not channel quality information. Channel quality information is useful for power allocation \cite{5089951} or user scheduling \cite{4299617, 1386524}, which is beyond our scope. (\textbf{Do we need CQI ?? for this paper})

%For tractability of our analysis, we assume a near-optimal codebook \cite{1512149}. 

For analytical tractability, we adopt the quantization cell approximation \cite{muk:tit:03, shen:twc:05,  jh:twc:16} instead of using a specific limited feedback strategy. This approximates each quantization cell as a Voronoi region of a spherical cap \cite{gers:tit:79}. 
%This is a standard approach in vector quantization \cite{1056067, gersho:book} and it is used to deal with the irregular shape of the Voronoi quantization regions. 
%quantization cell upper bound (QUB).
%For convenience, we call this technique as the spherical-cap approximation of vector quantization (SCVQ) \cite{jh:twc:16}.
In this technique, assuming that $B$-bits feedback is used, the area of the quantization cell is $2^{-B}$ and this leads to an expression of the CDF of quantization error 
\begin{align} \label{def:q_error_cdf}
F_{{\rm sin^2}\theta_i}\left(x\right) = \left\{\begin{array}{cc}2^{B} x^{N_{\pi(i)}-1}, &  0\le x \le \delta \\ 1, & \delta \le x \end{array}, \right.
\end{align}
where ${\rm sin^2}\theta_i = 1 - \left| (\tilde {\bf{h}}_{1,i})^* \hat{\bf{h}}_{1,i} \right|^2$ and $\delta = 2^{-\frac{B}{N_{\pi(i)}-1}}$. 
%In \cite{4299617}, for any quantization codebook that has a quantization error CDF $F_{{\rm sin^2}\tilde \theta_k}\left(x\right)$, we have $F_{{\rm sin^2}\theta_k}\left(x\right) \ge F_{{\rm sin^2}\tilde \theta_k}\left(x\right)$. 
%Based on this characteristic, in this paper, we call the used approximation technique as quantization cell upper bound (QUB) quoted from \eqref{4299617} for convenience.
%As the SCVQ gives fewer quantization errors than those by conventional codebook construction methods, 
%Due to this property, SCVQ provides an upper bound performance with limited CSI feedback.
%Based on this characteristic of the used approximation technique, 
In isotropic channel distribution, this approximation technique provides an upper bound of the quantization performance, while the gap to a lower bound provided by random vector quantization is reasonably small \cite{jh:twc:16}. 

%Comparing to random vector quantization (RVQ), which is another analytical approach in limited feedback, RVQ provides a lower bound on the performance of limited feedback 

%Since random vector quantization (RVQ), which is another analytical approach in limited feedback
%The performance of SCVQ and RVQ represent two extreme cases, i.e., an upper bound and a lower bound on the performance of limited feedback, respectively.

\subsection{Performance Metrics}
Since cellular systems are usually interference limited \cite{dhil:jsac:12}, we focus on the SIR. We consider two cases depending on the coordination set size $L$.

{\textbf{Non-cooperative case}}: In this case with $L = 1$, there is no coordination among the BSs. The BS uses single-user MRT precoding based on the quantized channel direction, ${\bf{v}}_{i_1} = \hat {\bf{h}}_{1, i_1}$. 
%To do this, the typical user feeds back the channel ${\bf{h}}_{1, i_1}$ to the BS at ${\bf{d}}_{i_1}$. 
We assume that all the users in the $k$-th tier use the $B_k$ bits feedback, so that the total feedback used on average in an unit area is $B_{\rm total} = \sum_{k=1}^{K} \lambda_k B_k$. 
%This is because there are $\lambda_k$ users in an unit area. 
The instantaneous SIR of the typical user is written by
\begin{align}
{\rm SIR}_{\rm NC} = \frac{P_{\pi(i_1)}\left\| {\bf{d}}_{i_1} \right\|^{-\beta} \left|\left({\bf{h}}_{1,i_1} \right)^* \hat {\bf{h}}_{1,i_1}\right|^2}{\sum_{j \in \mathbb{N} \backslash i_1 } P_{\pi(j)}\left\|{\bf{d}}_{j} \right\|^{-\beta} \left| ({\bf{h}}_{1,j})^{*} {\bf{v}}_{j} \right|^2 }.
\end{align}
%Assuming that the typical user is associated with a BS in the $k$-th tier, the conditioned SIR is 
%\begin{align} \label{eq:sir_cond}
%{\rm SIR}_{{\rm NC}}^{k} = \frac{P_{k}\left\| {\bf{d}}_1^{k} \right\|^{-\beta} \left|\left({\bf{h}}_{1,1}^{k} \right)^* \hat {\bf{h}}_{1,1}^{k}\right|^2 }{\sum_{i=1}^{K} \sum_{{\bf{d}}_j^i \in \Phi_i \backslash {\bf{d}}_1^{k}}   P_{i} \left\| {\bf{d}}_j^{i} \right\|^{-\beta} \left| ({\bf{h}}_{1,j}^{i})^*{\bf{v}}_j^{i} \right|^2 },
%\end{align}
%where $\Phi_{i} \backslash {\bf{d}}_1^k$ means a set of the locations of BSs included in the $i$-th tier conditioned on that the typical user is associated with a $k$-th tier BS. 
%where $\Phi_{i} \backslash {\bf{d}}_1^k$ incorporates the exclusion region, i.e., $\Phi_{i} \backslash {\bf{d}}_1^k = \Phi_{i} \backslash \CMcal{B}(0,\left(\frac{P_i}{P_k} \right)^{\frac{1}{\beta}} \left\| {\bf{d}}_1^k \right\| ) $.
%where $\CMcal{B}(0,R )$ is a ball whose a center is on the origin and a radius is $R$.
%Recalling that a $k$-th tier BS has $N_k$ transmit antennas, we assume that $B_k$ feedback bits is used for the $k$-th tier and the BSs in the same tier uses the same feedback bits. 
%Given the typical user is associated with a $k$-th tier BS, 
Conditioning on that ${\bf{d}}_{i_1} = {\bf{d}}_1^k$ (or $\pi(i_1) = k$), i.e., the typical user is associated with a BS in the $k$-th tier, we denote the conditioned SIR as ${\rm SIR}_{{\rm NC}|k}$.  The CCDF of the conditioned SIR is defined as
\begin{align} \label{eq:sir_ccdf_het}
&F^{\rm c}_{{\rm SIR}_{\rm NC|k}}\left(\beta, \bar \lambda_K, \bar N_K, \bar B_{{K}}, \bar P_K, \bar S_K;\gamma \right) = \mathbb{P}\left[{\rm SIR}_{\rm NC|k} \ge \gamma \right],
%\nonumber \\
%&= \mathbb{P}\left[\frac{P_{k}\left\| {\bf{d}}_1^{k} \right\|^{-\beta} \left|\left({\bf{h}}_{1,1}^{k} \right)^* \hat {\bf{h}}_{1,1}^{k}\right|^2 }{\sum_{i=1}^{K} \sum_{{\bf{d}}_j^i \in \Phi_i \backslash {\bf{d}}_1^{k}}   P_{i} \left\| {\bf{d}}_j^{i} \right\|^{-\beta} \left| ({\bf{h}}_{1,j}^{i})^*{\bf{v}}_j^{i} \right|^2 } \ge \gamma \right],
\end{align}
with a set of the densities: $\bar \lambda_K = \{\lambda_1, ..., \lambda_K\}$, a set of the antennas: $\bar N_K = \{N_1, ..., N_K \}$, a set of the feedback: $\bar B_K = \{B_1,..., B_K\}$, a set of the transmit power: $\bar P_K = \{P_1, ..., P_K\}$, and a set of the biasing factor: $\bar S_k \ \{S_1,...,S_K \}$. 
The conditioned ergodic spectral efficiency is defined as
\begin{align} \label{eq:rate_het}
&R_{\rm NC|k}\left(\beta, \bar \lambda_K, \bar N_K, \bar B_K, \bar P_K, \bar S_K \right) = \mathbb{E}\left[\log_2\left(1+ {\rm{SIR}}_{\rm NC|k} \right) \right]. 
%\nonumber \\
%&= \mathbb{E}\left[\log_2\left(1+ \frac{P_{k}\left\| {\bf{d}}_1^{k} \right\|^{-\beta} \left|\left({\bf{h}}_{1,1}^{k} \right)^* \hat {\bf{h}}_{1,1}^{k}\right|^2 }{\sum_{i=1}^{K} \sum_{{\bf{d}}_j^i \in \Phi_i \backslash {\bf{d}}_1^{k}}   P_{i} \left\| {\bf{d}}_j^{i} \right\|^{-\beta} \left| ({\bf{h}}_{1,j}^{i})^*{\bf{v}}_j^{i} \right|^2 }\right) \right].
\end{align}
%Since $R^{k}_{\rm NC}(\beta, \bar N_K, \bar B_K, \bar P_K,\bar S_K)$ is the ergodic spectral efficiency conditioned on that the typical user is associated with a $k$-th tier BS, 
%Since the density of the $k$-th tier BS is $\lambda_k$, the sum spectral efficiency per unit area defined as a weighted sum of the conditioned ergodic spectral efficiency, i.e., 
%\begin{align} \label{eq:area_se}
%R_{\Sigma} = \sum_{k=1}^{K} \lambda_{k} R_{\rm NC}^k\left(\beta, \bar N_K, \bar B_K, \bar P_K,\bar S_K \right),
%\end{align}
%One noticeable point is that the definition of the sum spectral efficiency \eqref{eq:area_se} is different from the average spectral efficiency in \cite{hsjo:2012_twc}. The difference is coming from that \cite{hsjo:2012_twc} only focused on the typical user, not a whole network snapshot. Specifically, denoting the probability that the typical user is associated with a $k$-th tier BS as $\CMcal{P}_{k}$, the average spectral efficiency of the typical user is $R_{\rm avg.} = \sum_{k=1}^{K} \CMcal{P}_{k} R_{\rm NC}^k$. This definition, however, is not suited in this paper since we consider an adaptive feedback partition where different amount of feedback bits is used in different tiers. This requires to capture a whole network snapshot including other users associated with different tiers, therefore \eqref{eq:area_se} is more appropriate to our purpose.

{\textbf{Cooperative case}}: In this case with $L \ge 2$, the BSs in $\CMcal{C}$ are coordinated. The beamforming vector is designed as multi-cell ZF to mitigate the intra-cluster interference. Specifically, ${\bf{v}}_{i_1}$ satisfies
\begin{align} \label{eq:zf_criterion}
(\hat {\bf{h}}_{\ell, i_{1}})^*{\bf{v}}_{i_1} = 0, {\ell} \in \CMcal{C} \backslash 1,
\end{align}
where $\bar {\bf{h}}_{\ell, i_1}$ is the quantized channel coefficient vector from the BS at ${\bf{d}}_{i_1}$ to the user $\ell$ associated with the BS ${\bf{d}}_{i_{\ell}}$.
The solution of \eqref{eq:zf_criterion} always exists if $L \le \mathop {\min}_{i_{\ell} \in \CMcal{C}} N_{\pi(i_{\ell})}$. 
%To do this, the typical user sends the $B_{i_{\ell}}$ bits feedback to the intra-cluster BS at ${\bf{d}}_{i_{\ell}}$ for $i_{\ell} \in \CMcal{C}$. 
We denote that the feedback used for $i_{\ell}$-th BS in the coordination set as $B_{i_{\ell}}$. Since the feedback information is only used for managing the intra-cluster interference, the typical user does not send the feedback to its associated BS, i.e., $B_{i_1} = 0$. The total feedback used in one coordination set is $B_{\rm total} = \sum_{{\ell} = 2}^{L} B_{i_{\ell}}$. 
Note that we slightly abuse the notation of the feedback in the non-cooperative case and the cooperative case. Specifically, in the non-cooperative case, $B_k$ means the feedback used in the $k$-th tier BS. In the cooperative case, $B_{i_{\ell}}$ means the feedback used in the $i_{\ell}$-th BS in the coordination set. 
For analytical simplicity, we assume that all the BSs in the same coordination set $\CMcal{C}$ use only $L$ antennas for multi-cell ZF, so that effectively the typical user has $L$-dimensional channel to each intra-cluster BS. 
Since we only use a part of the antennas, our analysis may indicate a lower bound on the spectral efficiency that can be achieved by using the full antennas. Using the full antennas, the BSs in the coordination set can mitigate the intra-cluster interference and also increase the desired signal power by coordinated beamforming. This case will be further investigated later. 
% XXX
%Analyzing the performance of limited feedback in such a case is interesting future work. 
%Due to the inaccurate channel feedback. i.e., ${\bf{h}}_{1, i_{\ell}} \neq \hat {\bf{h}}_{1, i_{\ell}}$, 

Due to the inaccurate channel feedback, the intra-cluster interference is not perfectly nullified. Considering the remaining intra-cluster interference, the instantaneous SIR is 
\begin{align} \label{eq:sir_coop}
{\rm SIR}_{\rm C} = \frac{ P_{\pi(i_1)}\left\| {\bf{d}}_{i_1} \right\|^{-\beta} \left| ({\bf{h}}_{1,i_1})^* {\bf{v}}_{i_1}\right|^2}{I_{\rm In} + I_{\rm Out}},
\end{align}
where $I_{\rm In} = \sum_{i_{\ell} \in \CMcal{C} \backslash i_1} P_{\pi(i_{\ell})} \left\|{\bf{d}}_{i_{\ell}} \right\|^{-\beta} \left| ({\bf{h}}_{1,i_{\ell}})^{*} {\bf{v}}_{i_{\ell}} \right|^2$, 
$I_{\rm Out} =\sum_{j \in \mathbb{N} \backslash \CMcal{C}} P_{\pi(j)} \left\|{\bf{d}}_{j} \right\|^{-\beta} \left| ({\bf{h}}_{1,j})^{*} {\bf{v}}_{j}\right|^2 $, each of which indicates the remaining intra-cluster interference and the out-of-cluster interference, respectively.
%and $\bar B_{ L} = \left\{B_2, ..., B_{L} \right\}$. $I_{\rm In}$ denotes the intra-cluster interference and $I_{\rm Out}$ denotes the out-of-cluster interference.
%To provide an adaptive feedback criterion depending on the intra-cluster geometry, 
%To minimize the amount of $I_{\rm In}$, it is efficient to design feedback as a function of the intra-cluster geometry. 
%Later, we also provide a feedback criterion that can be used if the intra-cluster geometry is not given. 
Similar to the non-cooperative case, we denote ${\rm SIR}_{\rm C|m}$ as the instantaneous SIR conditioned on that the typical user is associated with a BS in the $m$-th tier, i.e., $\pi(i_1) = m$. 
%Note that we leave the notation $k$ for later use.
The CCDF of the conditioned SIR is defined as
\begin{align} \label{def:sir_ccdf_coop}
&F^{\rm c}_{{\rm SIR}_{\rm C|m}}\left(\beta, \bar \lambda_K, \bar N_K, \bar B_{{L}}, \bar P_K, \bar S_K;\gamma \right) = \mathbb{P}\left[{\rm SIR}_{\rm C|m} \ge \gamma \right], 
%\nonumber \\
%&= \mathbb{P}\left[\frac{\left| {\bf{h}}_{1,1}^* {\bf{v}}_1\right|^2}{I_{\rm In}(\bar \delta_{1,  L}, \bar B_{L}) + I_{\rm Out}} \ge \gamma \right].
\end{align}
%where the number of antenna $N$ is omitted since it is equal to $K$.
where a set of the feedback $\bar B_{L} = \{B_{i_2}, ..., B_{i_L}\}$. The ergodic spectral efficiency is defined as
\begin{align} \label{def:ergodic_rate_coop}
&R_{\rm C|m}\left(\beta, \bar \lambda_K, \bar N_K, \bar B_{L}, \bar P_K, \bar S_K \right) = \mathbb{E}\left[\log_2\left(1+ {\rm{SIR}}_{\rm C|m} \right) \right]. 
\end{align}

We clarify the difference in the total feedback constraint between the non-cooperative and the cooperative cases. In the non-cooperative case, the total feedback used on average in an unit area is limited, so that $B_{\rm total} = \sum_{k=1}^{K} \lambda_k B_k$. In the cooperative case, the total feedback used in one BS coordination set is fixed, so that $B_{\rm total} = \sum_{{\ell} = 2}^{L} B_{i_{\ell}}$, where $B_{i_{\ell}}$ is the feedback allocated to the BS at  ${\bf{d}}_{i_{\ell}}$. 
%Note that we slightly abuse the notation 
The rationale of this difference is as follows. In the non-cooperative case, the typical user sends the feedback to its associated BS only, so that the density directly affects the totally feedback use in an unit area. Specifically, a dense network tier consumes more feedback in an unit area, therefore it is reasonable that the total feedback constraint is a weighted sum of the used feedback in each tier. 
%This makes the densities be incorporated into the total feedback constraint. 
In the cooperative case, the typical user user sends the feedback to all the intra-cluster BSs but not its associated BS. Further, the intra-cluster BSs' tiers can be different depending on the condition of the coordination set. 
%For this reason, it cannot be guaranteed that a dense network tier consumes more feedback if a coordination set includes not dense network tier 
For this reason, if a total feedback constraint is defined as a function of the density as in the non-cooperative case, each coordination set has a different total feedback constraint. This makes it difficult to provide a general feedback partition solution, which is applicable in the cooperative case regardless of the cluster members' tiers. As a result, it is more convenient to set the total feedback constraint for each coordination set. 
%We can incorporate by adjusting a total feedback constraint in the cooperative case.
%each tier's condition is interwoven in the one BS coordination set depending on the member of the coordination set, therefore it is not easy to manage the feedback resource For this reason, 

\section{Non-Cooperative Case}
In this section, we characterize the performance of a non-cooperative HetNet and formulate a feedback partition problem based on the performance characterization. Subsequently, we propose a solution for the problem to determine $B_k$. 
%In the second part of the paper, we consider a heterogeneous network scenario, where each different type of a BS serves a single user by using single-user MRT. As in the previous case, the required CSIT is obtained by limited feedback. We first explain the models and analyze the performance. Leveraging this, we propose a feedback partition strategy.
%In a heterogeneous network scenario, $K$ tier BS serve by using MRT. 

\subsection{Performance Characterization}
In this subsection, we analyze the SIR CCDF and the ergodic spectral efficiency. To this end, we first introduce the following lemma. 

\begin{lemma} \label{lem:associate_pdf_hetnet}
If the typical user is associated with the BS belonging to the $k$-th tier, the PDF of the distance between the typical user and the serving BS, i.e., $\left\| {\bf{d}}_1^{k} \right\|$, is 
\begin{align} \label{eq:pdf_first_touch}
f_{\left\| {\bf{d}}_1^{k}\right\|}(r) = 2\pi \sum_{i=1}^{K} \lambda_i \left(\frac{P_i S_i}{P_kS_k} \right)^{2/\beta} r \exp\left({-\pi \sum_{i=1}^{K} \lambda_i \left(\frac{P_{i}S_i}{P_{k} S_k} \right)^{2/\beta}} r^{2}\right),
\end{align}
\end{lemma}
\begin{proof}
See Lemma 3 in \cite{hsjo:2012_twc}.
\end{proof}
By leveraging Lemma \ref{lem:associate_pdf_hetnet}, we derive the SIR CCDF in the following theorem. 
\begin{theorem} \label{theo:sir_ccdf_hetnet}
When the typical user is associated with the $k$-th tier BS, the SIR CCDF is 
\begin{align} \label{eq:sir_ccdf_hetnet}
&F^{\rm c}_{{\rm SIR}_{\rm NC|k}}\left(\beta, \bar \lambda_K, \bar N_K, \bar B_{{K}},\bar P_K, \bar S_K;\gamma \right) =\sum_{m=0}^{N_{k}-1} \frac{\gamma^m}{m!} (-1)^m \left. \frac{\partial^m \CMcal{L}_{I/\cos^2\theta_1}(s)}{\partial s^m} \right| _{s = {\gamma}},
\end{align}
%where $I = \sum_{i=1}^{K}\sum_{{\bf{d}}_j^i \in \Phi_i \backslash {\bf{d}}_1^{k}}   \frac{P_{i}}{P_k} \left\|{\bf{d}}_1^k \right\|^{\beta} \left\| {\bf{d}}_j^{i} \right\|^{-\beta} \left| ({\bf{h}}_{1,j}^{i})^*{\bf{v}}_j^{i} \right|^2 $ and $\CMcal{L}_{I/\cos^2\theta}(s)$ is the Laplace transform of $I/\cos^2\theta_1$ obtained as 
where $\CMcal{L}_{I/\cos^2\theta}(s)$ is 
%the Laplace transform of $I/\cos^2\theta_1$ obtained as 
\begin{align} \label{eq:laplace_noncoopt_intheorem}
\CMcal{L}_{I/\cos^2\theta_1}(s) =  \int_{0}^{2^{-\frac{B_{k}}{N_{k}-1}}}  \frac{{2^{B_k} (N_k-1)x^{N_k-2}}   \sum_{i=1}^{K} \lambda_i \left(\frac{P_i S_i}{P_k S_k} \right)^{2/\beta} }{\sum_{i=1}^{K}\lambda_i \left(\frac{P_i S_i}{P_k S_k} \right)^{2/\beta} \left[1+\CMcal{D}(\frac{s}{1-x}\left(\frac{S_k}{S_i} \right), \beta) \right]} {\rm d} x,
\end{align}
with
\begin{align} \label{eq:het:dfunc}
\CMcal{D}( x, y) = \frac{2x}{y-2} {}_2F_1\left(1, 1-\frac{2}{y}, 2-\frac{2}{y}, -x \right).
\end{align}
\end{theorem}
\begin{proof}
See Appendix \ref{appen:theo1}.
\end{proof}

\begin{remark} \normalfont
Although we assume an open access policy in this paper, it is also possible to use a closed access policy. In this remark, we study how the SIR CCDF \eqref{eq:sir_ccdf_hetnet} changes when the $k$-th tier BSs use a closed access policy. Specifically, we assume the probability of a $k$-th tier BS being open access as $p_k < 1$, so that the density of the accessible $k$-th tier BSs  is $p_k \lambda_k$. For convenience of expression, we define $p_1=...=p_{k-1}= p_{k+1}=...=p_K = 1$. Then, the PDF of the distance to the serving BS \eqref{eq:pdf_first_touch} is modified as 
\begin{align} \label{eq:closed_asso_dist}
f_{\left\| {\bf{d}}_1^{k}\right\|, {\rm closed}}(r) = 2\pi \sum_{i=1}^{K} p_i\lambda_i \left(\frac{P_i S_i}{P_kS_k} \right)^{2/\beta} r \exp\left({-\pi \sum_{i=1}^{K} p_i\lambda_i \left(\frac{P_{i}S_i}{P_{k} S_k} \right)^{2/\beta}} r^{2}\right).
\end{align}
Next, we calculate the SIR CCDF in the closed access case. The Laplace transform of the interference $\CMcal{L}_{I/\cos^2(\theta_1)}(s) $ in \eqref{eq:laplace_noncoopt_intheorem} is modified as
\begin{align} \label{eq:laplace_close}
&\CMcal{L}_{I/\cos^2\theta_1, {{\rm closed}}}(s) \nonumber \\
& =  \mathbb{E}_{R, \cos^2\theta_1} \Bigg[  \prod_{i=1}^{K} \exp\left(-\pi p_i \lambda_i \left(\frac{P_i S_i}{P_{k} S_k} \right)^{2/\beta}R^{2} \CMcal{D}\left(\frac{s}{\cos^2\theta_1} \left( \frac{S_k}{S_i}\right) , \beta\right) \right)   \nonumber \\
& \;\;\;\;\;\;\;\;\;\;\;\;\;\;\;\;\;\;\; \cdot \exp\left( -\pi (1-p_k) \lambda_k R^2 \left(\frac{s}{\cos^2(\theta_1)}  \right)^{2/\beta} \frac{2\pi}{\beta} \csc\left(\frac{2\pi}{\beta} \right) \right) \Bigg] \nonumber \\
%&=\mathbb{E}_{R, \cos^2\theta_1}\left[ \exp\left(-\pi \sum_{i=1}^{K}\lambda_i \left(\frac{P_i S_i}{P_{k} S_k} \right)^{2/\beta}R^{2} \CMcal{D}\left(\frac{s}{\cos^2\theta_1}\left( \frac{S_k}{S_i}\right) , \beta\right) \right) \right] \nonumber \\
&= \int_{0}^{2^{-\frac{B_{k}}{N_{k}-1}}}  2^{B_k} (N_k-1)x^{N_k-2}  \int_{0}^{\infty}  f_{\left\| {\bf{d}}_1^{k}\right\|}(r)  \cdot \nonumber \\
& \;\;\;\;\;\;\;\;\;\;\;\;\;\;\;\;\;\; \exp\left(-\pi \sum_{i=1}^{K}p_i \lambda_i \left(\frac{P_i S_i}{P_{k}S_k} \right)^{\frac{2}{\beta}} \!\!\! r^{2} \CMcal{D}\left(\frac{s}{1-x}\left(\frac{S_k}{S_i} \right) , \beta\right) \right)\cdot \nonumber \\
& \;\;\;\;\;\;\;\;\;\;\;\;\;\;\;\;\;\; \exp\left( -\pi (1-p_k) \lambda_k r^2 \left(\frac{s}{1-x}  \right)^{2/\beta} \frac{2\pi}{\beta} \csc\left(\frac{2\pi}{\beta} \right) \right) {\rm d} r  {\rm d} x .
\end{align}
Plugging \eqref{eq:laplace_close} into \eqref{eq:sir_ccdf_hetnet} produces the SIR CCDF in the closed access case.
%\begin{align} \label{eq:rate_het_close}
%&R_{\rm NC, close|k}\left(\beta, \bar \lambda_K,\bar N_K, \bar B_K, \bar P_K,\bar S_K \right) =  \log_2(e) \int_{0}^{\infty} \frac{1}{z}   \left( 1 - \left(\frac{1}{z+1} \right) \left(\frac{1}{1+z \left(1 - 2^{-\frac{B_k}{N_k-1}} \right)} \right)^{N_k - 1}\right) \cdot\nonumber \\
%&\;\;\;\;\;\;\;\;\;\;\;\;\;\;\;\;\;\;\;\;\;\;\;\;\;\;\;\;\;\;\;\;\;\;\;\;\;\;\;\;\;\;    \left(\frac{\sum_{i=1}^{K} \lambda_i \left(\frac{P_i S_i}{P_k S_k} \right)^{2/\beta} }{\sum_{i=1}^{K}\lambda_i \left(\frac{P_i S_i}{P_k S_k} \right)^{2/\beta} \left[1+\CMcal{D}(z\left(\frac{S_k}{S_i} \right), \beta) \right]} \right){\rm d} z,
%\end{align}
\end{remark}

Next, we obtain the ergodic spectral efficiency. 
Before proceeding further, we introduce two useful lemmas. The first one is for characterizing the Laplace transform of the desired signal power as a function of the used feedback; the second one is for obtaining the ergodic spectral efficiency by using the Laplace transform of the desired signal and the interference.

\begin{lemma} \label{lem:laplace_mrt}
Assume that the typical user is associated with the $k$-th tier BS. We denote ${\bf{h}}_{1,i}^{k}$ as the channel coefficient between the $k$-tier BS located at ${\bf{d}}_{i}^{k}$ and the typical user. When using MRT with $B_k$ feedback, the Laplace transform of the desired signal power is 
\begin{align}
\mathbb{E}\left[ e^{-s\left| ({\bf{h}}_{1,i}^k)^* \hat {\bf{h}}_{1,i}^k \right|^2}\right] = \left(\frac{1}{s+1} \right) \left(\frac{1}{1+s \left(1 - 2^{-\frac{B_k}{N_k-1}} \right)} \right)^{N_k - 1}.
\end{align}
\end{lemma}
\begin{proof}
See Lemma 1 in \cite{jh:twc:16}.
\end{proof}

\begin{lemma} \label{lem:useful}
For the non-negative and independent random variables $X$ and $Y$, we have
\begin{align}
&\mathbb{E}\left[\ln \left(1+\frac{X}{Y+1} \right) \right] = \int_{0}^{\infty}\frac{e^{-z}}{z}\left(1-\mathbb{E}\left[e^{-zX} \right] \right) \mathbb{E}\left[e^{-zY} \right]{\rm d}z.
\end{align}
\end{lemma}
\begin{proof}
See Lemma 1 in \cite{hamdi:useful}
\end{proof}

By exploiting Lemma \ref{lem:laplace_mrt} and \ref{lem:useful}, we derive the sum spectral efficiency in the following corollary. 
\begin{corollary} \label{coro:sum_se}
When the typical user is associated with a $k$-th tier BS, the ergodic spectral efficiency is 
\begin{align} \label{eq:rate_het}
&R_{\rm NC|k}\left(\beta, \bar \lambda_K,\bar N_K, \bar B_K, \bar P_K,\bar S_K \right) =  \log_2(e) \int_{0}^{\infty} \frac{1}{z}   \left( 1 - \left(\frac{1}{z+1} \right) \left(\frac{1}{1+z \left(1 - 2^{-\frac{B_k}{N_k-1}} \right)} \right)^{N_k - 1}\right) \cdot\nonumber \\
&\;\;\;\;\;\;\;\;\;\;\;\;\;\;\;\;\;\;\;\;\;\;\;\;\;\;\;\;\;\;\;\;\;\;\;\;\;\;\;\;\;\;    \left(\frac{\sum_{i=1}^{K} \lambda_i \left(\frac{P_i S_i}{P_k S_k} \right)^{2/\beta} }{\sum_{i=1}^{K}\lambda_i \left(\frac{P_i S_i}{P_k S_k} \right)^{2/\beta} \left[1+\CMcal{D}(z\left(\frac{S_k}{S_i} \right), \beta) \right]} \right){\rm d} z,
\end{align}
and $\CMcal{D}(x,y)$ is defined as \eqref{eq:het:dfunc}.
\end{corollary}
\begin{proof}
In the proof of Theorem \ref{theo:sir_ccdf_hetnet}, we find the Laplace transform of the aggregated interference $\CMcal{L}_I(s)$.
% where $I = \sum_{i=1}^{K}\sum_{{\bf{d}}_j^i \in \Phi_i \backslash {\bf{d}}_1^{k}}   \frac{P_{i}}{P_k} \left\|{\bf{d}}_1^k \right\|^{\beta} \left\| {\bf{d}}_j^{i} \right\|^{-\beta} \left| ({\bf{h}}_{1,j}^{i})^*{\bf{v}}_j^{i} \right|^2 $. 
The Laplace transform of the desired signal power is obtained in Lemma \ref{lem:laplace_mrt}.
Plugging the obtained Laplace transforms into Lemma \ref{lem:useful} completes the proof.
\end{proof}

%Unlike a cooperative network, however, it is hard to obtain a feedback criterion that maximizes the SIR CCDF or the sum spectral efficiency in a heterogeneous network. As an alternative, we obtain a lower bound on the sum spectral efficiency then find a feedback criterion to maximize the obtained lower bound. To this end, we first present the following lemma that gives a lower bound form.
%\begin{corollary}
%Given $\bar B_K$, the sum spectral efficiency per unit area for a heterogeneous network is lower bounded by
%\begin{align}
%R_{\Sigma} \ge R_{\Sigma, {\rm lb}} = \sum_{k=1}^{K} \lambda_k R_{\rm NC, lb}^{k}, 
%\end{align}
%where 
%\begin{align}
%&R_{\rm NC, lb}^{k} =  \int_{0}^{\infty} \frac{1}{z}   \left( 1 - \left(\frac{1}{z+1} \right) \left(\frac{1}{1+z \left(1 - 2^{-\frac{B_k}{N_k-1}} \right)} \right)^{N_k - 1}\right) \cdot\nonumber \\
%&\;\;\;\;\;\;\;\;\;\;\;\;\;\;\;\;\;\;\;\;\;\;\;\;\;\;\;\;\;\;\;\;\; \left(\frac{\sum_{i=1}^{K} \lambda_i \left(\frac{P_i S_i}{P_k S_k} \right)^{2/\beta} }{\sum_{i=1}^{K}\lambda_i \left(\frac{P_i S_i}{P_k S_k} \right)^{2/\beta} \left[1+\CMcal{D}(z\left(\frac{S_k}{S_i} \right), \beta) \right]} \right){\rm d} z,
%\end{align}
%\end{corollary}
%\begin{proof}
%\end{proof}

\subsection{Adaptive Feedback Partition in Non-Cooperative HetNets}
%We now propose an adaptive feedback partition solution.  target is 
We now determine $B_k$, $k \in \CMcal{K}$ to maximize the ergodic sum spectral efficiency defined as 
\begin{align} \label{eq:sum_se}
R_{\Sigma} = \sum_{k=1}^{K} \lambda_k R_{{\rm NC}|k}\left(\beta, \bar \lambda_K,\bar N_K, \bar B_K, \bar P_K,\bar S_K \right).
\end{align}
The sum spectral efficiency measures the average spectral efficiency provided in an unit area. 
%Since the users associated in the $k$-th tier use the same feedback $B_k$, the feedback used on average in an unit area is obtained as $\sum_{k=1}^{K} \lambda_k B_k$, so that the constraint is  $\sum_{k=1}^{K} \lambda_k B_k \le B_{\rm total}$. 

%Unfortunately, it is hard to find the solution of the problem \eqref{opt:hetnet} since the ergodic spectral efficiency $R^{k}_{\rm NC}$ is a complicated integral form. For this reason, 

Unfortunately, directly maximizing \eqref{eq:sum_se} is not easy since the ergodic spectral efficiency $R_{\rm NC|k}$ in \eqref{eq:rate_het} has a complicated integral form. For this reason, we rather maximize a lower bound on the sum spectral efficiency: 
%A lower bound on $R^{k}_{\rm NC}$ is easily obtained by slightly modifying Corollary 1 of \cite{jh:twc:16}. 
\begin{align} \label{eq:lower_bound_noncoopt}
R_{\rm NC|k}\left(\beta, \bar \lambda_K,\bar N_K, \bar B_K, \bar P_K,\bar S_K \right) & \ge 
R_{\rm NC, lb|k}\left(\beta,\bar \lambda_K, \bar N_K, \bar B_K, \bar P_K,\bar S_K \right) \nonumber \\
& = \log_2\left(1 + \left(1 - 2^{-\frac{B_k}{N_k - 1}} \right) \frac{\exp(\psi(N))}{\sum_{i=1}^{K}\mathbb{E}[I_i]} \right),
\end{align}
where $\psi(\cdot)$ is the digamma function defined as
\begin{align}
\psi(x) = \int_{0}^{\infty} \frac{e^{-t}}{t} - \frac{e^{-xt}}{1- e^{-t}} {\rm d} t	
\end{align}
and the inequality follows from Corollary 1 of \cite{jh:twc:16}.
The mean interference of the $i$-th iter $\mathbb{E}[I_i]$ is calculated as follows 
%\begin{align}
%&\frac{\partial L(\bar B_K, \mu)}{\partial B_k} = \nonumber \\
%& 2^{-\frac{B_k}{N_k-1}}\int_{0}^{\infty} \left(\frac{1}{1+\CMcal{D}(z,\beta)}\right)\left(\frac{1}{1+z} \right) \left(\frac{1}{1+z\left(1-2^{-\frac{B_k}{N_k-1}}\right)} \right)^{N_k} {\rm d} z + \lambda_k \mu \label{eq:lagrangian_het}\\
%&\frac{\partial L_{\rm het}(\bar B_K, \mu)}{\partial \mu} =   \sum_{k=1}^{K} \lambda_k B_{k} - B_{\rm total}.
%\end{align}
%Solving \eqref{eq:lagrangian_het}, however, is infeasible due to the integral form. To resolve this, we approach to obtain a lower bound of $\bar B^{\star}_K$
%\begin{align}
%\frac{\partial L_{\rm het}(\bar B_K, \mu)}{\partial B_k} & \ge \left( \frac{\partial L_{\rm het}(\bar B_K, \mu)}{\partial B_k}\right)_{\rm LB} \nonumber \\
%& = \frac{2^{-\frac{B_k}{N_k-1}}}{\frac{2}{\beta-2} +1 + N_k \left(1- 2^{-\frac{B_k}{N_k-1}} \right)}
%\end{align}
%We first calculate the following
\begin{align}
\mathbb{E}\left[I_i \right] 
%&= \mathbb{E}\left[\frac{\left\| {\bf{d}}_1^k\right\|^{\beta}}{P_k} \sum_{j \in \Phi_i \backslash {\bf{d}}_1^k} P_i \left\| {\bf{d}}_j^i\right\|^{-\beta} \left| ({\bf{h}}_{1,j}^{i})^*{\bf{v}}_j^{i} \right|^2 \right] \nonumber \\
%&= \mathbb{E}_{R}\left[\mathbb{E}\left[ \left. R^{\beta} \sum_{{\bf{d}}_j^i \in \Phi_i \backslash {\bf{d}}_1^k} \frac{P_i}{P_k} \left\| {\bf{d}}_j^i\right\|^{-\beta} \left| ({\bf{h}}_{1,j}^{i})^*{\bf{v}}_j^{i} \right|^2 \right| \left\| {\bf{d}}_1^k\right\| = R \right] \right] \nonumber \\
& \mathop {=}^{(a)} \mathbb{E}_R \left[2 \pi \lambda_i R^{\beta} \frac{P_i}{P_k} \int_{(\frac{P_iS_i}{P_kS_k})^{1/\beta}R}^{\infty} r^{-\beta+1} {\rm d} r\right] \nonumber \\
%&= \mathbb{E}_R \left[2 \pi \lambda_i  \left(\frac{P_i}{P_k}\right)^{2/\beta} \left(\frac{S_i}{S_k} \right)^{2/\beta-1}\frac{R^{2}}{\beta-2} \right] \nonumber \\
& = \frac{2\lambda_i \left(\frac{P_iS_i}{P_k S_k} \right)^{2/\beta}\left(\frac{S_k}{S_i} \right)}{(\beta-2)\sum_{i=1}^{K} \lambda_i \left(\frac{P_i S_i}{P_kS_k} \right)^{2/\beta}},
\end{align}
where (a) follows that each interference fading $\left| ({\bf{h}}_{1,j}^{i})^*{\bf{v}}_j^{i} \right|^2$ is an exponential random variable with unit mean.
By summing up from $i=1$ to $K$, the following is obtained as
\begin{align} \label{eq:expect_i}
\sum_{i=1}^{K}\mathbb{E}[I_i] = \frac{2\sum_{i=1}^{K}\lambda_i \left(\frac{P_iS_i}{P_k S_k} \right)^{2/\beta}\left(\frac{S_k}{S_i} \right)}{(\beta-2)\sum_{i=1}^{K} \lambda_i \left(\frac{P_i S_i}{P_kS_k} \right)^{2/\beta}}.
\end{align}
For simplicity, we denote that $\sum_{i=1}^{K}\mathbb{E}[I_i] = I_{\rm mean}$. 
Then a lower bound on the ergodic sum spectral efficiency is
\begin{align} \label{eq:nc_lowerb}
R_{\Sigma, {\rm lb}} = \sum_{k = 1}^{K} \lambda_k R_{\rm NC, lb|k}\left(\beta,\bar \lambda_K, \bar N_K, \bar B_K, \bar P_K,\bar S_K \right).
\end{align}
Leveraging \eqref{eq:nc_lowerb}, a feedback partition is obtained in the following proposition.
%This is because, the number of users associated in the $k$-th tier 
%\begin{align} \label{eq:fb_const_nc}
%\sum_{k=1}^{K} \lambda_k B_k \le B_{\rm total},
%\end{align}
%where $B_{\rm total}$ means the feedback constraint in the unit area. 

%In a heterogeneous network, however, it is hard to obtain the exact optimal criterion that maximizes the sum spectral efficiency. For this reason, we derive a water-filling like feedback strategy that maximizes a lower bound on the sum spectral efficiency. In the later section, we show that the derived strategy brings substantial performance improvement. 

\begin{proposition} \label{prop:het}
In the non-cooperative case, the feedback partition that maximizes the lower bound is 
\begin{align} \label{eq:fb_partition_het}
(\tilde B_k^{\star})_{\rm NC} = (N_k-1)\log_2\left(\frac{ \exp(\psi(N) )\left( N_k -1- \frac{1}{\mu}\right)}{ (N_k-1) \left( \exp(\psi(N)) + \frac{2\sum_{i=1}^{K}\lambda_i \left(\frac{P_iS_i}{P_k S_k} \right)^{2/\beta}\left(\frac{S_k}{S_i} \right)}{(\beta-2)\sum_{i=1}^{K} \lambda_i \left(\frac{P_i S_i}{P_kS_k} \right)^{2/\beta}}\right) } \right).
\end{align}
The parameter $-1/\mu$ is the maximum value that satisfies
\begin{align}
\sum_{k=1}^{K} \lambda_k (N_k-1)\log_2\left(\frac{ \exp(\psi(N) )\left( N_k -1- \frac{1}{\mu}\right)}{ (N_k-1) \left( \exp(\psi(N)) + \frac{2\sum_{i=1}^{K}\lambda_i \left(\frac{P_iS_i}{P_k S_k} \right)^{2/\beta}\left(\frac{S_k}{S_i} \right)}{(\beta-2)\sum_{i=1}^{K} \lambda_i \left(\frac{P_i S_i}{P_kS_k} \right)^{2/\beta}}\right) } \right) \le B_{\rm total},
\end{align}
where $B_{\rm total}$ is the total feedback constraint per unit area. 
%In practice, we use $\lfloor \{ (\tilde B_{k}^{\star} )_{\rm NC}  \}^{+} \rfloor$ to meet the positive integer constraint.
\end{proposition}

\begin{proof}
Using \eqref{eq:nc_lowerb}, we formulate the feedback partition problem as
%\begin{align} \label{opt:hetnet}
%&\mathop{\rm maximize}_{B_{k} \in \mathbb{Z}^+ , k\in \{1,...,K\}}:\; \sum_{k=1}^{K} \lambda_{k} R_{\rm NC}^k\left(\beta, \bar \lambda_K,\bar N_K, \bar B_K, \bar P_K,\bar S_K \right),  \nonumber \\
%&\;\;\;\;\;{\rm subject\;to}: \; \sum_{k=1}^{K} \lambda_k B_{k} \le B_{\rm total}.
%\end{align}
%One should note that the total feedback constraint is defined as the weighted sum of the feedback allocated to each tier. This is because we consider the feedback constraint as the average feedback bits used in unit area. If the users in the $k$-th tier use $B_k$ for $k \in \CMcal{K}$, the average feedback bits used in unit area is $\sum_{k=1}^{K} \lambda_k B_{k}$.
%on the weighted sum of the feedback bits of each tier, where each weight is corresponding to the tier's density. Assuming that the users in the $k$-th tier use $B_k$ feedback bits, the average feedback bits used per unit area is $\lambda_k B_k$. For this reason, the total feedback constraint $B_{\rm total}$ in a heterogeneous network means the average feedback bits used per unit area. 
% With the obtained lower bound, the new feedback partition problem is formulated as
\begin{align} \label{opt:hetnet_lb}
&\mathop{\rm maximize}_{B_{k} \in \mathbb{Z}^+ , k\in \{1,...,K\}} \; \sum_{k=1}^{K} \lambda_{k} R_{\rm NC,lb|k}\left(\beta, \bar \lambda_K,\bar N_K, \bar B_K, \bar P_K,\bar S_K \right),  \nonumber \\
&\;\;\;\;\;{\rm subject\;to} \; \sum_{k=1}^{K} \lambda_k B_{k} \le B_{\rm total}.
\end{align}
Avoiding integer programming, we relax the feasible field of a solution to $\mathbb{R}^+$. Then the corresponding Lagrangian function is 
\begin{align}
L(\bar B_K, \mu) = \sum_{k=1}^{K} \lambda_{k} R_{\rm NC, lb|k}\left(\beta, \bar \lambda_K,\bar N_K, \bar B_K, \bar P_K,\bar S_K \right) + \mu\left( \sum_{k=1}^{K} \lambda_k B_{k} - B_{\rm total}\right).
\end{align}
According to the KKT condition, we have
\begin{align}
&\frac{\partial L(\bar B_K, \mu)}{\partial B_k} =   \frac{\lambda_k 2^{-\frac{B_k}{N_k-1}} \exp(\psi(N))}{ (N_k-1) \left(I_{\rm mean}+ \left(1-2^{-\frac{B_k}{N_k-1}} \right)\exp(\psi(N)) \right)} + \lambda_k\mu, \label{eq:kkt1_het_lb} \\
&\frac{\partial L(\bar B_K, \mu)}{\partial \mu} =   \sum_{k=1}^{K} \lambda_k B_{k} - B_{\rm total} = 0.
\end{align}
Solving the first condition $\eqref{eq:kkt1_het_lb}$, we have
\begin{align}
(\tilde B_k^{\star})_{\rm NC} = (N_k-1)\log_2\left(\frac{ \exp(\psi(N) )\left( N_k -1- \frac{1}{\mu}\right)}{ (N_k-1) \left( \exp(\psi(N)) + \frac{2\sum_{i=1}^{K}\lambda_i \left(\frac{P_iS_i}{P_k S_k} \right)^{2/\beta}\left(\frac{S_k}{S_i} \right)}{(\beta-2)\sum_{i=1}^{K} \lambda_i \left(\frac{P_i S_i}{P_kS_k} \right)^{2/\beta}}\right) } \right).
\end{align}
The parameter $-1/\mu$ is determined as the maximum value that satisfies 
\begin{align}
\sum_{k=1}^{K} \lambda_k (N_k-1)\log_2\left(\frac{ \exp(\psi(N) )\left( N_k -1- \frac{1}{\mu}\right)}{ (N_k-1) \left( \exp(\psi(N)) + \frac{2\sum_{i=1}^{K}\lambda_i \left(\frac{P_iS_i}{P_k S_k} \right)^{2/\beta}\left(\frac{S_k}{S_i} \right)}{(\beta-2)\sum_{i=1}^{K} \lambda_i \left(\frac{P_i S_i}{P_kS_k} \right)^{2/\beta}}\right) } \right) \le B_{\rm total}.
\end{align}
%Due to the parameter $-1/\mu$, the low $$
This completes the proof.
\end{proof}
%Unlike a cooperative network, it is not easy to understand the interplay between the system parameters, e.g., the number of antenna, the density, the transmit power, and the biasing factor, in the obtained feedback criterion. In the next section, we provide intuition by numerical simulations

\begin{remark} \label{remark:floor} \normalfont
Since the feedback has an positive integer value in practice, we have to perform further processes to the relaxed solution $(\tilde B_k^{\star})_{\rm NC}$. We introduce two possible methods.
First, we can use the round function $\lfloor (\tilde B_k^{\star})_{\rm NC} \rceil$. With the round function, it is possible that $\sum_{k=1}^{K} \lambda_k \lfloor ( \tilde B_k^{\star} )_{\rm NC} \rceil \ge B_{\rm total}$, therefore the manual feedback adjustment is necessary after applying the round function. 
%Furthermore, using the round function only marginally affects the complexity compared to that of the manual search with the floor function. For this reason, we focus on the floor function in this paper.
Second, we can iteratively add a feedback bit to each tier. For example, starting with the floored solution, we iteratively find which tier is the best choice for adding a remaining feedback bit by computing the sum ergodic spectral efficiency. Subsequently, we add a feedback bit to the selected tier. We repeat this until the used feedback equals to $B_{\rm total}$. 
%the proposed solution can be further optimized. For example, starting from $\lfloor {\bf{b}}^{\star} \rfloor$, we could manually search which antenna is the best choice for adding resolution by calculating the GMI, then add a resolution to the chosen antenna. We repeat this until the sum resolution of the modified solution meets $C$. 
More detail processes for supplementing the floor function are described in Remark 3 and 4 of \cite{park:tcom:17}. 
%Since considering the remaining resolutions at one time requires too much complexity, we 
%Detail process of this manual search algorithm is explained in Algorithm \ref{alg:Aux}. 
Due to space limitation we do not explore these methods in this paper.
\end{remark}

\begin{remark} \normalfont \label{remark:water}
We observe that Proposition \ref{prop:het} is a water-filling type solution where $-1/\mu$ determines the water level. As in a conventional water-filling method, $-1/\mu$ can be found by an iterative algorithm. Moreover, the iterative algorithm does not have to be performed frequently since Proposition \ref{prop:het} only depends on the system parameters that do not change often, for example the number of antennas $N_k$, the transmit power $P_k$, the biasing factor $S_k$, the density $\lambda_k$, and the path-loss exponent $\beta$. 
%the iterative algorithm does not have to be done frequently. 
%This makes it possible that distributed way, which is crucial in heterogeneous network. 
In particular, Proposition \ref{prop:het} does not change depending on instantaneous SIR, so that the complexity required for Proposition \ref{prop:het} is low.
%without sharing instantaneous information, e.g., user's instantaneous locations. 
%This is particularly beneficial in a HetNet, where heavy overheads are required for sharing information with many other BSs when they are densely deployed.
\end{remark}

\begin{remark} \normalfont
When there is no biased association, i.e., $S_i = 1$ for $i \in \CMcal{K}$, the proposed strategy becomes 
\begin{align}
(\tilde B_k^{\star})_{\rm NC} = (N_k-1)\log_2\left(\frac{ \exp(\psi(N) )\left( N_k -1- \frac{1}{\mu}\right)}{ (N_k-1) \left( \exp(\psi(N)) + \frac{2}{\beta-2}\right) } \right),
\end{align}
which is a function of the number of antennas $N_k$ and the path-loss exponent $\beta$. 
%For this reason, the feedback strategy is independent to the each tier's density $\lambda_k$ and the transmit power $P_k$
This result matches with \cite{hsjo:2012_twc}, which showed that when there is no biased association in a single antenna $K$-tier HetNet, the spectral efficiency is only a function of the path-loss exponent $\beta$. 
\end{remark}

\begin{remark} \normalfont
In the proposed feedback partition strategy \eqref{eq:fb_partition_het}, we observe that the feedback for single-user MRT in a particular tier is useful only if the corresponding mean interference $I_{\rm mean}$ is small enough. Specifically, recalling \eqref{eq:expect_i}, $(\tilde B_k^{\star})_{\rm NC} \ge 1$ only if 
\begin{align} \label{eq:cond_nofb}
\exp(\psi(N_k)) \left(\frac{(N_k-1-\frac{1}{\mu})2^{-\frac{1}{N_k-1}}}{N_k - 1} -1 \right)\ge  I_{\rm mean}.
\end{align}
Since $-1/\mu$ is proportional to $B_{\rm total}$, if the amount of mean interference is too large compared to $B_{\rm total}$, it is beneficial not to allocate the feedback to the corresponding tier.
%e observe that if the biasing factor of the $k$-th tier $S_k$ increases, the allocated feedback decreases. 
%Focusing on the $k$-th tier, the mean interference $I_{\rm mean }$ increases when $S_k$ increases. This is because the desired power of the $k$-th tier BS is overestimated so that the actual SIR of a $k$-th tier user decreases. 
As the mean interference increases, the SIR becomes low and this makes the spectral efficiency improved marginally by using the feedback. For this reason, saving the feedback of the corresponding tier for other tiers is more beneficial. 
%the SIR of a $k$-th tier user decreases. In the low SIR regime, the feedback efficiency \cite{jh:twc:16}, which measures the spectral efficiency improvement as increasing a small amount of feedback bits, is also low. 
%This implies that the spectral efficiency only marginally improves by using the feedback, so that saving the feedback resource for other tiers is more advantageous. 
%For this reason, decreasing the feedback bits of the corresponding tier for saving the feedback resource is more beneficial. 
%{\textbf{Mention about expect value of Interference}}
\end{remark}

\section{Cooperative Case} 
In the previous section, if a particular tier's user experiences a large amount of interference, single-user MRT based on the limited feedback is not useful to improve the spectral efficiency. 
%This implies that, when a certain tier suffers from strong interference, it is not efficient to use single-user MRT based on feedback.
In this case, we can use the feedback information to mitigate the other BS interference through multi-cell coordination. 
%Motivated by this, we consider a cooperative case in this section. 
In this section, we analyze the performance of BS coordination in a HetNet and propose a feedback partition by leveraging the derived expressions.

\subsection{Performance Characterization}
%Theorem \ref{theo:sir_ccdf} is the main result in this subsection.
%we use the PDF of $\delta_{1,K}$ obtained in Lemma 1 of \cite{NY:dynamic}. 
%\begin{align}
%f_{\delta_{1,K}}(x) = 2(K-1)x(1-x^2)^{K-2}. 
%\end{align}
%The 

%The difficulty in characterizing the performance of the cooperative case is that each tier BS has different transmit power and the biasing factor, so that it is not trivial to capture the 

%We define $\CMcal{K}'$ as a set of the tiers that do not satisfy \eqref{eq:cond_nofb}, and 

In the performance characterization of the cooperative case, a challenging part is obtaining the distribution of the distance to the BS at ${\bf{d}}_{i_L}$ (the PDF of $\left\| {\bf{d}}_{i_L}\right\|$). This indicates the distance of the BS located furthest from the typical user in the coordination set $\CMcal{C}$. It is important because it determines a boundary between the intra-cluster interference and the out-of-cluster interference, which is necessary for the feedback allocation. The main source of the difficulty is that each tier uses a different transmit power and biasing factor, so that the intensity measure of aggregated signal power of each tier has different features. Due to this heterogeneity, ordering the BSs according to their biased power across the tiers is complicated. 
%In particular, it is challenging to directly obtain 
%the distribution of the distance to the $L$-th closest BS (the PDF of $\left\| {\bf{d}}_{i_L}\right\|$), which is necessary to characterize the performance of the cooperative case.
%In particular, it is not easy to obtain the $L$-th closest BS (last BS included in the cluster)
To resolve this, we first derive the following lemma that transforms a $K$-tier HetNet to a statistically equivalent single-tier network.
%Now we characterize the SIR CCDF and the ergodic spectral efficiency defined in the previous subsection. We first obtain the SIR CCDF in Theorem \ref{theo:sir_ccdf_coop}.

\begin{lemma} [Transformation lemma] \label{lem:equivalent_transition}
Consider the $\ell$-th tier network for $\ell \in \CMcal{K}$ denoted as $\Phi_\ell^{\rm M} = \{{\bf{d}}_{i}^{\ell}, P_\ell, S_\ell, i \in \mathbb{N}\} $ with density $\lambda_{\ell}$. The intensity measure of biased signal power of $\Phi_\ell^{\rm M}$ received by the typical user, i.e., $P_\ell S_\ell \left\| {\bf{d}}_i^{\ell}\right\|^{-\beta}$, is statistically equivalent to that of $\Phi_{\ell \rightarrow k}^{\rm M} = \{{\bf{d}}_{i}^{\ell \rightarrow k}, P_k, S_k, i \in \mathbb{N}\} $ with density $\tilde \lambda_{\ell}$, provided that the density $\tilde \lambda_{\ell}$ is scaled to
%The SIR distribution of this network is statistically equivalent to that of the $\ell$-th tier network with different parameters ${P}_{\ell}$ and ${S}_{\ell}$, $\{\Phi_{\ell}, {\lambda}_{\ell}, {S}_{\ell}, {P}_{\ell}\}$, provided that the density ${ \lambda}_{\ell}$ is rescaled such that 
\begin{align}
{\tilde \lambda}_{\ell}=\lambda_\ell \left( \frac{P_\ell S_\ell}{P_k S_k}\right)^{\frac{2}{\beta}}.
\end{align}
\end{lemma}
\begin{proof}
By the displacement theorem \cite{baccelli:inria}, the intensity measure of biased signal power of $\Phi_\ell^{\rm M}$ experienced by the typical user is
\begin{align}
\Lambda_{\ell}((0,t]) 
&= \mathbb{E}\left[\sum_{{\bf{d}}_i^{\ell} \in \Phi_\ell^{\rm M}} {\bf{1}}\left( \frac{\left\| {\bf{d}}_i^{\ell} \right\|^{\beta}}{P_{\ell} S_{\ell}}  < t \right) \right] \nonumber \\
& \mathop{=}^{(a)} 2 \pi \lambda_{\ell } \int_{0}^{({P_\ell S_\ell t})^{\frac{1}{\beta}}} r {\rm d} r \nonumber \\
& = \pi \lambda_{\ell} \left( P_{\ell} S_{\ell} \right)^{2/\beta} t^{2/\beta}
%\frac{1}{t^{2/\beta}},
\end{align}
where (a) follows Campbell's theorem \cite{baccelli:inria}. Similarly, the intensity measure of biased signal power of $\Phi_{\ell \rightarrow k}^{\rm M}$ is
\begin{align}
\Lambda_{\ell \rightarrow k}((0,t])  = \pi \tilde \lambda_{\ell} \left( P_{k} S_{k} \right)^{2/\beta} t^{2/\beta}.
\end{align}
For this reason, if ${\tilde \lambda}_{\ell}=\lambda_\ell \left( \frac{P_\ell S_\ell}{P_k S_k}\right)^{\frac{2}{\beta}}$, the two biased signal power becomes equivalent. This completes the proof.
\end{proof}

The implication of Lemma \ref{lem:equivalent_transition} is that by rescaling each density as ${\tilde \lambda}_{\ell }=\lambda_\ell \left( \frac{P_\ell S_\ell}{P_k S_k}\right)^{\frac{2}{\beta}}$ for $\ell \in \CMcal{K}$, a $K$-tier HetNet 
%where each tier has different biasing factor and transmit power 
can be transformed to a statistically equivalent network where the transmit power and the biasing factor are same as $P_k$ and $S_k$. Leveraging this, we obtain the PDF of $\left\| {\bf{d}}_{i_L}\right\|$ in the following lemma.

\begin{lemma} \label{lem:het_L_closest}
Assume that the furthest BS of the coordination set $\CMcal{C}$ belongs to the $k$-th tier, i.e., $\pi(i_L) = k$. Then the PDF of the distance $\left\| {\bf{d}}_{i_L} \right\|$ is 
\begin{align}
f_{\left\| {\bf{d}}_{i_L} \right\|}(r) = \frac{2\left( \pi \sum_{i=1}^{K} \lambda_i \left(\frac{P_i  S_i} {P_k S_k} \right)^{2/\beta}r^2 \right)^L}{r \Gamma(L)} \exp\left( { -\pi \sum_{i=1}^{K} \lambda_i \left(\frac{P_i  S_i} {P_k S_k} \right)^{2/\beta} r^2} \right).
\end{align}
\end{lemma}
\begin{proof}
We first transform a $K$-tier HetNet to a single-tier network whose transmit power and biasing factor are equal to $P_k$ and $S_k$. By exploiting Lemma \ref{lem:equivalent_transition}, we rescale the density as $\lambda_\ell \left(\left(P_\ell S_\ell\right) / \left( P_k S_k \right)\right)^{2/\beta}$. By doing this, we transform the $\ell$-th tier network to $\Phi_{\ell \rightarrow k}^{\rm M} = \{{\bf{d}}_{i}^{\ell \rightarrow k}, P_k, S_k, i \in \mathbb{N}\} $ with density $\lambda_\ell \left(\left(P_\ell S_\ell\right) / \left( P_k S_k \right)\right)^{2/\beta}$. Note that the original $\ell$-th tier network $\Phi_{\ell}^{\rm M}$ and the transformed $\ell$-th tier network $\Phi_{\ell \rightarrow k}^{\rm M}$ are statistically equivalent as shown in Lemma \ref{lem:equivalent_transition}. Then, by the superposition theorem \cite{baccelli:inria}, the aggregated network $\sum_{\ell \in \CMcal{K}} \Phi_{\ell \rightarrow k}^{\rm M}$ is a homogeneous network with transmit power $P_k$, biasing factor $S_k$, and density $\sum_{i=1}^{K} \lambda_i \left(\frac{P_i  S_i} {P_k S_k} \right)^{2/\beta}$.
% We write that $\tilde \Phi^{\rm M} = \{{\bf{d}}_i, S_{k}, P_{k}, i \in \mathbb{N} \}$
Since $\sum_{\ell \in \CMcal{K}} \Phi_{\ell \rightarrow k}^{\rm M}$ is a homogeneous network, we can use the conventional PDF of the distance presented in \cite{haenggi:tit:05}. In a homogeneous PPP with density $\lambda$, the PDF of the $L$-th closest point to the origin is
\begin{align} \label{dist:kthpoint_PPP}
f(r) = \frac{2(\lambda\pi r^2)^L}{r \Gamma\left( L \right)}e^{-\lambda \pi r^2}.
\end{align}
Plugging $\sum_{i=1}^{K} \lambda_i \left(\frac{P_i  S_i} {P_k S_k} \right)^{2/\beta}$ into $\lambda$ completes the proof.
\end{proof}

\begin{remark} \normalfont
When $L = 1$, i.e., the non-cooperative case, the obtained PDF in Lemma \ref{lem:het_L_closest} boils down to \eqref{eq:pdf_first_touch}, which describes the PDF of the closest BS conditioned on that the typical user is associated with the $k$-th tier BS. This implies that our transformation lemma can be applied to characterize a general distance distribution in a $K$-tier HetNet. 
\end{remark}

Next, we define the intra-cluster BS geometry parameter $\delta_{1, \ell}$, $\ell \in \{2,...,L\}$ to characterize the relative intra-cluster interference power.
%To characterize the intra-cluster BS geometry, we adopt a concept of the geometric parameter $\delta_{1, \ell}$ for $\ell \in \{2,...,L\}$ introduced in \cite{ny:twc:15_dynamic}. 
We define the geometric parameter $\delta_{1, \ell}$ as the ratio between the path-loss of the home BS and the $\ell$-th closest BS for $\ell \in\{2,...,L\}$, i.e., $\delta_{1, \ell} =  \left( P_{\pi(i_{\ell})}\left\| {\bf{d}}_{i_{\ell}} \right\|^{-\beta} \right)/ \left( P_{\pi(i_1)}\left\| {\bf{d}}_{i_1} \right\|^{-\beta} \right)$.
% By the definition, $0<\delta_{1,\ell}<1$. 
 We note that the geometric parameter $\delta_{1, \ell}$ is originally introduced in \cite{lee:twc:15}, and is generalized for HetNets in our work. As explained in \cite{lee:twc:15}, $\delta_{1, \ell}$ measures the relative intra-cluster interference power coming from ${\bf{d}}_{i_\ell}$, so that a large value of $\delta_{1, \ell}$ means large amount of intra-cluster interference. 
When each biasing factor is same, i.e., $S_1 = ... = S_K$, $\delta_{1, \ell_1} > \delta_{1, \ell_2}$ if $\ell_1 < \ell_2$ by the definition. For general biasing factors, however, this is not necessarily guaranteed. 
We denote a set of the geometric parameters as $\bar \delta_{1,L} = \{\delta_{1,2}, ..., \delta_{1, L}\}$, and analyze the performance of the cooperative case under the assumption that the relative intra-cluster interference power is fixed, while out-of-cluster interference is random as in \cite{lee:twc:15}.
%With the cluster size $L$, $\delta_{1,L}$ means the distance ratio between the home BS and the furthest BS (the $L$-th closest BS) in the cluster. For this reason, a small value of $\delta_{1,L}$ implies that the $L$-th closest BS is located relatively far from the typical user compared to the home BS, so that the typical user is well protected from the out-of-cluster interference.

By using Lemma \ref{lem:het_L_closest} and the intra-cluster BS geometry, we derive the following theorem that presents the SIR CCDF of the cooperative case.

\begin{theorem} \label{theo:sir_ccdf_coop}
Assume that $\bar \delta_{1, L}$, $\bar B_{L}$ is given, and also $\pi(i_1) = m$, $\pi(i_L) = k$. Then, the conditioned SIR CCDF of a K-tier HetNet in the cooperative case is
\begin{align} \label{eq:sir_ccdf_coop}
%&F^{\rm c}_{\rm SIR_{\rm coop.}}\left(L,\beta, \bar B_{L}, \bar \delta_{1,L}; \gamma\right) =
&F^{\rm c}_{{\rm SIR}_{\rm C}^{m}}\left(\beta, \bar \lambda_K,\bar N_K, \bar B_{{L}}, \bar P_K, \bar S_K, \bar \delta_{1,L};\gamma \right) = \nonumber \\
&  \prod_{i_\ell \in \CMcal{C}\backslash i_1} \left(\frac{1}{1+\gamma \delta_{1,\ell}  2^{-\frac{B_{i_\ell}}{L-1}} } \right) 
\left( \frac{\sum_{i=1}^{K} \lambda_i \left(\frac{P_i S_i}{P_k S_k} \right)^{2/\beta} }{\sum_{i=1}^{K} \lambda_i \left(\frac{P_i S_i}{P_k S_k} \right)^{2/\beta}\left[1+ \CMcal{D}\left(\gamma \delta_{1, L} \cdot \left(\frac{S_k}{S_i} \right), \beta \right) \right]} \right)^L,
%\prod_{\ell=2}^{L} \left(\frac{1}{1+\gamma \left(\delta_{1,\ell}  \right)^{\beta} 2^{-\frac{B_\ell}{L-1}} } \right) \left( \frac{1}{1+\CMcal{D}\left(\gamma(\delta_{1,L})^{\beta},\beta \right)} \right)^L,
\end{align}
where 
\begin{align} \label{eq:coop:dfunc}
\CMcal{D}( x, y) = \frac{2x}{y-2} {}_2F_1\left(1, 1-\frac{2}{y}, 2-\frac{2}{y}, -x \right),
\end{align}
with ${}_2F_1\left(\cdot,\cdot,\cdot,\cdot\right)$ is the Gaussian hypergeometric function.
\end{theorem}
\begin{proof}
See Appendix \ref{appen:theo2}.
\end{proof}
%Theorem \ref{theo:sir_ccdf} shows that the SIR coverage performance %consists of the products between the intra-cluster interference and the out-of-cluster interference.
%In particular, the intra-cluster interference is determined by 
%as a function of the relevant systems parameters: intra-cluster geometry $\delta_{1,\bar K}$, the number of feedback bits allocated to each cooperative BS $B_{\bar K}$, and the path-loss exponent $\beta$. 
%As the number of feedback bits goes to infinity so that the intra-cluster interference is nullified perfectly, Theorem \ref{theo:sir_ccdf_coop} boils down to Corollary 1 in \cite{ny:twc:15_dynamic}. 
%Theorem \ref{theo:sir_ccdf_coop} provides implicit intuition regarding an efficient feedback partition criterion. In \eqref{eq:sir_ccdf_coop}, we observe that the SIR CCDF consists of a product of the Laplace transform of each intra-cluster interference term. 
%By the arithmetic and geometric means inequality, a product of non-negative variables is maximized when all the variables are same.
%For this reason, to maximize the SIR CCDF, it is beneficial to allocate many feedback bits to a closer intra-cluster BS to make each Laplace transform same. 
%To do this, the feedback bits should be allocated proportional to the inverse of the intra-cluster BS's distance (or proportional to the received signal power). 
%This observation will be confirmed formally in the next subsection.

%Now we derive the ergodic spectral efficiency. Prior to the derivation, we introduce Lemma \ref{lem:useful} which is useful to obtain the ergodic spectral efficiency.

We summarize the conditions presented in Theorem \ref{theo:sir_ccdf_coop}. The SIR CCDF is derived under the conditions that (i) $\pi(i_1)$ is fixed as $m$, (ii) $\pi(i_L)$ is fixed as $k$, and (iii) $\bar \delta_{1,L}$ is fixed, so that the relative intra-cluster interference power is given. 
%Henceforth, we denote that $\pi(i_1) = m$ for simplicity.
Now the ergodic spectral efficiency is derived as an integral form in Corollary \ref{coro:rate_integral_coop}.
\begin{corollary} \label{coro:rate_integral_coop}
Assume that $\bar \delta_{1, L}$, $\bar B_{L}$ is given and $\pi(i_L) = k$. 
Then, the ergodic spectral efficiency of a K-tier HetNet in the cooperative case is
\begin{align} \label{eq:rate_integral}
& R_{\rm C}^{m} \left(\beta, \bar \lambda_K,\bar N_K, \bar B_{{L}}, \bar P_K, \bar S_K, \bar \delta_{1,L} \right) = \nonumber \\
%&R_{\rm coop}(L, \beta, \bar B_{L}, \bar \delta_{1, L}) = \nonumber \\
& \log_2(e) \int_{0}^{\infty} \!\!\! \frac{1}{1+z}\!\!\! \prod_{i_\ell \in \CMcal{C}\backslash i_1} \left(\frac{1}{1+z \delta_{1,\ell}  2^{-\frac{B_{i_\ell}}{L-1}} } \right) \!\!\! \left( \frac{\sum_{i=1}^{K} \lambda_i \left(\frac{P_i S_i}{P_k S_k} \right)^{2/\beta} }{\sum_{i=1}^{K} \lambda_i  \left(\frac{P_i S_i}{P_k S_k} \right)^{2/\beta}\left[1+ \CMcal{D}\left(z \delta_{1, L} \cdot \left(\frac{S_k}{S_i} \right), \beta \right) \right]} \right)^L  \!\! {\rm d} z,
\end{align}
where $\CMcal{D}(x,y)$ is defined as \eqref{eq:coop:dfunc}.
\end{corollary}
\begin{proof}
%The proof is directly from Lemma \ref{lem:useful}.
From Theorem \ref{theo:sir_ccdf_coop}, we get the Laplace transform of the intra-cluster and the out-of-cluster interference. Plugging them into Lemma \ref{lem:useful} completes the proof. 
\end{proof}

\subsection{Adaptive Feedback Partition in Cooperative HetNets}
%Since the intra-cluster interference $I_{\rm In}(\bar \delta_{1, L}, \bar B_{ L})$ is a function of 
%Now we obtain a feedback partition to efficiently manage the intra-cluster interference $I_{\rm In}$. 
Now we determine $B_{i_{\ell}}$ for $\ell \in \{2, ..., L\}$ to maximize the ergodic spectral efficiency  \eqref{eq:rate_integral}.
%$R_{\rm C}^{m} \left(\beta, \bar \lambda_K,\bar N_K, \bar B_{{L}}, \bar P_K, \bar S_K, \bar \delta_{1,L} \right)$. 
%Recall that the total feedback bits used in one coordination set is constrained as 
%$\sum_{\ell=2}^{L} B_{i_{\ell}} \le B_{\rm total}$. 
%Specifically, we first present Proposition \ref{prop:scheme1} to maximize the ergodic spectral efficiency when the intra-cluster geometry $\bar \delta_{1,L}$ is known. 
%In Proposition \ref{prop:scheme2}, we propose the feedback partition strategy for a case that $\bar \delta_{1,L}$ is unknown.

\begin{proposition} \label{prop:scheme1}
In the cooperative case, the feedback partition that maximizes the ergodic spectral efficiency $R_{\rm C}^{m} \left(\beta,\bar \lambda_K, \bar N_K, \bar B_{{L}}, \bar P_K, \bar S_K, \bar \delta_{1,L} \right)$ is 
\begin{align} \label{scheme1_coop}
&(B_{i_{\ell}}^{\star})_{\rm C} =   \frac{ B_{\rm total}}{L-1}  + (L-1)\log_2\left( \frac{\delta_{1,\ell}}{\left( \prod_{\ell=2}^{L}\delta_{1,\ell}\right)^{\frac{1}{L-1}} }\right).
\end{align}
%Since $(B_{i_{\ell}}^{\star})_{\rm C} $ for $\ell \in \{2,...,L\}$ is a positive integer in practice, we use $\lfloor \{(B_{i_{\ell}}^{\star})_{\rm C} \}^{+} \rfloor$, where $\{\cdot\}^{+}=\max(\cdot, 0)$.
We refer Remark \ref{remark:floor} to make $(B_{i_{\ell}}^{\star})_{\rm C}$ be a positive integer.
\end{proposition}
\begin{proof}
We first formulate the optimization problem for maximizing the SIR CCDF \eqref{eq:sir_ccdf_coop}. Since the Laplace transform of the out-of-cluster interference is independent to the feedback, we can treat this as a constant and omit it in the problem. Then the problem is 
\begin{align} \label{opt:prod}
&\mathop{\rm maximize}_{B_{i_{\ell}} \in \mathbb{Z}^+ , \ell\in \{2,...,L\}}\; \prod_{i_\ell \in \CMcal{C}\backslash i_1} \left(\frac{1}{1+\gamma \delta_{1,\ell}  2^{-\frac{B_{i_\ell}}{L-1}} } \right),  \nonumber \\
&\;\;\;\;\;{\rm subject\;to} \; \sum_{\ell=2}^{L} B_{i_{\ell}} \le B_{\rm total}.
\end{align}
Since \eqref{opt:prod} is integer programming which is hard to solve, we first relax the feasible field of $B_{\ell}$ to $\mathbb{R}^+$ and apply the floor function to the solution later. 
%For the processes to make up the loss caused by the floor function, please see Remark \ref{remark:floor}. 
%Later, we can turn back to the original field of $B_{\ell}$ by $\lfloor B_{\ell} \rfloor$. 
%Denoting the optimal feedback bits used in the $\ell$-th intra-cluster BS as $(B_{\ell}^{\star})_{\rm coop.}$, it is satisfied that $\sum_{\ell=2}^{L} (B_{\ell}^{\star})_{\rm coop.} \le B_{\rm total}$ by the constraint in \eqref{opt:prod} therefore $\sum_{\ell=2}^{L} \lfloor (B_{\ell}^{\star})_{\rm coop.} \rfloor \le B_{\rm total}$.
%Now we rewrite the objective function of 
%Even with this relaxation, however, it is still not straightforward to solve \eqref{opt:prod} so 
We rewrite \eqref{opt:prod} as
\begin{align} \label{opt:log}
& \mathop{\rm minimize}_{B_{i_\ell} \in \mathbb{R}^+, \ell\in \{2,...,L\}}\; \sum_{\ell=2}^{L} \ln \left({1+\gamma  \delta_{1,\ell} 2^{-\frac{B_{i_\ell}}{L-1}} }  \right),  \nonumber \\
&\;\;\;\;\;{\rm subject\;to} \; \sum_{\ell=2}^{L} B_{i_{\ell}} \le B_{\rm total}.
\end{align}
Since the function $f(B) = \ln (1+C2^{-\frac{B}{L-1}})$ is monotonically increasing function and convex for any positive $C$, we apply a convex optimization technique to solve \eqref{opt:log}. At first, the corresponding Lagrangian function of the objective function in \eqref{opt:log} is
\begin{align} \label{eq:lagrangian}
&L(\bar B_{L}, \mu) = \sum_{\ell=2}^{L} \ln \left({1+\gamma \delta_{1,\ell} 2^{-\frac{B_{i_\ell}}{L-1}} }  \right) + \mu \left(\sum_{\ell=2}^{L} B_{i_\ell}  -  B_{\rm total}\right),
\end{align}
where $\mu$ denotes the Lagrangian multiplier.
% instead of conventionally used $\lambda$. This is because $\lambda$ is already used to represent the density of $\Phi$.
%The KKT condition for \eqref{eq:lagrangian} is 
%\begin{align} 
%& \frac{\partial L(B_k, \mu)}{\partial B_k} = \frac{2^{-\frac{B_k}{K-1}}(\delta_{1,k} )^{\beta} \gamma \ln(2) }{(K-1)\left(1+2^{-\frac{B_k}{K-1}} (\delta_{1,k})^{\beta} \gamma \right)} + \mu = 0, \label{eq:KKT1} \\
%\end{align}
%\begin{align}
%& \frac{\partial L(B_k, \mu)}{\partial \mu} = \sum_{k=1}^{K} B_{k} - B_{\rm total} = 0.\label{eq:KKT2}
%\end{align}
%
%The rate loss due to the limited feedback is
%\begin{align}
%\Delta R &\le \mathbb{E}\left[\log_2\left(1+I_{\rm IUI}(\delta_{1, \bar K}, B_{\bar K} \right)+I_{\rm OCI} \right] \nonumber \\
%&\le \log_2\left(1+ \mathbb{E}\left[I_{\rm IUI}\left(\delta_{1, \bar K}, B_{\bar K} \right) \right] +\mathbb{E}\left[ I_{\rm OCI}\right] \right)
%\end{align}
Solving the KKT conditions for \eqref{eq:lagrangian} leads to 
\begin{align} \label{scheme:bit_allo_inproof}
&(B_{i_\ell}^{\star})_{\rm C} =  \frac{ B_{\rm total}}{L-1}  + (L-1)\log_2\left( \frac{\delta_{1,\ell}}{\left( \prod_{\ell=2}^{L}\delta_{1,\ell}\right)^{\frac{1}{L-1}} }\right).
\end{align}
Since the obtained feedback partition \eqref{scheme:bit_allo_inproof} is not a function of a specific threshold $\gamma$, this is optimal for any threshold, which means it is optimal for maximizing the ergodic spectral efficiency $R_{\rm C}^{m} \left(\beta, \bar \lambda_K,\bar N_K, \bar B_{{L}}, \bar P_K, \bar S_K, \bar \delta_{1,L} \right)$. This completes the proof.
\end{proof}

%Since Proposition \ref{prop:scheme1} is a function of $\bar \delta_{1, L}$, a user is required to measure the intra-cluster interference power from each intra-cluster BS to apply the proposed feedback partition. 

\begin{remark} \normalfont
Proposition \ref{prop:scheme1} implies that the feedback is allocated proportional to the intra-cluster interference power, i.e., $B_{i_\ell} \propto  \delta_{1,\ell}$. Note that this is similar to the previous results \cite{ny:twc:11_adap, akoum:tsp:13}, in which adaptive feedback allocation is proposed in a homogeneous cooperative network for minimizing the rate gap to perfect CSIT case. To use the prior work \cite{ny:twc:11_adap, akoum:tsp:13}, however, not only the relative intra-cluster BS power but also the exact instantaneous SINR should be obtained. On the contrary, Proposition \ref{prop:scheme1} only depends on relative intra-cluster BS power while it does not change depending on instantaneous SIR. 
\end{remark}

\subsection{A General Number of Antennas Case}
In this subsection, we study feedback allocation in a general number of antennas case, where different tier BSs use different number of antennas. The intra-cluster BSs are equipped with $N_{\pi(i_1)}, ..., N_{\pi(i_L)}$ antennas, where $N_{\pi(i_{\ell})} \ge L$, $\ell \in \{1,...,L\}$. 
%In this case, the ergodic spectral efficiency \eqref{eq:rate_integral} is slightly modified to
%\begin{align} \label{eq:rate_integral_general}
%& R_{\rm C, general}^{m} \left(\beta, \bar \lambda_K,\bar N_K, \bar B_{{L}}, \bar P_K, \bar S_K, \bar \delta_{1,L} \right) = \nonumber \\
%%&R_{\rm coop}(L, \beta, \bar B_{L}, \bar \delta_{1, L}) = \nonumber \\
%& \log_2(e) \!\! \int_{0}^{\infty} \!\!\!\!\! \frac{1}{1+z}\!\!\! \prod_{i_\ell \in \CMcal{C}\backslash i_1} \!\!
%\left(\frac{1}{1+z \delta_{1,\ell}  2^{-\frac{B_{i_\ell}}{N_{\pi(i_{\ell})}-1}} } \right) \!\!\! \left( \frac{\sum_{i=1}^{K} \lambda_i \left(\frac{P_i S_i}{P_k S_k} \right)^{2/\beta} }{\sum_{i=1}^{K} \lambda_i  \left(\frac{P_i S_i}{P_k S_k} \right)^{2/\beta}\left[1+ \CMcal{D}\left(z \delta_{1, L} \cdot \left(\frac{S_k}{S_i} \right), \beta \right) \right]} \right)^L  \!\! {\rm d} z.
%\end{align}
Using the additional antennas, we use coordinated beamforming, which mitigates the intra-cluster interference and also increases the desired signal power. 
%We explain this in detail focusing on the BS ${\bf{d}}_{i_1}$. The BS ${\bf{d}}_{i_1}$ should nullify $L-1$ number of channels corresponding to each user associated with other intra-cluster BS. This requires $L-1$ spatial degrees of freedom. For this reason, under the previous assumptions that all the intra-cluster BSs use only $L$ antennas, the remaining spatial degrees of freedom is $1$, so that the BS ${\bf{d}}_{i_1}$ is able to serve the typical user but it cannot enhance the desired signal power due to lack of the spatial resources. When the BS ${\bf{d}}_{i_1}$ has $N_{\pi(i_1)}$ antennas, however, it has remaining $N_{\pi(i_1)} - L$ spatial degrees of freedom after the interference cancellation, and the remaining spatial resources are used to increase the desired signal power. 
Specifically, the beamforming vector ${\bf{v}}_{i_1}$ ($\left\| {\bf{v}}_{i_{1}} \right\| = 1$) used in the BS ${\bf{d}}_{i_1}$ is designed by solving the following optimization problem. We denote that ${\bf{h}}_{\ell, i_1} \in \mathbb{C}^{N_{\pi (i_1)}}$ is a channel vector from the BS at ${\bf{d}}_{i_1}$ to the intra-cluster user associated with the BS at ${\bf{d}}_{i_\ell}$.
% a beamforming vector used in the BS ${\bf{d}}_{i_1}$, denoted as ${\bf{v}}_{i_{1}}$ ($\left\| {\bf{v}}_{i_{1}} \right\| = 1$), is designed by solving
\begin{align} \label{eq:opt_prob_general}
&{\rm maximize}\; \left|({\bf{h}}_{1, i_{1}})^* {\bf{v}}_{i_1} \right|^2, \nonumber \\
&{\rm subject\;to} \; ({\bf{h}}_{\ell, i_1})^* {\bf{v}}_{i_1} = 0, \; \ell \in \mathcal{C}\backslash 1.
\end{align}
The solution of \eqref{eq:opt_prob_general} always exists when $N_{\pi(i_{1})} \ge L$. The other beamforming vector ${\bf{v}}_{i_{\ell}}$ where $\ell \in \{2,...,L\}$ can be designed in the similar way. 
%The other intra-cluster BSs ${\bf{d}}_{i_2}, ..., {\bf{d}}_{i_L}$ also makes its beamforming vector in the equivalent way with that of the BS ${\bf{d}}_{i_1}$.
%\begin{align}
%R_{\rm C}^{m} \left(\beta, \bar \lambda_K,\bar N_K, \bar B_{{L}}, \bar P_K, \bar S_K, \bar \delta_{1,L} \right) \ge R_{\rm C, lb}^{m} \left(\beta, \bar \lambda_K,\bar N_K, \bar B_{{L}}, \bar P_K, \bar S_K, \bar \delta_{1,L} \right)
%\end{align}

Due to limited feedback, the quantized channel $\tilde {\bf{h}}_{\ell, i_{1}}$ is used in \eqref{eq:opt_prob_general} instead of perfect channel ${\bf{h}}_{\ell, i_{1}}$, whose accuracy is determined by the feedback amount. 
Now we propose a heuristic feedback design method applicable in the general number of antennas case. Note that $B_{i_1} > 0$ in the general number of antennas case while we suppose $B_{i_1} = 0$ in our previous assumption. 
%We first consider the problem by separating $B_{i_1}$, i.e., the feedback for the desired channel, and $\bar B_{\rm other} = \{B_{i_2},...,B_{i_L}\}$, i.e., the feedback for the intra-cluster interference channels. 
%The intuition behind this separation is that the usage of $B_{i_1}$ and $\bar B_{\rm other}$ is different. Specifically, $B_{i_1}$ is used for increasing the desired signal power, while $\bar B_{\rm other}$ is used for mitigating the intra-cluster interference. Due to this difference, jointly optimizing $B_{i_1}$ and $\bar B_{\rm other}$ is too complicated. For this reason, we rather alternatively determine $B_{i_1}$ and $\bar B_{\rm other}$ assuming that the other one is given. 
%It is straightforward to see that the feedback for the desired channel is used for increasing the desired signal power, while the feedback for the intra-cluster interference channels is for mitigating the interference. 
%Key idea of the proposed method is to determine $B_{i_1}$ and $\{B_{i_2}, ..., B_{i_L}\}$ alternatively by assuming the the other feedback is given. 
First, we assume that  $\tilde B_{\rm total} = B_{\rm total} - B_{i_1} = \sum_{\ell=2}^{L} B_{i_\ell}$ is given. 
Then, the feedback $B_{i_2},..., B_{i_L}$ can be determined by solving the following problem, which is modified from the optimization problem \eqref{opt:log}
\begin{align} \label{opt:log_general}
& \mathop{\rm minimize}_{B_{i_{\ell}} \in \mathbb{R}^+, \ell\in \{2,...,L\}}\; \sum_{\ell=2}^{L} \ln \left({1+\gamma  \delta_{1,\ell} 2^{-\frac{B_{i_\ell}}{N_{\pi(i_\ell)}-1}} }  \right),  \nonumber \\
&\;\;\;\;\;{\rm subject\;to} \; \sum_{\ell=2}^{L} B_{i_{\ell}} \le \tilde B_{\rm total}.
\end{align}
The corresponding Lagrangian function is
\begin{align}
L(\bar B_{\rm other}, \mu) = \sum_{\ell=2}^{L} \ln \left( 1 + \gamma \delta_{1, \ell}.  2^{-\frac{B_{i_{\ell}}}{N_{\pi(i_{\ell})} - 1}}\right) + \mu\left(\sum_{\ell=2}^{L}B_{i_{\ell}} - \tilde B_{\rm total} \right).
\end{align}
By the KKT condition, we have
\begin{align} \label{eq:KKT_general}
\frac{\gamma \delta_{1,\ell} \ln (2)}{\left(\gamma \delta_{1, \ell} + 2^{\frac{B_{i_\ell}}{N_{\pi(i_\ell)} - 1}}\right)(N_{\pi(i_\ell)}-1)} = \mu
\end{align}
and $\sum_{\ell = 2}^{L} B_{i_{\ell}} = \tilde B_{\rm total}$. Solving \eqref{eq:KKT_general}, we have
\begin{align} \label{eq:sol_general}
\left( B_{i_{\ell}}^\star \right)_{\rm C,gen} = (N_{\pi(i_\ell)}-1) \log_2\left(\gamma \delta_{1,\ell} \right) + (N_{\pi(i_{\ell})}-1) \log_2\left(\frac{\ln(2)}{\mu (N_{\pi(i_{\ell})}-1)} -1\right).
\end{align}
The parameter $\mu$ is determined as the minimum value that satisfies $\sum_{\ell=2}^{L} B_{i_{\ell}}^\star \le \tilde B_{\rm total}$. 
Note from \eqref{eq:sol_general} is that the SIR threshold $\gamma$ remains in the solution. In the previous case, the parameter $\gamma$ vanishes during the optimization process, so that the solution \eqref{scheme:bit_allo_inproof} is optimal for all the SIR threshold. This implies that the solution \eqref{scheme:bit_allo_inproof} is also optimal for the ergodic spectral efficiency. 
On the contrary, in the general number of antennas case, the obtained solution \eqref{eq:sol_general} is optimal for a particular SIR threshold, not for the ergodic spectral efficiency. 
It might be more straightforward to directly optimize the ergodic spectral efficiency \eqref{eq:rate_integral}. Unfortunately, this is infeasible since the ergodic spectral efficiency is an complicated integral form as shown in \eqref{eq:rate_integral}. As an alternative, we select an appropriate value of $\gamma$ by using simulations. For example, with \eqref{eq:sol_general}, we examine various values of $\gamma$ and then select a proper value of $\gamma$ that provides the maximum ergodic spectral efficiency. 
%Unlike the previous case that all the intra-cluster BSs use $L$ antennas, 

The next step is determining $B_{i_1}$. We use a line search method relying on simulations. Specifically, we first assume $B_{i_1} = 0$. Then, we allocate feedback $B_{i_2}, ..., B_{i_L}$ using \eqref{eq:sol_general}. Then we iterate this process by increasing $B_{i_1}$. After searching over a whole region, i.e., $0 \le B_{i_1} \le B_{\rm total}$, we select $B_{i_1}^{\star}$ that provides the maximum ergodic spectral efficiency. We summarize the whole feedback design procedure in the following proposition.
\begin{proposition} \label{prop:scheme_general} 
Assume the general number of antennas case, where the intra-cluster BSs have different number of antennas $N_{\pi(i_1)},...,N_{\pi(i_L)}$. Then, a heuristic way to allocate feedback is as follows. 
\begin{enumerate}
	\item Assume $B_{i_1} = 0$.
	\item With $\tilde B_{\rm total} = B_{\rm total} - B_{i_1}$, allocate feedback by using \eqref{eq:sol_general}. The parameter $\gamma$ is selected so as to provide the maximum ergodic spectral efficiency. 
	\item Iterate 1) and 2) by increasing $B_{i_1} $ until $B_{i_1} \le B_{\rm total}$.
	\item Select $B_{i_1}^{\star},...,B_{i_L}^{\star}$ that provides the maximum ergodic spectral efficiency. 
\end{enumerate}
%Assuming a general antenna case, we allocate feedback to each BS as in the following method. 
%Denoting the ergodic spectral efficiency of $R_{\rm C, general}^{m} \left(\beta, \bar \lambda_K,\bar N_K, \bar B_{{L}}, \bar P_K, \bar S_K, \bar \delta_{1,L} \right)$, we 
\end{proposition}
We will evaluate Proposition \ref{prop:scheme_general} in the later section. 

\subsection{A Single-Tier Case}
In this subsection, we consider a case where $\left| \CMcal{K} \right| = 1$, i.e., when coordination is applied in a single-tier network. 
For ease of understanding, we enumerate the features of the single-tier case as follows. 
First, the coordination set $\CMcal{C} = \{i_1, ..., i_L\}$ boils down to $\CMcal{C} = \{1,... , L\}$. Specifically, the typical user is connected to the $L$ closest BSs located at ${\bf{d}}_1, ..., {\bf{d}}_L$ since in a single-tier network, the $L$ BSs whose biased signal powers are strongest is the same as the $L$ closest BSs to the typical user. 
%For this reason, the feedback partition set is also $\bar B_{L} = \{B_2, ..., B_L \}$. 
Second, $\delta_{1, \ell}$ simplifies to $\left\| {\bf{d}}_{\ell}\right\|^{-\beta} / \left\| {\bf{d}}_1 \right\|^{-\beta}$ since all the BSs use the same transmit power.
Third, the biasing factor is neglected since it is only useful in a HetNet scenario. 
We note that Proposition \ref{prop:scheme1} is general for the number of the tiers $\left| \CMcal{K}\right|$, whereby it is applied without any modification for a case of a single-tier network, i.e., $\left| \CMcal{K} \right| = 1$. 
In this single-tier network, we present an approximate feedback partition that does not need $\bar \delta_{1, L}$. 
%, which is more robust than Proposition \ref{prop:scheme1}, so that
%that is even useful when $\bar \delta_{1, L}$ is not available.
\begin{proposition}  \label{prop:scheme2}
Assume $\left| \CMcal{K} \right| = 1$ and the intra-cluster geometry $\bar \delta_{1,L}$ is unknown. In that case, an approximate feedback partition as an alternative of \eqref{scheme1_coop} is 
\begin{align} \label{eq:scheme_allo2}
(\tilde B_{{\ell}}^{\star})_{\rm C} \mathop{=}^{} \frac{ B_{\rm total}}{L-1}  - \frac{\beta (L-1)}{2\ln 2} H_{\ell-1}
+ \frac{\beta}{2\ln 2} \sum_{\ell=2}^{L} H_{\ell-1},
\end{align}
where $H_{\ell}$ is the $\ell$-th harmonic number defined as $H_\ell = \sum_{i=1}^{\ell}\frac{1}{i}$. 
\end{proposition}
\begin{proof}
If $\bar \delta_{1,L}$ is unknown, a possible alternative is taking the expectation to \eqref{scheme1_coop} with regard to $\bar \delta_{1, L}$. To calculate this, we use the probability density function (PDF) of $\frac{\left\| {\bf{d}}_1 \right\|}{ \left\| {\bf{d}}_{\ell} \right\|}$, presented in Lemma 1 of \cite{lee:twc:15} as
\begin{align} \label{eq:pdf_ratio}
f_{\frac{\left\| {\bf{d}}_1 \right\|}{ \left\| {\bf{d}}_{\ell} \right\|}}(x) = 2(\ell -1)x(1-x^2)^{\ell-2}.
\end{align}
Since $\delta_{1,\ell} = \left(\frac{\left\|{\bf{d}}_1 \right\|}{ \left\| {\bf{d}}_{\ell}\right\| }\right)^{\beta}$, we compute the following by exploiting \eqref{eq:pdf_ratio}.
\begin{align} \label{eq:not_fix_coop}
\mathbb{E}_{}\left[(B_{{\ell}}^{\star})_{\rm C} \right] &= (\tilde B_{{\ell}}^{\star})_{\rm C}  \nonumber \\
&= \frac{ B_{\rm total}}{L-1}  + \beta (L-1) \mathbb{E}\left[ \log_2\left( {\delta_{1,\ell}}\right)^{\frac{1}{\beta}} \right]- \beta \sum_{\ell=2}^{L} \mathbb{E} \left[ \log_2\left(\delta_{1,\ell}\right)^{\frac{1}{\beta}} \right] \nonumber \\
&\mathop{=}^{(a)} \frac{ B_{\rm total}}{L-1}  - \frac{\beta (L-1)}{2\ln 2} H_{\ell-1}
+ \frac{\beta}{2\ln 2} \sum_{\ell=2}^{L} H_{\ell-1},
\end{align}
where (a) follows that $\mathbb{E}\left[\log_2\left(\frac{\left\| {\bf{d}}_1\right\|}{\left\| {\bf{d}}_{\ell} \right\|} \right) \right] = -\frac{1}{2 \ln 2} H_{\ell}$ with $H_{\ell} = \frac{1}{1} + \frac{1}{2} + \cdots + \frac{1}{\ell}$. Since $\sum_{\ell=2}^{L} (\tilde B_{{\ell}}^{\star})_{\rm C} = B_{\rm total}$, \eqref{eq:scheme_allo2} is a feasible feedback partition. This completes the proof.
\end{proof} 
\begin{remark} \normalfont
%A key point of Proposition \ref{prop:scheme2} is that it does not need the relative intra-cluster BS power $\bar \delta_{1, L}$. 
The feedback partition in Proposition \ref{prop:scheme2} is only a function of the path-loss exponent $\beta$ and the index of the intra-cluster BS ${\ell}$, so that no instantaneous SIR is needed to be measured. 
%it can be applied without even measuring the relative intra-cluster interference power. 
% and sharing the measured distance to each intra-cluster BSs. 
%For this reason, Proposition \ref{prop:scheme2} is
%For this reason, Proposition \ref{prop:scheme2} is more robust than Proposition \ref{prop:scheme1}. 
%Thanks to this feature, all the users served by BS coordination are able to use the same size codebook \eqref{eq:scheme_allo2} regardless of their individual conditions.
Similar to Proposition \ref{prop:scheme1}, Proposition \ref{prop:scheme2} also implies that allocating more feedback to closer intra-cluster BSs is beneficial since $(\tilde B_{\ell_1}^{\star} )_{\rm C} \ge (\tilde B_{\ell_2}^{\star})_{\rm C} $ if $\ell_1 < \ell_2$. 
\end{remark}

Now we investigate the relationship between the effective BS coordination set size and the total feedback $B_{\rm total}$.  
This is important because if the cluster size is too large compared to the total feedback $B_{\rm total}$, some intra-cluster BSs (particularly far BSs) are not allocated enough feedback. Those BSs only increase overheads associated with channel estimation, and are not helpful in mitigating the interference due to the lack of the feedback. 
Unfortunately, this is hard to be analyzed because the existing feedback partition is a function of particular intra-cluster BS power, so that a general relationship cannot be extracted. To resolve this, we use Proposition \ref{prop:scheme2}, which is independent to particular intra-cluster BS power. 
%so that it is important to determine an appropriate BS coordination set size when $B_{\rm total}$ is given. 
As a stepping stone to reveal the relationship, we first define the effective cluster size $L_{\rm eff}$ as the BS's index that satisfies 
%Among those BSs, we call the closest BS's index as the effective cluster size, denoted as $L_{\rm eff}$. The effective cluster size satisfies
\begin{align} \label{eq:effectivesize}
( \tilde B_{L_{\rm eff}}^{\star} )_{\rm C} \le 1,
\end{align}
which means that the BS's index whose allocated feedback is less than $1$. 
Recalling that Proposition \ref{prop:scheme2} gives fewer feedback to a further BS, the $\ell$-th closest BS for $\ell > L_{\rm eff}$ also has less than $1$-bit feedback. For this reason, increasing the cluster size over $L_{\rm eff}$ does not provide spectral efficiency gain. 
In the following corollary, with Proposition \ref{prop:scheme2}, we investigate the scaling behavior of $L_{\rm eff}$ depending on the total feedback $B_{\rm total}$ and the path-loss exponent $\beta$ 

%when Proposition \ref{prop:scheme2} is applied, how the effective cluster size $L_{\rm eff}$ scales with the system parameters as $B_{\rm total}$ increases.
%when the total feedback bits $B_{\rm total}$ becomes large.
%\begin{definition}
%The 
%\begin{align}
%( \tilde B_{L}^{\star} )_{\rm coop.} \le 1
%\end{align}
%\end{definition}

\begin{corollary} \label{coro:effectivesize}
Assume that Proposition \ref{prop:scheme2} is used. Under this assumption, if the total feedback $B_{\rm total}$ is large enough, the effective cluster size $L_{\rm eff}$ in Proposition \ref{prop:scheme2} is
\begin{align}
L_{\rm eff} &\ge \frac{ \ln 2}{\beta} \sqrt{\frac{2 \beta}{\ln 2} B_{\rm total} + 1} - \frac{\ln 2 }{\beta} + 1 \nonumber \\
%& \gtrsim { \sqrt{2 \ln 2}} \frac{\sqrt{B_{\rm total}}}{\sqrt{\beta}} \nonumber \\
& \gtrsim \frac{\sqrt{B_{\rm total}}}{\sqrt{\beta}}.
\end{align}
\end{corollary}
\begin{proof}
Since the effective cluster size satisfies \eqref{eq:effectivesize}, we have
\begin{align} \label{eq:effective_1}
&( \tilde B_{L_{\rm eff}}^{\star} )_{\rm C} \le 1 \nonumber \\
\Leftrightarrow &\frac{B_{\rm total}}{L_{\rm eff}-1} - \frac{\beta(L_{\rm eff}-1)}{2\ln 2} H_{L_{\rm eff}-1} + \frac{\beta}{2 \ln 2} \sum_{\ell = 2}^{L_{\rm eff}} H_{\ell-1} \le 1.
\end{align}
The harmonic number $H_{\ell}$ is tightly approximated as $H_{\ell} \simeq \gamma_{\rm Euler} + \ln \ell$, where $\gamma_{\rm Euler}$ is the Euler-Mascheroni constant. With this approximation, \eqref{eq:effective_1} is
\begin{align} \label{eq:effective_2}
\frac{B_{\rm total}}{L_{\rm eff}-1} - \frac{\beta (L_{\rm eff}-1)}{2 \ln 2} (\gamma_{\rm Euler} + \ln (L_{\rm eff}-1)) + \frac{\beta}{2 \ln 2} \sum_{\ell = 2}^{L_{\rm eff}} (\gamma_{\rm Euler} + \ln (\ell-1)) \le 1 
%\nonumber \\
\end{align}
\begin{align}
\Leftrightarrow & \frac{B_{\rm total}}{L_{\rm eff}-1} - 1 \le \frac{\beta}{2 \ln 2} \ln \left( \frac{(L_{\rm eff}-1)^{(L_{\rm eff}-1)}}{(L_{\rm eff}-1) !}\right).
\end{align}
By using the Stirling's approximation, we have
\begin{align}
\ln (L_{\rm eff}-1)^{(L_{\rm eff}-1)} - \ln((L_{\rm eff}-1)!) + O(\ln (L_{\rm eff}-1))
= L_{\rm eff}-1,
 \end{align}
where $O(\cdot)$ is defined as follows: $f(x) = O(g(x))$ as $x \rightarrow a$ if and only of $\lim \sup_{x \rightarrow a} \left| f(x)/g(x)\right| < \infty$. 
 If $B_{\rm total}$ is large enough, the approximation is tight. Then \eqref{eq:effective_2} is 
\begin{align}
0 \le \frac{\beta}{2 \ln 2} (L_{\rm eff}-1)^2 + (L_{\rm eff}-1) - B_{\rm total}.
\end{align}
Solving the quadratic equation, $L_{\rm eff}$ is
\begin{align} \label{eq:prop3_approx}
%L_{\rm eff}-1 \ge \frac{-1+\sqrt{ 4 B_{\rm total} \frac{\beta}{2 \ln 2}+1}}{\frac{\beta}{2 \ln 2}}
L_{\rm eff} \ge \frac{ \ln 2}{\beta} \sqrt{\frac{2 \beta}{\ln 2} B_{\rm total} + 1} - \frac{\ln 2 }{\beta} + 1.
\end{align}
For large enough $B_{\rm total}$, \eqref{eq:prop3_approx} is further approximated as
\begin{align}
L_{\rm eff} &
%\ge \frac{2 \ln 2}{\beta} \sqrt{\frac{2 \beta}{\ln 2} B_{\rm total} + 1} - \frac{2\ln 2 }{\beta} + 1 \nonumber \\
 \gtrsim { \sqrt{2 \ln 2}} \frac{\sqrt{B_{\rm total}}}{\sqrt{\beta}}  \ge \frac{\sqrt{B_{\rm total}}}{\sqrt{\beta}}.
\end{align}
This completes the proof. 
\end{proof}

\begin{remark} \normalfont
A major finding in Corollary \ref{coro:effectivesize} is that the effective cluster size $L_{\rm eff}$ scales with the square root of the total feedback $B_{\rm total}$, and inversely with the square root of the path-loss exponent $\beta$. 
%It is worthwhile to note that these are new findings. 
%This sheds light on that how to determine $L_{\rm eff}$ depending on $B_{\rm total}$. 
This provides a system guideline on how to determine $L_{\rm eff}$ depending on $B_{\rm total}$ and $\beta$. The intuition of the former relationship ($L_{\rm eff} \propto \sqrt{B_{\rm total}}$) is natural since the larger effective cluster size is supported as the total feedback increases.
The rationale of the latter relationship ($L_{\rm eff} \propto \frac{1}{\sqrt{\beta}}$) is as follows. When the path-loss exponent increases, the interference power coming from far BSs decays fast, so that there is no need to allocate feedback to those BSs.  This shrinks the effective cluster size. 
\end{remark}

%It is worthwhile to note that the effective cluster size $L_{\rm eff}$ also can be obtained in Proposition \ref{prop:scheme1}, but it mainly depends on the particular intra-cluster geometry $\bar \delta_{1, L}$. For this reason, it is hard to find an intuitive relationship between $L_{\rm eff}$ and $B_{\rm total}$ or $\beta$ as in Corollary \ref{coro:effectivesize}.

%\begin{remark} \normalfont
%We assume that the number of antennas and the density are inversely proportional, i.e., $N_n \ge N_m$ if $\lambda_n \le \lambda_m$. This is reasonable since as the size of a BS increases, it is easy to have many antennas, while dense deployment is difficult. For example, a macro BS is appropriate to have many antennas than a femto BS, but the density of a femto BS is likely to be larger than that of a macro BS.
%Under this assumption, the allocated feedback bits is proportional to the number of antennas, i.e., inversely proportional to the density. When a BS has many antennas, the gain from the feedback is large since 
%, so that more transmit antennas are likely to be equipped while its deployment is not easy 
%\end{remark}

\section{Simulation Results}
In this section, we provide simulation results to show the spectral efficiency improvement by using our feedback partitions presented in Proposition \ref{prop:het}, \ref{prop:scheme1}, and \ref{prop:scheme2}. 

\begin{figure}[t]
\centering
$\begin{array}{cc}
{\resizebox{0.46\columnwidth}{!}
{\includegraphics{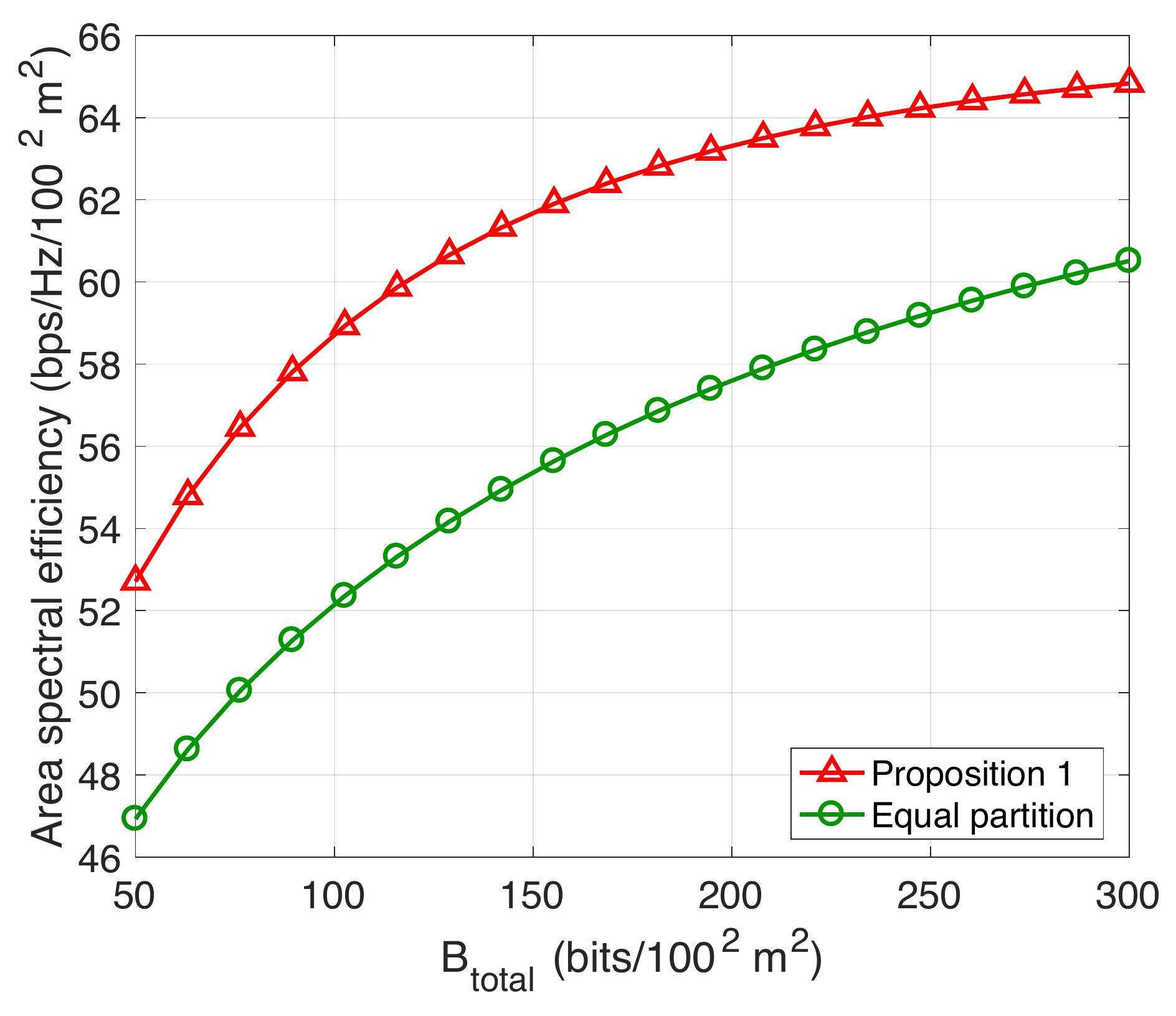}}}  &
{\resizebox{0.48\columnwidth}{!}
{\includegraphics{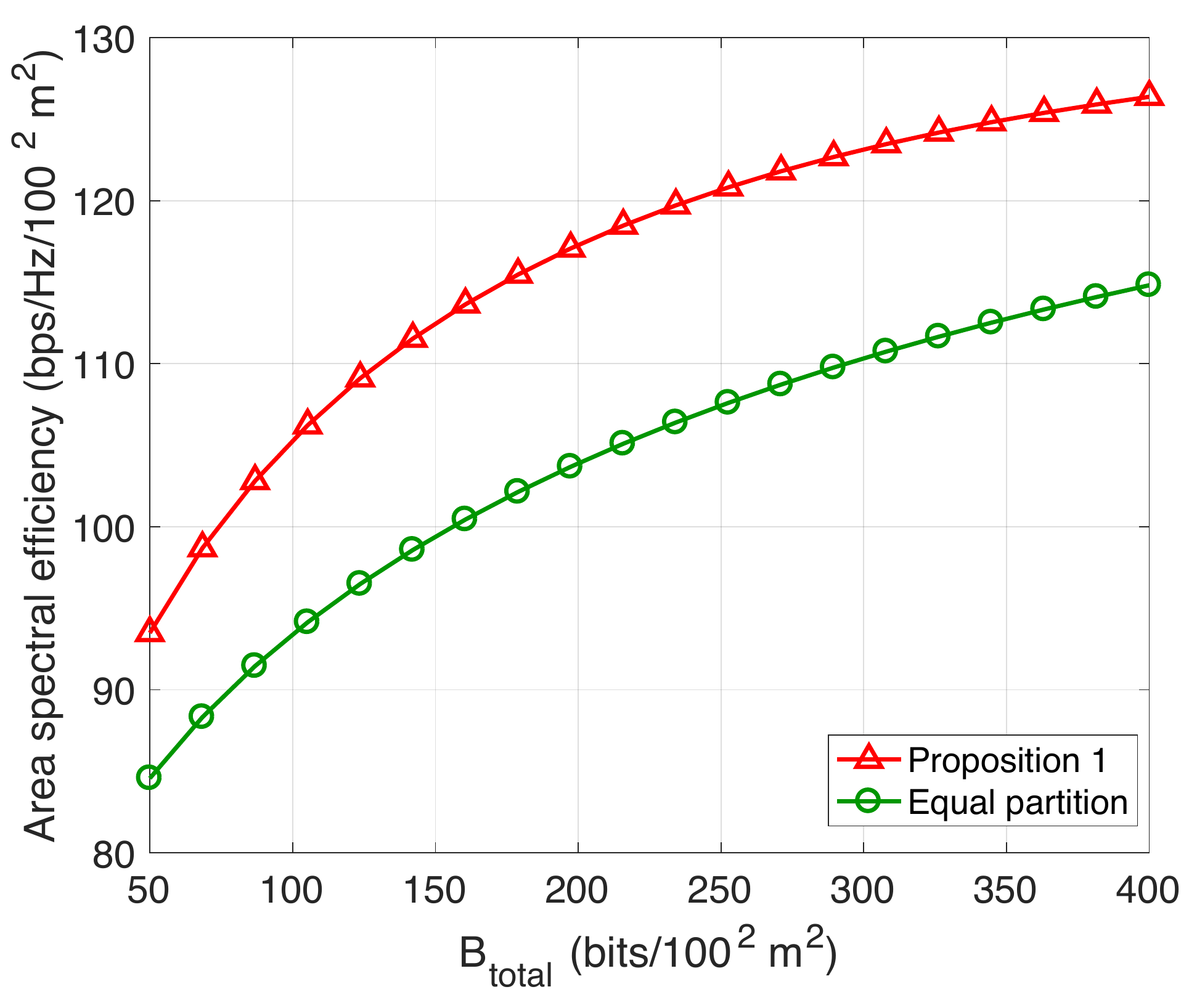}}}  \\ 
\mbox{(a)} &
\mbox{(b)} 
\end{array}$
\caption{An area spectral efficiency comparison in a non-cooperative HetNet. A $3$-tier HetNet is assumed. The simulation parameters are as follows. In (a), $\bar N_K = \{8,6,6\}$, $\bar \lambda_K = \{0.5\lambda_{\rm ref} , 5\lambda_{\rm ref} , 40\lambda_{\rm ref} \}$ where $\lambda_{\rm ref} = 10^{-4}/\pi$, $\bar P_K = \{20 , 15, 10 \}{\rm dBm}$, $\bar S_K = \{0, 3, 5\}{\rm dB}$, and $\beta = 4$. In (b), the other parameters are same with (a) except $\bar \lambda_K = \{0.5\lambda_{\rm ref} , 10\lambda_{\rm ref} , 80\lambda_{\rm ref} \}$. 
%$\bar N_K = \{6,4,2\}$, $\bar \lambda_K = \{10^{-5}/\pi, 5\lambda_1, 10\lambda_1\}$, $\bar P_K = \{30, 20, 10\}{\rm dBm}$, and $\bar S_K = \{1, 5, 10\}{\rm dB}$. The path-loss exponent is $\beta = 4$ and the unit area is $100^2 {\rm m}^2$ in both cases. 
%Area spectral efficiency comparison in a heterogeneous network. In (a), it is assumed that {N1,N2,N3} = {4,3,2}, {λ1,λ2,λ3} = {10−5/π,5λ1,20λ1}, {P1,P2,P3} = {20dBm,5dBm,2dBm}, {S1,S2,S3} = {1dBm,3dBm,5dBm}. In (b), it is assumed that {N1, N2, N3} = {6, 4, 2}, {λ1, λ2, λ3} = {10−5/π, 5λ1, 10λ1}, {P1, P2, P3} = {30dBm,20dBm,10dBm}, {S1,S2,S3} = {1dBm,5dBm,10dBm}. In both cases, the path-loss exponent is β = 4 and the unit area is 1002m2.
}
   \label{fig:noncoopt1}
\end{figure} 

First, we assume the non-cooperative case. The area spectral efficiency comparison between Proposition \ref{prop:het} and the baseline method is depicted in Fig.~\ref{fig:noncoopt1}. We note that the area spectral efficiency in Fig.~\ref{fig:noncoopt1} is generated by using an exact expression \eqref{eq:rate_het}, not a lower bound \eqref{eq:lower_bound_noncoopt}.
The baseline method is the per-tier equal partition, where the feedback of the $k$-tier is determined as $B_k = B_{\rm total}/K/\lambda_k$, so that the feedback consumed in each tier is same each other, i.e., $\lambda_1 B_1 = ... =  \lambda_K B_K = B_{\rm total}/K$. 
%We use the floor function to both of the baseline and Proposition \ref{prop:het} in the simulation. 
The other system parameters assumed in the simulations are described in the caption of Fig.~\ref{fig:noncoopt1}. The main difference between Fig.~\ref{fig:noncoopt1}-(a) and (b) is the densities, where Fig.~\ref{fig:noncoopt1}-(b) assumes more dense HetNets.  
As shown in Fig.~\ref{fig:noncoopt1}, Proposition \ref{prop:het} increases the area spectral efficiency by $11.3\%$ in Fig.~\ref{fig:noncoopt1}-(a) and by $12\%$ in Fig.~\ref{fig:noncoopt1}-(b).
%In Fig.~\ref{fig:noncoopt2}, Proposition \ref{prop:het} increases the area spectral efficiency by $11.3\%$ in (a) and $12.3\%$ in (b). 
%Specifically, assuming that the operating bandwidth $W = 20 {\rm MHz}$ and $B_{\rm total} = 50$, Proposition \ref{prop:het} obtains $12 {\rm Mbps}/100^2 {\rm m}^2$ rate gain in Fig. 5-(a) and $2 {\rm Mbps}/100^2{\rm m}^2 $ rate gain in Fig. 5-(b) compared to the baseline method. 
In both cases, we observe that Proposition \ref{prop:het} provides the meaningful gains compared to the equal partition. We expect that more gains can be obtained when a HetNet becomes dense, i.e., if $\lambda_3 \gg \lambda_1$. This is because, assuming that $\lambda_3$ is large, only a few amount of feedback is allocated to the third tier in the equal partition. This causes that the BSs in the third tier only provides marginal spectral efficiency due to lack of accurate CSIT. On the contrary, in Proposition \ref{prop:het}, the appropriate amount of feedback can be allocated to the third tier even when $\lambda_3$ is large, leading to the spectral efficiency gain. 
As described in Remark \ref{remark:water}, Proposition \ref{prop:het} does not depend on instantaneous SIR, so that there is no need to change the allocated feedback frequently depending on instantaneous SIR. For this reason, the complexity of Proposition \ref{prop:het} is almost equivalent to that of the equal partition. 

\begin{figure}[t]
\centering
$\begin{array}{cc}
{\resizebox{0.49\columnwidth}{!}
{\includegraphics{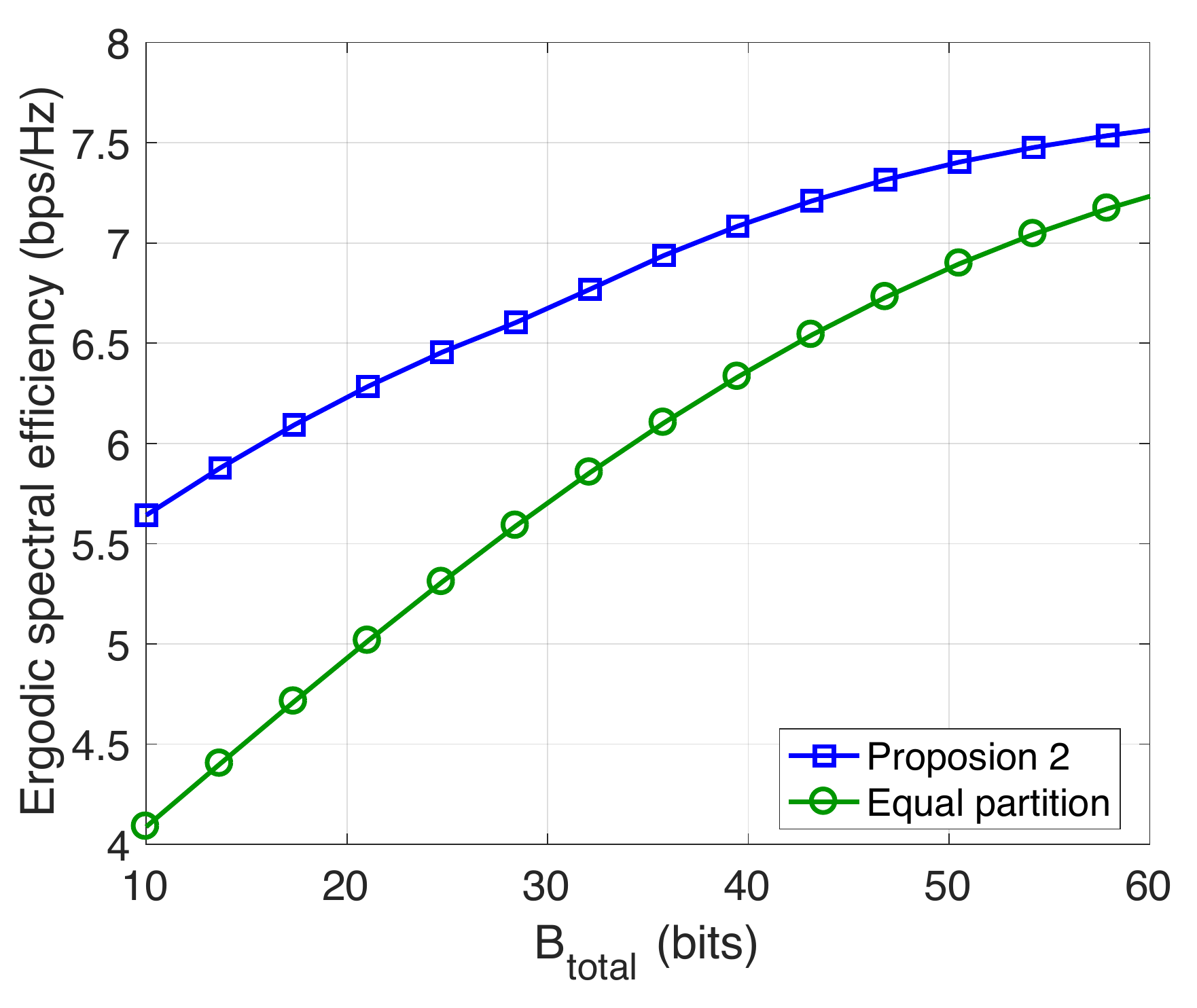}}}  &
{\resizebox{0.47\columnwidth}{!}
{\includegraphics{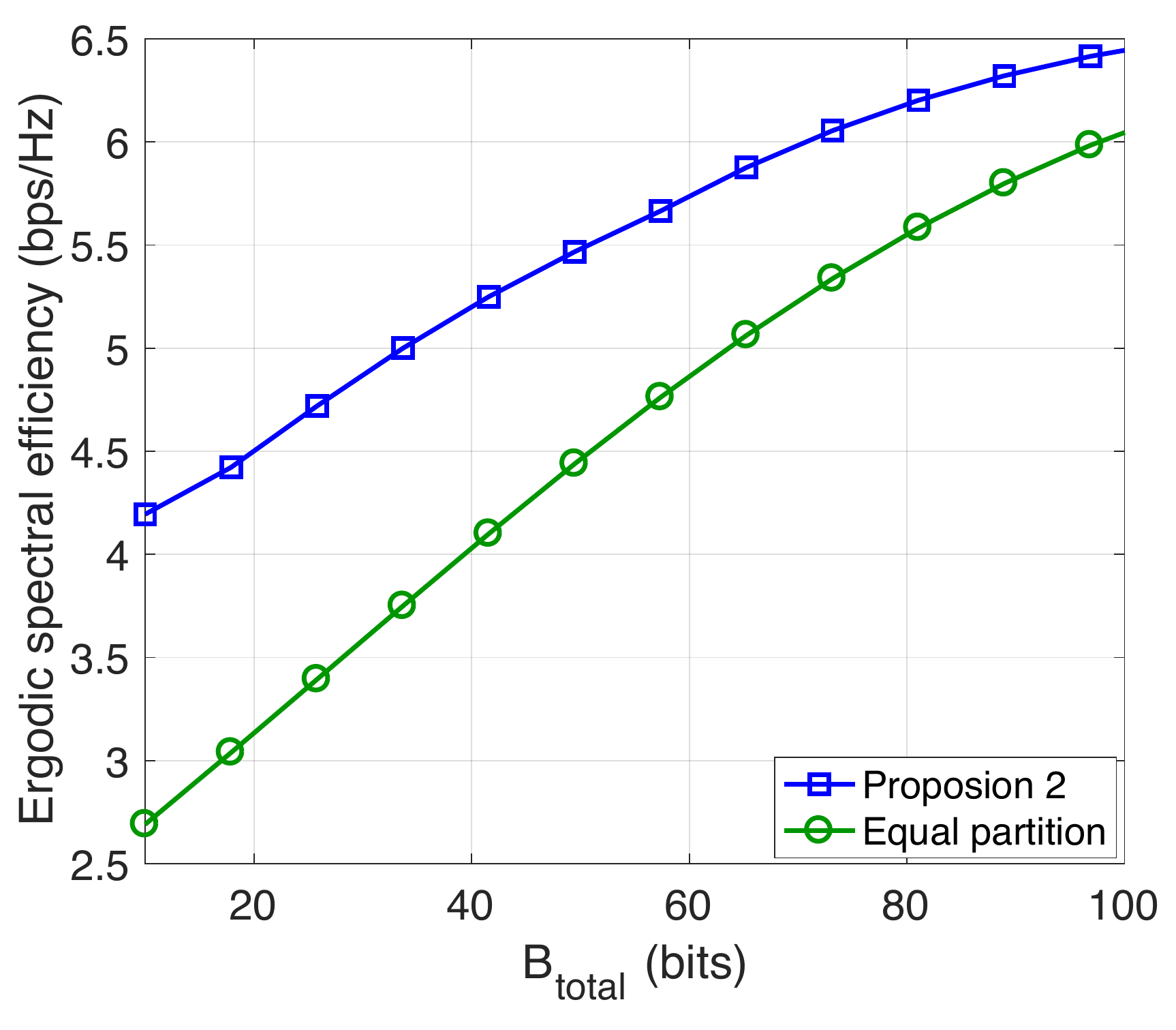}}}  \\ 
\mbox{(a)} &
\mbox{(b)} 
\end{array}$
\caption{The ergodic spectral efficiency comparison in a cooperative HetNet. A $3$-tier HetNet is assumed. The simulation parameters are as follows: In (a), $\bar N_K = \{L,L,L\}$, $L = 4$, $\bar \lambda_K = \{1\lambda_{\rm ref} , 10\lambda_{\rm ref} , 20\lambda_{\rm ref} \}$ where $\lambda_{\rm ref} = 10^{-4}/\pi$, $\bar P_K = \{20 , 15, 10 \}{\rm dBm}$, $\bar S_K = \{0, 3, 5\}{\rm dB}$, $\bar\delta_{1,L} = \{0.1, 0.01, 0.001\}$, $\pi(i_1)=1$, $\pi(i_L)=2$, and $\beta = 4$. In (b), the other parameters are same except that $L = 5$ and $\bar \delta_{1, L} = \{0.2, 0.04, 0.008, 0.0016\}$.
}
   \label{fig:coopt}
\end{figure} 

Next, we assume the cooperative case. We compare the ergodic spectral efficiency of Proposition \ref{prop:scheme1} and the baseline method in Fig.~\ref{fig:coopt}, whose caption includes the simulation setting. Similar to the non-cooperative case, the baseline method is the equal partition, where the total feedback is equally partitioned to each of intra-cluster BS, i.e., $B_2=..=B_L = B_{\rm total}/(L-1)$. 
In Fig.~\ref{fig:coopt}-(a), we have $38.2\%$ spectral efficiency gain by using Proposition \ref{prop:scheme1} at $B_{\rm total} = 10$, and in Fig.~\ref{fig:coopt}-(b), we have $56.2\%$ gain at $B_{\rm total} = 10$. We observe that Proposition \ref{prop:scheme1} provides more gains when 1) $L$ increases or 2) $B_{\rm total}$ decreases. This is because, when $L$ increases or $B_{\rm total}$ decreases, the equal partition allocates smaller amount of feedback to the strong BSs whose $\delta_{1, \ell}$ is large. Then, due to lack of sufficient feedback, the interference from those strong BSs is not mitigated well, resulting in significant spectral efficient loss. 
On the contrary, by using Proposition \ref{prop:scheme1}, the appropriate amount of feedback is allocated to each BS proportional to $\delta_{1, \ell}$, so that considerable spectral efficiency gain is obtained even when $L$ increases or $B_{\rm total}$ decreases.

\begin{figure}[t]
\centering
$\begin{array}{cc}
{\resizebox{0.48\columnwidth}{!}
{\includegraphics{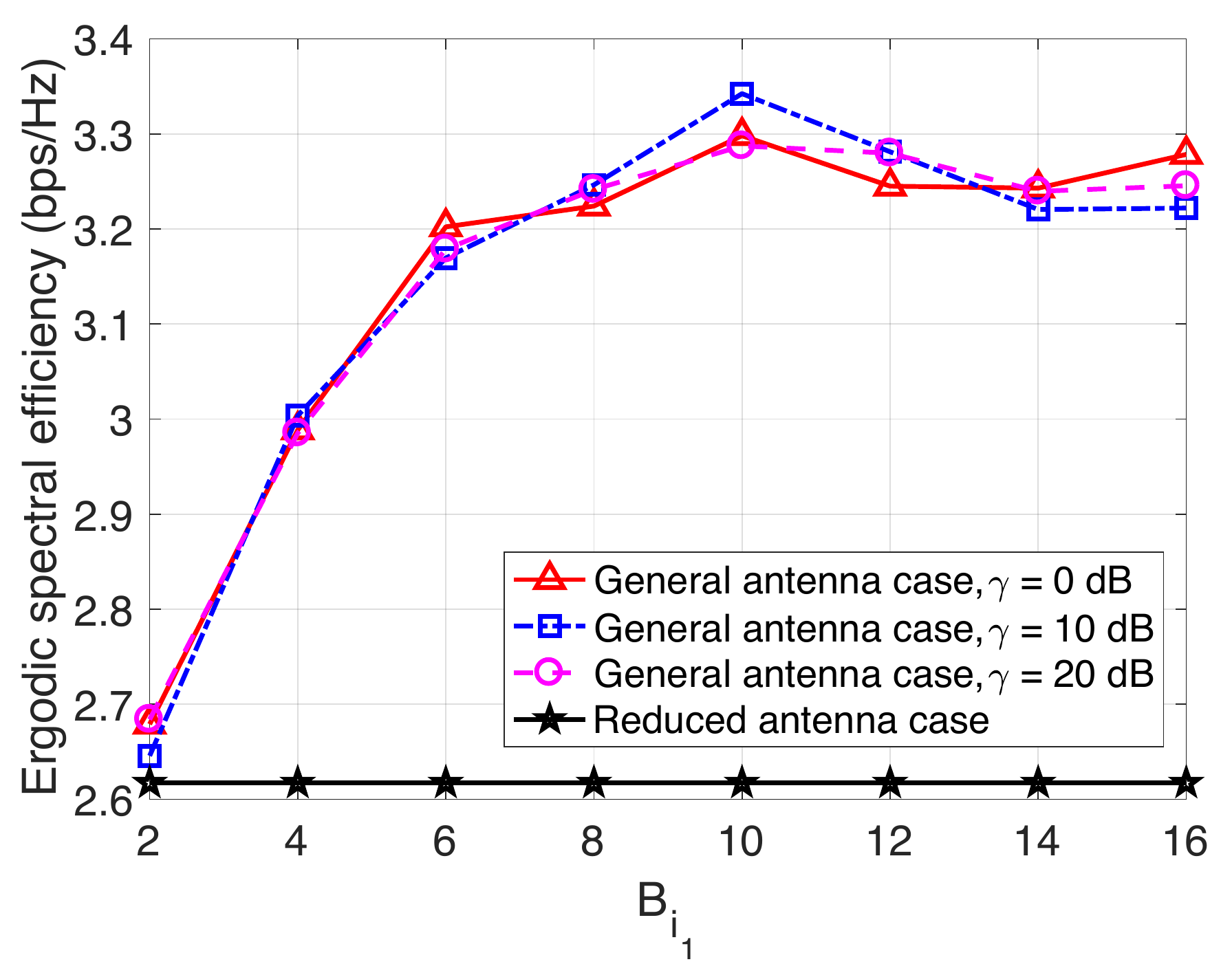}}}  &
{\resizebox{0.48\columnwidth}{!}
{\includegraphics{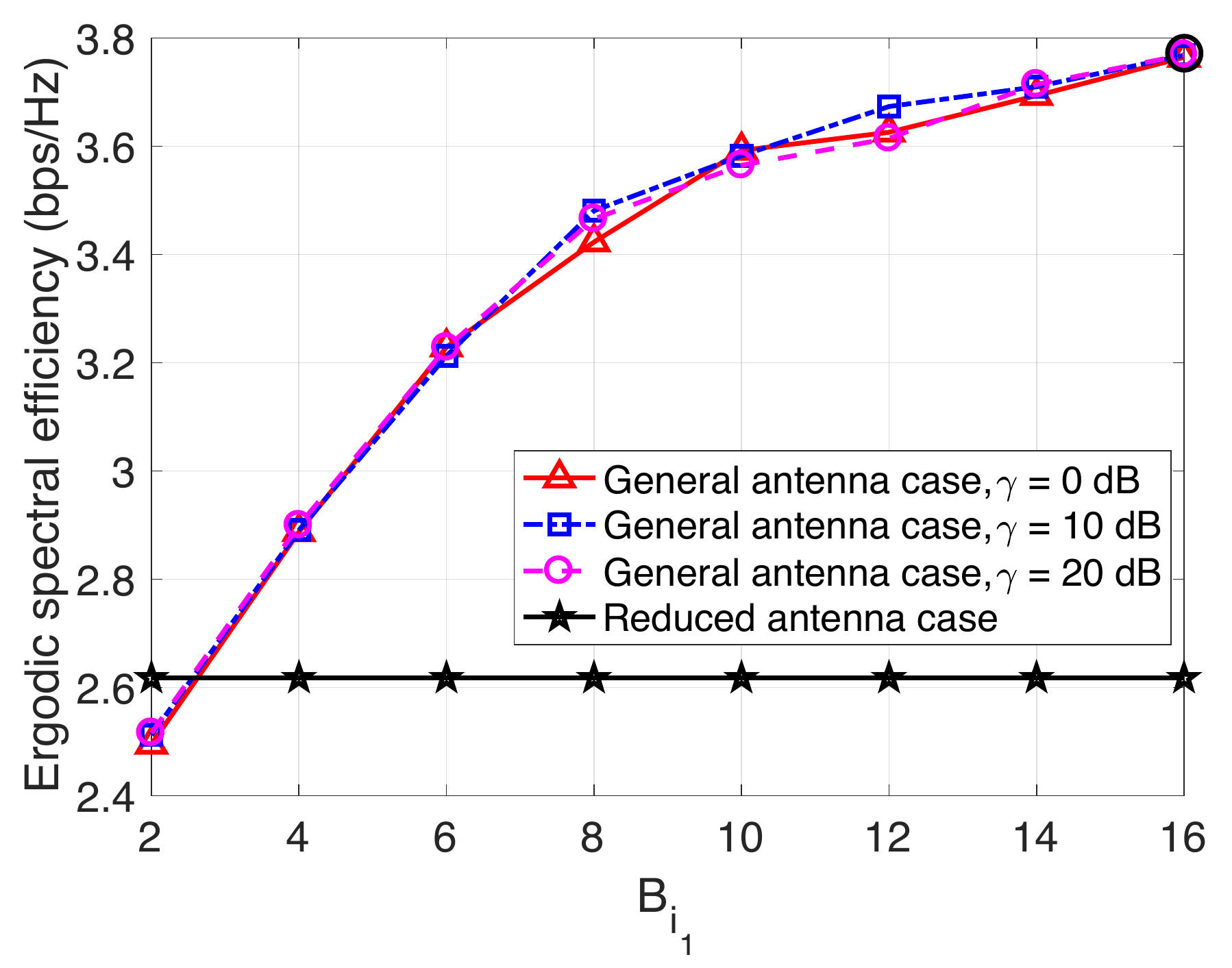}}}  \\ 
\mbox{(a)} &
\mbox{(b)} 
\end{array}$
\caption{The ergodic spectral efficiency comparison in the general number of antennas case. A $3$-tier HetNet is assumed. The simulation parameters are as follows: In (a), $\bar N_K = \{8, 6, 4\}$, $L = 4$, $\bar \lambda_K = \{1\lambda_{\rm ref} , 5\lambda_{\rm ref} , 20\lambda_{\rm ref} \}$ where $\lambda_{\rm ref} = 10^{-4}/\pi$, $ \bar P_K = \{20,15,10\}{\rm dBm}$, $\bar S_K = \{0, 3, 5\}{\rm dB}$, $\bar\delta_{1,L} = \{0.1, 0.01, 0.001\}$, $\pi(i_1) = 2$,  $\pi(i_2) = 3$, $\pi(i_3) = 2$, $\pi(i_4) = 1$ and $\beta = 4$. The total feedback is $B_{\rm total} = 16$. In (b), the other parameters are same except that $\pi(i_1) = 1$, $\pi(i_2) = 1$, $\pi(i_3) = 2$, $\pi(i_4) = 3$. In the reduced number of antennas case, we assume $\bar N_K = \{L, L, L\}$ and use Proposition \ref{prop:scheme1}.
}
   \label{fig:coopt_general}
\end{figure} 

Subsequently, we consider the general antenna case. The result is depicted in Fig.~\ref{fig:coopt_general} and the parameter setting is described in its caption. Since it is hard to obtain an analytical expression for the general antenna case, we rely on numerical simulation to produce results. In the simulation, we assume that the intra-cluster conditions are fixed, so that $\bar \delta_{1,L}$ and $N_{\pi(i_1)},...,N_{\pi(i_L)}$ are given. 
As shown in Fig.~\ref{fig:coopt_general}-(a), the ergodic spectral efficiency is maximized at $B_{i_1}^{\star} = 10$ with $\gamma = 10{\rm dB}$. The feedback for the other intra-clutser BSs is $\{B_{i_2}^{\star}, B_{i_3}^{\star}, B_{i_4}^{\star}\} = \{6,0,0\}$. Compared to the reduced number of antennas case, the general antenna case with Proposition \ref{prop:scheme_general} increases the ergodic spectral efficiency by $25\%$. 
In Fig.~\ref{fig:coopt_general}-(b), the ergodic spectral efficiency is maximized at $B_{i_1}^{\star} = 16$ and there is no observable difference in $\gamma$. Since $B_{i_1}^{\star} = B_{\rm total}$, no feedback is used for the other intra-cluster BSs. The performance gains of the general antenna case compared to the reduced number of antennas case is $41.2\%$. 
We point out that the main difference between Fig.~\ref{fig:coopt_general}-(a) and (b) is the number of antennas of the associated BS. Specifically, in Fig.~\ref{fig:coopt_general}-(a), the number of the associated BS is $6$, and $8$ in Fig.~\ref{fig:coopt_general}-(b). For this reason, the potential desired channel gain is large in (b), leading to allocate more feedback to the associated BS, i.e., $B_{i_1}$ increases. 
%The number of antennas of the other intra-cluster BSs also affects the feedback allocation. For example, in (b), 
%this more

\begin{figure}[t]
\centering
$\begin{array}{cc}
{\resizebox{0.49\columnwidth}{!}
{\includegraphics{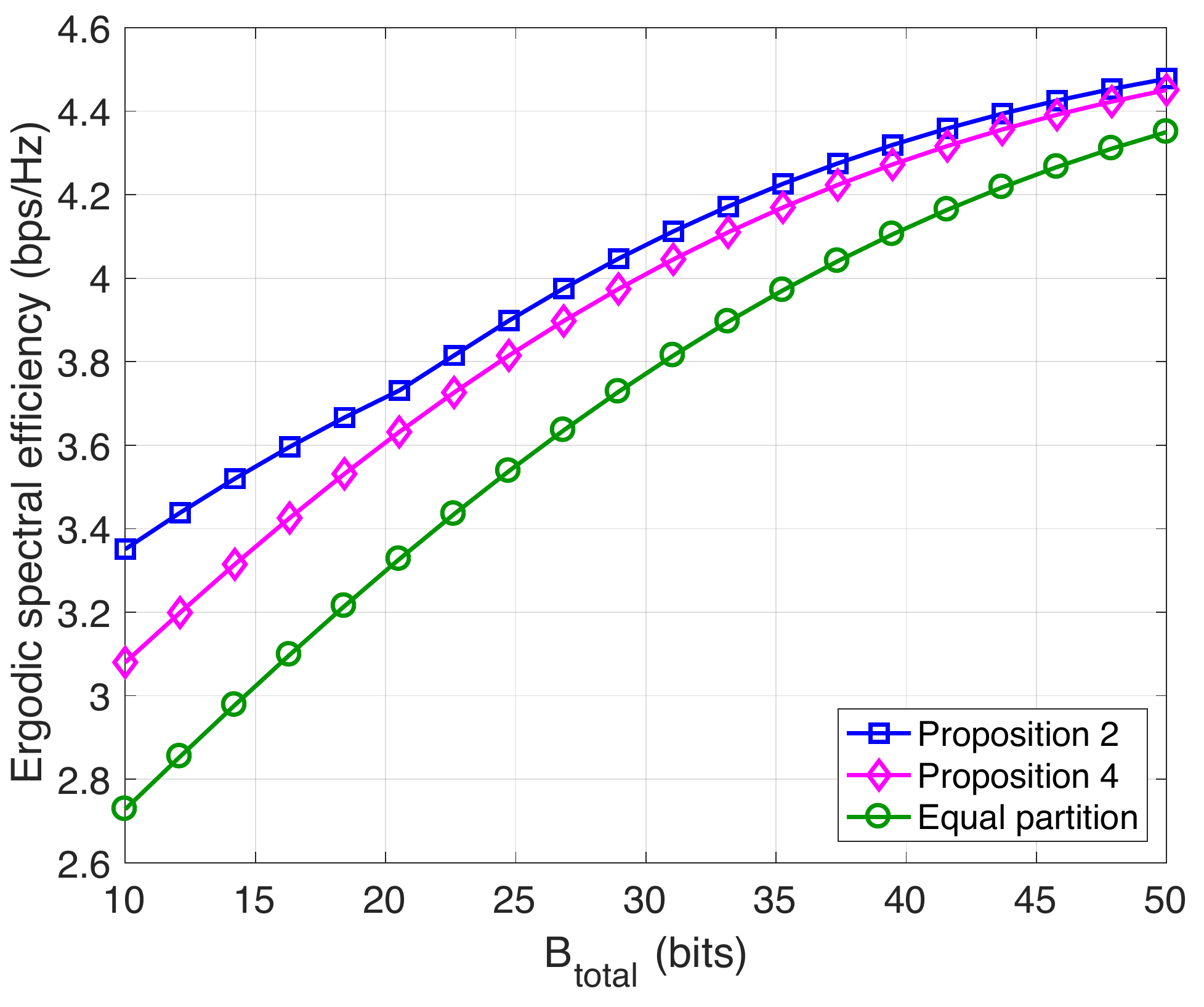}}}  &
{\resizebox{0.47\columnwidth}{!}
{\includegraphics{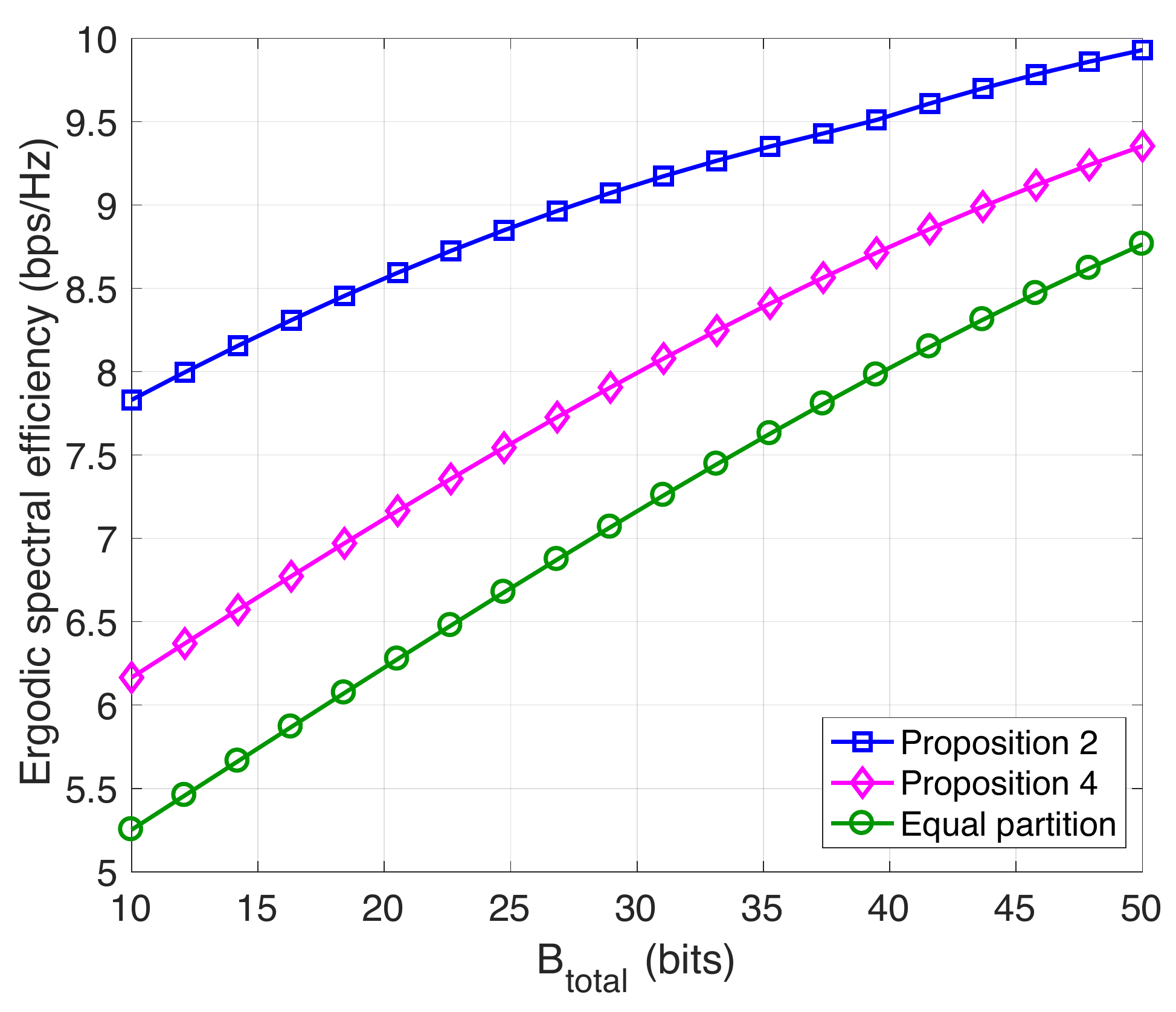}}}  \\ 
\mbox{(a)} &
\mbox{(b)} 
\end{array}$
\caption{The ergodic spectral efficiency comparison in a single-tier cooperative scenario. A $3$-tier HetNet is assumed. The simulation parameters are as follows: In (a), $N=L = 4$, $\lambda = 10^{-4}/\pi$, $ P = 20{\rm dBm}$, $\bar\delta_{1,L} = \{0.2, 0.04, 0.008\}$, and $\beta = 4$. In (b), the other parameters are same except that $\bar \delta_{1, L} = \{0.05,  0.0025, 0.0001\}$.
}
   \label{fig:coopt_special}
\end{figure}

Finally, we consider the single-tier case. We compare the ergodic spectral efficiency of Proposition \ref{prop:scheme1}, Proposition \ref{prop:scheme2}, and the equal partition in a single-tier network. We depict the results in Fig.~\ref{fig:coopt_special}, whose caption includes the simulation setting. 
%As shown in Fig.~\ref{fig:coopt_special}, Proposition \ref{prop:scheme1} increases the ergodic spectral efficiency compared to the equal partition. 
As shown in Fig.~\ref{fig:coopt_special}, at $B_{\rm total} = 10$, Proposition \ref{prop:scheme1} provides $23.1\%$ spectral efficiency gain in Fig.~\ref{fig:coopt_special}-(a) and $49.1\%$ gain in Fig.~\ref{fig:coopt_special}-(b). Proposition \ref{prop:scheme2} provides smaller gain than Proposition \ref{prop:scheme1}, specifically $13.2\%$ gain in Fig.~\ref{fig:coopt_special}-(a) and $17.1\%$ gain in Fig.~\ref{fig:coopt_special}-(b). 
The main reason of this performance gap is that Proposition \ref{prop:scheme1} allocates feedback more dynamically than Proposition \ref{prop:scheme2}. Specifically, since Proposition \ref{prop:scheme1} allocates feedback depending on intra-cluster geometry, two different clusters have different feedback allocation unless they have the same intra-cluster BSs' conditions. On the contrary, Proposition \ref{prop:scheme2} allocates feedback only depending on intra-cluster BSs' indices, therefore two different clusters have the same feedback allocation even if their intra-cluster geometries are different. 
%Proposition 2's feedback allocation for each cluster is different. On the contrary, Proposition 3 allocates feedback only as a function of the path-loss exponent and the BS's index. Even though the intra-cluster BSs' conditions are different, Proposition 3 gives the same amount of feedback to each BS. Summarizing, Proposition 2 allocates feedback more dynamically than Proposition 3. a function of intra-cluster geometry. 
 Despite of the decreased gain, Proposition \ref{prop:scheme2} can be useful since it uses fixed amount of feedback to each BS independent to the intra-cluster BS power. Since Proposition \ref{prop:scheme2} is a function of $\beta$ and $B_{\rm total}$, we do not have to modify the allocated feedback unless $\beta$ or $B_{\rm total}$ change. 
%In Fig.~\ref{fig:coopt}-(b), as a special case, we assume $K=1$ and compare Proposition \ref{prop:scheme1}, Proposition \ref{prop:scheme2}, and the baseline method. 
%The system parameters are presented in the caption of Fig.~\ref{fig:noncoopt}. 
% We use the floor function to meet the integer constraint. 
%In Fig.~\ref{fig:coopt}, we observe the ergodic spectral efficiency improvement by using Propositions. Specifically, assuming that $W = 20 {\rm MHz}$ and $B_{\rm total} = 30$, we obtain $2.4 {\rm Mbps}$ gain in Fig.~\ref{fig:coopt}-(a) and (b) by using Proposition \ref{prop:scheme1}. Proposition \ref{prop:scheme2} provides $0.84{\rm Mbps}$ gain in Fig.~\ref{fig:coopt}-(b), which is less gain than Proposition \ref{prop:scheme1}. One advantage of Proposition \ref{prop:scheme2} is that it can be applied regardless of intra-cluster user geometry due to its independence to individual user SIR.
%has advantages in implementation complexity since it does not depend on specific user locations. 

\section{Conclusions}
In this paper, we studied adaptive feedback partition problems in $K$-tier HetNets. We considered the non-cooperative and cooperative HetNet operations. Using stochastic geometry, we characterized the SIR CCDF and the ergodic spectral efficiency mainly as functions of the feedback and other relevant system parameters. 
%Specifically, to analyze the cooperative case, we derived a transformation lemma showing that the intensity measure of signal power in a $K$-tier HetNet can be transformed to the intensity measure of signal power in a statistically equivalent single-tier network. 
Leveraging the obtained expressions, we formulated the feedback partition problems and proposed solutions. 
The simulation results showed that the proposed feedback partitions bring some gains in the spectral efficiency compared to the equal partition. 
%The proposed partition methods provide some intuition regarding system design. 
Our major findings are summarized as follows. In the non-cooperative case, the feedback is useful only if the corresponding mean interference is small enough. In the cooperative case, allocating the feedback proportional to the intra-cluster BS power is efficient. 
We also showed that the proposed feedback allocation in the cooperative case is also useful for the general antenna case. 
Further, assuming a single-tier network as a special case, the effective cluster size increases with the square root of the total feedback.
%By adaptively partitioning the feedback bits according the proposed criterions, users in a cooperative or a heterogeneous network achieve gains in the ergodic spectral efficiency. 
%The proposed criterions provide preferred channel codebook sizes 
%Especially, Proposition \ref{prop:scheme2} for a cooperative network and Proposition \ref{prop:het} for a heterogeneous network are not functions of specific user geometry, so that they provide a static codebook size which is not changed by the short-term fading or the long-term path-loss. By this feature, the proposed channel codebook size can be shared to each user in a network. 

There are several possible directions for future work. One is to consider the channel quality information (CQI) feedback. Specifically, the obtained CQI can be exploited to select a preferred user, so that scheduling is also involved \cite{khoshnevis:jsac:13}. Interpreting this in a random network model is promising. 
In another direction, different sources of CSIT inaccuracy can be considered. While we only focus on the limited feedback, feedback delay \cite{hao:tcom:16} or the channel estimation error \cite{zaib:tcom:16} also degrades the CSIT accuracy, so that it is interesting to incorporate their effects into the performance characterization.  

\appendices
\section{Proof of Theorem 1} \label{appen:theo1}
%Starting from the definition of the SIR CCDF \eqref{eq:sir_ccdf_het},
%\begin{align}
%F_{{\rm SIR}_{\rm NC}^{k}}^{c} 
%%&= \mathbb{P} \left[ \frac{P_{k}\left\| {\bf{d}}_1^{k} \right\|^{-\beta} \left|\left({\bf{h}}_{1,1}^{k} \right)^* \hat {\bf{h}}_{1,1}^{k}\right|^2 }{\sum_{i=1}^{K} \sum_{{\bf{d}}_j^i \in \Phi_i \backslash {\bf{d}}_1^{k}}   P_{i} \left\| {\bf{d}}_j^{i} \right\|^{-\beta} \left| ({\bf{h}}_{1,j}^{i})^*{\bf{v}}_j^{i} \right|^2 } >\gamma  \right] \nonumber \\
%%&= \mathbb{P} \left[ \left|\left({\bf{h}}_{1,1}^{k} \right)^* \hat {\bf{h}}_{1,1}^{k}\right|^2 > \frac{\gamma}{P_{k} \left\|{\bf{d}}_1^{k} \right\|^{-\beta}} {\sum_{i=1}^{K} \sum_{{\bf{d}}_j^i \in \Phi_i \backslash {\bf{d}}_1^{k}}   {P_{i}} {\left\| {\bf{d}}_j^{i} \right\|^{-\beta}} \left| ({\bf{h}}_{1,j}^{i})^*{\bf{v}}_j^{i} \right|^2 }  \right] \nonumber \\
%&= \mathbb{P} \left[ \left|\left({\bf{h}}_{1,1}^{k} \right)^* \hat {\bf{h}}_{1,1}^{k}\right|^2 > {\gamma} {\sum_{i=1}^{K}I_i}  \right],
%\end{align}
Defining $I_i = \sum_{{\bf{d}}_j^i \in \Phi_i \backslash {\bf{d}}_1^{k}}   \frac{P_{i}}{P_k} \left\|{\bf{d}}_1^k \right\|^{\beta} \left\| {\bf{d}}_j^{i} \right\|^{-\beta} \left| ({\bf{h}}_{1,j}^{i})^*{\bf{v}}_j^{i} \right|^2 $, we rewrite the SIR CCDF  \eqref{eq:sir_ccdf_het} as 
\begin{align}
&\mathbb{P} \left[ \left|\left({\bf{h}}_{1,1}^{k} \right)^* \hat {\bf{h}}_{1,1}^{k}\right|^2 > {\gamma} {\sum_{i=1}^{K}I_i}  \right] 
= \mathbb{P}\left[\left\|{\bf{h}}_{1,1}^{k} \right\|^2 > \frac{\gamma}{\cos^2\theta_1 } {\sum_{i=1}^{K}I_i} \right]
\nonumber \\
& \mathop {=} \limits^{(a)} \mathbb{E}\left[ \mathbb{E}\left[\left.   \sum_{m=0}^{N_{k}-1} \frac{\gamma^m}{m!}  \frac{( \sum_{i=1}^{K}I_i) ^m }{\cos^{2m}\theta_1 }  \exp\left( -\frac{\gamma\sum_{i=1}^{K} I_i }{\cos^2 \theta_1} \right) \right| \cos^2 \theta_1 , \{I_i\}_{i=1,...,K} \right] \right] \nonumber \\
&\mathop = \limits^{(b)} \sum_{m=0}^{N_{k}-1} \frac{\gamma^m}{m!} (-1)^m \left. \frac{\partial^m \CMcal{L}_{I/\cos^2\theta_1}(s)}{\partial s^m} \right| _{s = {\gamma}}, 
%\mathbb{P}\left[\left\|{\bf{h}}_1 \right\| ^2> \gamma \frac{1}{\cos^2\theta_1}\left\| {\bf{d}}_1 \right\|^{\beta}   { \sum_{i=2}^{\infty}\left\| {\bf{d}}_i \right\|^{-\beta} \left| {\bf{h}}_{1,i}^*{\bf{v}}_i \right|^{2} }\right] \nonumber \\
\end{align}
where (a) follows that $\left\| {\bf{h}}_{1,1}^{k} \right\|^2$ follows the Chi-squared distribution with $2N_{k}$ degrees of freedom and (b) follows the derivative property of the Laplace transform, which is
%Now we use the derivative property of the Laplace transform, which is 
$\mathbb{E}\left[ X^m e^{-sX}\right] = (-1)^m \partial^m \CMcal{L}_X (s) / \partial s^m$ and $I = \sum_{i=1}^{K}I_i$. Now we obtain $\CMcal{L}_{I/\cos^2\theta_1}(s)$. 
\begin{align} \label{eq:laplace_hetnet}
\CMcal{L}_{I/\cos^2\theta_1}(s) 
%&= \mathbb{E}\left[e^{-s\frac{I}{\cos^2 \theta_1}} \right] \nonumber \\
&= \mathbb{E}\left[e^{-\frac{s}{\cos^2 \theta_1}{\left\| {\bf{d}}_{1}^{k} \right\|^{\beta}   { \sum_{i=1}^{K} \sum_{{\bf{d}}_j^i \in \Phi_i \backslash {\bf{d}}_1^{k}}  \frac{ P_{i}}{P_{k}} \left\| {\bf{d}}_j^{i} \right\|^{-\beta} \left| ({\bf{h}}_{1,j}^{i})^*{\bf{v}}_j^{i} \right|^2 }}{}} \right], {\rm Denoting }\; z = \frac{s}{\cos^2\theta_1},\nonumber \\
& \mathop {=} \limits^{(a)} \prod_{i=1}^{K} \mathbb{E}_{R, \cos^2\theta_1}\left[ \mathbb{E}_{\Phi_i }\left[\left. \prod_{{\bf{d}}_j^i \in \Phi_i \backslash {\bf{d}}_1^k}
\frac{1}{1+z\frac{P_i}{P_{k}}R^{\beta}\left\|{\bf{d}}_j^{i}  \right\|^{-\beta}}
%e^{-zR^{\beta}\sum_{i=2}^{\infty}\left\| {\bf{d}}_i \right\|^{-\beta} \left| {\bf{h}}_{1,i}^*{\bf{v}}_i \right|^{2} } 
\right| \left\| {\bf{d}}_1^{k} \right\| = R, \cos^2\theta_1 \right] \right]  \nonumber \\
&= \prod_{i=1}^{K}\mathbb{E}_{R, \cos^2\theta_1}\left[\exp\left(-2\pi \lambda_k \int_{R_i}^{\infty} \left(1 - \frac{1}{1+z\frac{P_i}{P_{k}} R^{\beta} r^{-\beta}} \right) r {\rm d} r\right) \right],
\end{align}
where $R_i = \left( \frac{P_i S_i}{P_{k} S_k}\right)^{1/\beta} R$ and (a) follows the independence between each tier. By leveraging the proof of Theorem 1 in \cite{hsjo:2012_twc}, \eqref{eq:laplace_hetnet} is calculated as
\begin{align}
&\CMcal{L}_{I/\cos^2\theta_1}(s) = \prod_{i=1}^{K} \mathbb{E}_{R, \cos^2\theta_1}\left[ \exp\left(-\pi \lambda_i \left(\frac{P_i S_i}{P_{k} S_k} \right)^{2/\beta}R^{2} \CMcal{D}\left(\frac{s}{\cos^2\theta_1} \left( \frac{S_k}{S_i}\right) , \beta\right) \right) \right] \nonumber \\
% \nonumber \\
%&=\mathbb{E}_{R, \cos^2\theta_1}\left[ \exp\left(-\pi \sum_{i=1}^{K}\lambda_i \left(\frac{P_i S_i}{P_{k} S_k} \right)^{2/\beta}R^{2} \CMcal{D}\left(\frac{s}{\cos^2\theta_1}\left( \frac{S_k}{S_i}\right) , \beta\right) \right) \right] \nonumber \\
%&= \int_{0}^{2^{-\frac{B_{k}}{N_{k}-1}}} \!\!\!\!\!\!\!\!\!\!\! 2^{B_k} (N_k-1)x^{N_k-2} \!\!\! \int_{0}^{\infty} \!\!\!\! f_{\left\| {\bf{d}}_1^{k}\right\|}(r) \exp\left(-\pi \sum_{i=1}^{K}\lambda_i \left(\frac{P_i S_i}{P_{k}S_k} \right)^{\frac{2}{\beta}} \!\!\! r^{2} \CMcal{D}\left(\frac{s}{1-x}\left(\frac{S_k}{S_i} \right) , \beta\right) \right)\! {\rm d} r  {\rm d} x \nonumber \\
&= \int_{0}^{2^{-\frac{B_{k}}{N_{k}-1}}} {2^{B_k} (N_k-1)x^{N_k-2}}\cdot \frac{\sum_{i=1}^{K} \lambda_i \left(\frac{P_i S_i}{P_k S_k} \right)^{2/\beta} }{\sum_{i=1}^{K}\lambda_i \left(\frac{P_i S_i}{P_k S_k} \right)^{2/\beta} \left[1+\CMcal{D}(\frac{s}{1-x}\left(\frac{S_k}{S_i} \right), \beta) \right]} {\rm d} x.
\end{align}
This completes the proof. $\hfill \blacksquare$

\section{Proof of Theorem 2} \label{appen:theo2}
For the desired channel, since ${\bf{v}}_{i_1}$ is independent to ${\bf{h}}_{1,i_1}^*$ and isotropic, $\left|{\bf{h}}_{1,i_1}^* {\bf{v}}_{i_1} \right|^2$ follows the exponential distribution with unit mean. For the intra-cluster interference $I_{\rm In }$ ($2 \le \ell \le L$), $\left| {\bf{h}}_{1,i_{\ell}}^* {\bf{v}}_{i_\ell}\right|^2$ is equivalent to $\left\| {\bf{h}}_{1,i_{\ell}}\right\|^2 \sin^2\theta_{i_\ell} \beta\left( 1,N_{\pi(i_{\ell})}-2\right)$, where $\beta\left(1,N_{\pi(i_{\ell})}-2 \right)$ is a Beta random variable that follows ${\rm Beta}\left(1,N_{\pi(i_{\ell})}-2 \right)$ and $\sin^2\theta_{i_\ell}$ follows \eqref{def:q_error_cdf}. We note that this is from the derivation in \cite{yoo:jsac:07}.
By the distribution of a product of a Gamma random variable and a Beta random variable \cite{prodgammabeta},
$\left\| {\bf{h}}_{1,{i_\ell}}\right\|^2 \sin^2\theta_{i_\ell} \beta\left( 1,N_{\pi(i_{\ell})}-2\right)$ boils down to $\Gamma\left(1, \delta \right)$ with $\delta = 2^{-\frac{B_{i_\ell}}{N_{\pi(i_{\ell})}-1}}$. Since we assume that the intra-cluster BS only uses $L$ antennas, $N_{\pi(i_{\ell})} = L$ for $i_{\ell} \in \CMcal{C}$. Accordingly, the Laplace transform of the intra-cluster interference fading $\left|{\bf{h}}_{1, i_{\ell}}^* {\bf{v}}_{i_{\ell}} \right|^2$ is
\begin{align} \label{eq:laplace_int_coop}
\mathbb{E}\left[e^{-s\left| {\bf{h}}_{1, i_{\ell}}^* {\bf{v}}_{i_{\ell}} \right|^2} \right] = \frac{1}{1+ s2^{-\frac{B_{i_{\ell}}}{L-1}}}.
\end{align}
Finally, for the out-of-cluster interference links $I_{\rm Out}$ ($L \le \ell $), $\left| {\bf{h}}_{1,{i_{\ell}}}^* {\bf{v}}_{i_{\ell}} \right|^2$ is an exponential random variable with unit mean due to the random beamforming effect.
By leveraging these results, the SIR CCDF \eqref{def:sir_ccdf_coop} is written as follows
\begin{align} 
&\mathbb{P}\left[\frac{ P_{\pi(i_1)}\left\| {\bf{d}}_{i_1} \right\|^{-\beta}  \left| {\bf{h}}_{1,i_1}^* {\bf{v}}_{i_1}\right|^2}{I_{\rm In} + I_{\rm Out}} \ge \gamma \right] \nonumber \\
&\mathop {=}^{(a)} \mathbb{E}\left[\prod_{i_{\ell} \in \CMcal{C}\backslash i_1} e^{-\gamma  \delta_{1,\ell}\left|{\bf{h}}_{1,i_{\ell}}^* {\bf{v}}_{i_\ell}  \right|^2} \right] \cdot \mathbb{E} \left[\prod_{j \in \mathbb{N}\backslash \CMcal{C}} e^{-\gamma \frac{P_{\pi(j)}}{P_{\pi(i_1)}} \left\| {\bf{d}}_{i_1}\right\|^{\beta} \left\| {\bf{d}}_j \right\|^{-\beta}\left|  {\bf{h}}_{1,j}^* {\bf{v}}_j \right|^2   } \right] \nonumber \\
%&\mathop {=}^{(b)} \prod_{i_\ell \in \CMcal{C}\backslash i_1} \left(\frac{1}{1+\gamma \delta_{1,\ell}  2^{-\frac{B_{i_\ell}}{N_{\pi(i_{\ell})}-1}} } \right) \cdot  \prod_{i = 1}^K \mathbb{E}_R\left[\mathbb{E} \left[\left.
%\prod_{{\bf{d}}_j^i \in \Phi_i \backslash \CMcal{B}\left(0, R_i \right)} e^{-\gamma \delta_{1,i_L} \frac{P_i}{P_{i_L}} R_i^{\beta} \left\| {\bf{d}}_i\right\|^{-\beta} \left|{\bf{h}}_{1,i}^* {\bf{v}}_i \right|^2}\right|  \left\|{\bf{d}}_L  \right\| = R \right] \right] \nonumber \\
&\mathop {=}^{(b)} \underbrace{\prod_{i_\ell \in \CMcal{C}\backslash i_1} \left(\frac{1}{1+\gamma \delta_{1,\ell}  2^{-\frac{B_{i_\ell}}{L-1}} } \right)}_{(c)}  \!\cdot\!   \underbrace{\prod_{i = 1}^K \mathbb{E}_R\left[\mathbb{E} \left[\left.
\prod_{{\bf{d}}_j^i \in \Phi_i \backslash \CMcal{B}\left(0, R_i \right)} \left(\frac{1}{1+ \gamma \delta_{1,L} \frac{P_i}{P_{i_L}} R^{\beta} \left\| {\bf{d}}_i\right\|^{-\beta} }\right) \right|  \left\|{\bf{d}}_{i_L}  \right\| = R \right] \right] }_{(d)}
%&\mathop {=} \prod_{i_\ell \in \CMcal{C}\backslash i_1} \left(\frac{1}{1+\gamma \delta_{1,\ell}  2^{-\frac{B_{i_\ell}}{N_{\pi(i_{\ell})}-1}} } \right) \cdot  \mathbb{E}_R\left[\mathbb{E} \left[\left.
%\prod_{{\bf{d}}_i \in \Phi \backslash \CMcal{B}\left(0, R \right)} e^{-\gamma \left(\delta_{1,N} \right)^{\beta} R^{\beta} \left\| {\bf{d}}_i\right\|^{-\beta} \left|{\bf{h}}_{1,i}^* {\bf{v}}_i \right|^2}\right|  \left\|{\bf{d}}_L  \right\| = R \right] \right] \nonumber \\
%& \mathop{=}^{(c)} \prod_{i_{\ell} \in \CMcal{C}\backslash i_1} \left(\frac{1}{1+\gamma \delta_{1,\ell}   2^{-\frac{B_{i_\ell}}{N_{\pi(i_\ell)}-1}} } \right) \mathbb{E}_{R}\left[\exp\left(- \frac{2\pi \lambda  \gamma \left(\delta_{1,L} \right)^\beta R^2} {\beta-2}  {}_2F_1\left(1, 1-\frac{2}{\beta} , 2-\frac{2}{\beta}, -\gamma (\delta_{1,N})^\beta \right)   \right)\right], \label{eq:sir_ccdf_coop_conditioned}
\end{align}
%\begin{align} 
%&\mathbb{P}\left[\frac{\left| {\bf{h}}_{1,1}^* {\bf{v}}_1\right|^2}{I_{\rm In}(\bar \delta_{1,  L}, \bar B_{L}) + I_{\rm Out}} \ge \gamma \right] \nonumber \\
%&\mathop {=}^{(a)} \mathbb{E}\left[\prod_{\ell=2}^{L} e^{-\gamma \left( \delta_{1,\ell}\right)^{\beta}\left|{\bf{h}}_{1,\ell}^* {\bf{v}}_\ell  \right|^2} \right] \cdot \mathbb{E} \left[\prod_{i=L+1}^{\infty} e^{-\gamma \left\| {\bf{d}}_1\right\|^{\beta} \left\| {\bf{d}}_i \right\|^{-\beta}\left|  {\bf{h}}_{1,i}^* {\bf{v}}_i \right|^2   } \right] \nonumber \\
%&\mathop {=}^{(b)} \prod_{\ell=2}^{L} \left(\frac{1}{1+\gamma \left(\delta_{1,\ell}  \right)^{\beta} 2^{-\frac{B_\ell}{N-1}} } \right) \cdot  \mathbb{E}_R\left[\mathbb{E} \left[\left.
%\prod_{{\bf{d}}_i \in \Phi \backslash \CMcal{B}\left(0, R \right)} e^{-\gamma \left(\delta_{1,N} \right)^{\beta} R^{\beta} \left\| {\bf{d}}_i\right\|^{-\beta} \left|{\bf{h}}_{1,i}^* {\bf{v}}_i \right|^2}\right|  \left\|{\bf{d}}_L  \right\| = R \right] \right] \nonumber \\
%& \mathop{=}^{(c)} \prod_{\ell=2}^{L} \left(\frac{1}{1+\gamma \left(\delta_{1,\ell}  \right)^{\beta} 2^{-\frac{B_\ell}{N-1}} } \right) \mathbb{E}_{R}\left[\exp\left(- \frac{2\pi \lambda  \gamma \left(\delta_{1,L} \right)^\beta R^2} {\beta-2}  {}_2F_1\left(1, 1-\frac{2}{\beta} , 2-\frac{2}{\beta}, -\gamma (\delta_{1,N})^\beta \right)   \right)\right], \label{eq:sir_ccdf_coop_conditioned}
%\end{align}
where $R_i = \left( \frac{P_i S_i}{P_{k} S_k}\right)^{1/\beta} R$ with $\left\| {\bf{d}}_{i_L} \right\| = R$. We explain each step of the above derivation as follows: (a) follows that the desired link's signal power ($\left|{\bf{h}}_{1, i_1}^* {\bf{v}}_{i_1} \right|^2$) is distributed as the exponential distribution with unit mean and there is independence between the intra-cluster interference and the out-of-cluster interference. (b) comes from \eqref{eq:laplace_int_coop} and $\left| {\bf{h}}_{1,{i_{\ell}}}^* {\bf{v}}_{i_{\ell}} \right|^2 \sim \rm{exp}(1)$ for $\ell \ge L$. We note that (c) indicates the Laplace transform of the intra-cluster interference $I_{\rm In}$ and (d) is the Laplace transform of the out-of-cluster interference. 
Since we assume the fixed intra-cluster BSs' conditions, (c) involves no randomness. 
Focusing on (d), we assume that the furthest BS in the coordination set $\CMcal{C}$ is included in the $k$-th tier, i.e., $\pi(i_L) = k$. Then, we have the following due to the probability generating functional of the PPP. 
%Under the assumption that the furthest BS in the coordination set $\CMcal{C}$ is in the $k$-th tier, we have
\begin{align} \label{eq:out_int_coop}
({d}) 
%&= \prod_{i=1}^{K}\mathbb{E}_{R}\left[\exp\left(-2\pi \lambda_k \int_{R_i}^{\infty} \left(1 - \frac{1}{1+\gamma \delta_{1,L}\frac{P_i}{P_{k}} R^{\beta} r^{-\beta}} \right) r {\rm d} r\right) \right] \nonumber \\
%&= \prod_{i=1}^{K} \mathbb{E}_{R}\left[ \exp\left(-2 \pi \lambda_i \frac{\gamma \delta_{1,L} }{\beta-2} \left( \frac{P_i}{P_k} \right) R^2 \left( \frac{P_i S_i}{P_k S_k} \right)^{2/\beta} \left(  \frac{P_i S_i}{P_k S_k}  \right)^{-1} {}_2F_1\left(1, 1-\frac{2}{\beta}, 2-\frac{2}{\beta}, -\gamma \delta_{1,L} \frac{P_i}{P_k} \left(\frac{P_i S_i}{P_k S_k} \right)^{-1} \right) \right) \right] \nonumber \\
%&= \prod_{i=1}^{K} \mathbb{E}_{R}\left[ \exp\left(- \pi \lambda_i R^2 \left( \frac{P_i S_i}{P_k S_k} \right)^{2/\beta} \left(  \frac{S_i}{S_k}  \right)^{-1} {}_2F_1\left(1, 1-\frac{2}{\beta}, 2-\frac{2}{\beta}, -\gamma \delta_{1,L}  \left(\frac{S_i}{S_k} \right)^{-1} \right) \right) \right] \nonumber \\
&= \mathbb{E}_{R}\left[ \exp\left(- \sum_{i=1}^{K} \pi \lambda_i R^2 \left( \frac{P_i S_i}{P_k S_k} \right)^{2/\beta} \CMcal{D}\left(\gamma \delta_{1,L} \left(\frac{S_k}{S_i}\right), \beta \right) \right) \right],
%\left(  \frac{S_i}{S_k}  \right)^{-1} {}_2F_1\left(1, 1-\frac{2}{\beta}, 2-\frac{2}{\beta}, -\gamma \delta_{1,L}  \left(\frac{S_i}{S_k} \right)^{-1} \right) \right) \right] 
\end{align}
where $\CMcal{D}\left( \cdot, \cdot \right)$ is defined in \eqref{eq:het:dfunc}. 
%Note that ${}_2F_1\left(\cdot, \cdot, \cdot, \cdot \right) $ denotes the gaussian hypergeometry function.
Now we marginalize \eqref{eq:out_int_coop} with respect to $R$, whose the distribution function is obtained in Lemma \ref{lem:het_L_closest}.
%\begin{align} \label{dist:kthpoint_PPP}
%f_{\left\| {\bf{d}}_N \right\|}(R) = \frac{2(\lambda\pi R^2)^L}{R \Gamma\left( L \right)}e^{-\lambda \pi R^2}.
%\end{align}
The Laplace transform of the out-of-cluster interference is
\begin{align} \label{eq:laplace_oci_unconditioned}
&\mathbb{E}_{R}\left[ \exp\left(- \sum_{i=1}^{K} \pi \lambda_i R^2 \left( \frac{P_i S_i}{P_k S_k} \right)^{2/\beta} \CMcal{D}\left(\gamma \delta_{1,L} \left(\frac{S_k}{S_i}\right), \beta \right) \right) \right] \nonumber \\
%&= \int_{0}^{\infty}  \exp\left(- \sum_{i=1}^{K} \pi \lambda_i R^2 \left( \frac{P_i S_i}{P_k S_k} \right)^{2/\beta} \CMcal{D}\left(\gamma \delta_{1,L} \left(\frac{S_k}{S_i}\right), \beta \right) \right) f_{\left\| {\bf{d}}_{i_L} \right\|}(R) {\rm d} R \nonumber \\
%\end{align}
%\begin{align}
&= \left( \frac{\sum_{i=1}^{K} \lambda_i \left(\frac{P_i S_i}{P_k S_k} \right)^{2/\beta} }{\sum_{i=1}^{K} \lambda_i \left(\frac{P_i S_i}{P_k S_k} \right)^{2/\beta}\left[1+ \CMcal{D}\left(\gamma \delta_{1, L} \left(\frac{S_k}{S_i} \right), \beta \right) \right]} \right)^L,
\end{align}
%\begin{align} \label{eq:laplace_oci_unconditioned}
%&\mathbb{E}_{R}\left[\exp\left(- 2\pi \lambda \frac{\gamma \left(\delta_{1,L} \right)^\beta R^2} {\beta-2}  \cdot {}_2F_1\left(1, 1-\frac{2}{\beta} , 2-\frac{2}{\beta}, -\gamma (\delta_{1,L})^\beta \right)   \right)\right] \nonumber \\
%&= \int_{0}^{\infty} \exp\left(- 2\pi \lambda \frac{\gamma \left(\delta_{1,L} \right)^\beta R^2} {\beta-2}  \cdot {}_2F_1\left(1, 1-\frac{2}{\beta} , 2-\frac{2}{\beta}, -\gamma (\delta_{1,L})^\beta \right)   \right) \frac{\left(2\lambda \pi R^2\right)^L}{R \Gamma(L)} e^{-\lambda \pi R^2} {\rm d} R \nonumber \\
%%\end{align}
%%\begin{align}
%&= \left( \frac{1}{1+\CMcal{D}\left(\gamma (\delta_{1,L})^{\beta},\beta \right)} \right)^L, 
%\end{align}
%where 
%\begin{align}
%\CMcal{D}( x, y) = \frac{2x}{y-2} {}_2F_1\left(1, 1-\frac{2}{y}, 2-\frac{2}{y}, -x \right),
%\end{align}
%With $L=N$ by the assumption, plugging \eqref{eq:laplace_oci_unconditioned} into \eqref{eq:sir_ccdf_coop_conditioned} completes the proof.
which completes the proof. 
\hfill $\blacksquare$

\bibliographystyle{IEEEtran}
\bibliography{ref_adaptivefb_jour}

\begin{thebibliography}{10}
\providecommand{\url}[1]{#1}
\csname url@samestyle\endcsname
\providecommand{\newblock}{\relax}
\providecommand{\bibinfo}[2]{#2}
\providecommand{\BIBentrySTDinterwordspacing}{\spaceskip=0pt\relax}
\providecommand{\BIBentryALTinterwordstretchfactor}{4}
\providecommand{\BIBentryALTinterwordspacing}{\spaceskip=\fontdimen2\font plus
\BIBentryALTinterwordstretchfactor\fontdimen3\font minus
  \fontdimen4\font\relax}
\providecommand{\BIBforeignlanguage}[2]{{%
\expandafter\ifx\csname l@#1\endcsname\relax
\typeout{** WARNING: IEEEtran.bst: No hyphenation pattern has been}%
\typeout{** loaded for the language `#1'. Using the pattern for}%
\typeout{** the default language instead.}%
\else
\language=\csname l@#1\endcsname
\fi
#2}}
\providecommand{\BIBdecl}{\relax}
\BIBdecl

\bibitem{park:gc:16}
J.~Park and R.~W. Heath, ``Adaptive feedback partitions in dynamic zero-forcing
  beamforming based on stochastic geometry,'' in \emph{Proc. IEEE Workshop on
  Global Comm. Conf.}, Dec. 2016, pp. 1--6.

\bibitem{bhaga:tsp:11}
R.~Bhagavatula and R.~Heath, ``Adaptive bit partitioning for multicell
  intercell interference nulling with delayed limited feedback,'' \emph{IEEE
  Trans. Sig. Proc.}, vol.~59, no.~8, pp. 3824--3836, Aug. 2011.

\bibitem{ny:twc:11_adap}
N.~Lee and W.~Shin, ``Adaptive feedback scheme on {K}-cell {MISO} interfering
  broadcast channel with limited feedback,'' \emph{IEEE Trans. Wireless Comm.},
  vol.~10, no.~2, pp. 401--406, Feb. 2011.

\bibitem{yuan:twc:13}
F.~Yuan, C.~Yang, G.~Wang, and M.~Lei, ``Adaptive channel feedback for
  coordinated beamforming in heterogeneous networks,'' \emph{IEEE Trans.
  Wireless Comm.}, vol.~12, no.~8, pp. 3980--3994, Aug. 2013.

\bibitem{kerret:tit:12}
P.~de~Kerret and D.~Gesbert, ``Degrees of freedom of the network mimo channel
  with distributed csi,'' \emph{IEEE Trans. Info. Th.}, vol.~58, no.~11, pp.
  6806--6824, Nov. 2012.

\bibitem{rao:tsp:13}
X.~Rao, L.~Ruan, and V.~K.~N. Lau, ``Limited feedback design for interference
  alignment on {MIMO} interference networks with heterogeneous path loss and
  spatial correlations,'' \emph{IEEE Trans. Sig. Proc.}, vol.~61, no.~10, pp.
  2598--2607, May 2013.

\bibitem{niu:wpmc:14}
Q.~Niu, Z.~Zeng, T.~Zhang, Q.~Gao, and S.~Sun, ``Interference alignment and bit
  allocation in heterogeneous networks with limited feedback,'' in \emph{Proc.
  IEEE Symp. on Wireless Pers. Multimedia Comm.}, Sep. 2014, pp. 514--519.

\bibitem{akoum:tsp:13}
S.~Akoum and R.~W. Heath, ``Interference coordination: Random clustering and
  adaptive limited feedback,'' \emph{IEEE Trans. Sig. Proc.}, vol.~61, no.~7,
  pp. 1822--1834, Apr. 2013.

\bibitem{li:tcom:15}
C.~Li, J.~Zhang, M.~Haenggi, and K.~B. Letaief, ``User-centric intercell
  interference nulling for downlink small cell networks,'' \emph{IEEE Trans.
  Comm.}, vol.~63, no.~4, pp. 1419--1431, Apr. 2015.

\bibitem{jh:twc:16}
J.~Park, N.~Lee, J.~G. Andrews, and R.~W. Heath, ``On the optimal feedback rate
  in interference-limited multi-antenna cellular systems,'' \emph{IEEE Trans.
  Wireless Comm.}, vol.~15, no.~8, pp. 5748--5762, Aug. 2016.

\bibitem{kountouris:twc:12}
M.~Kountouris and J.~G. Andrews, ``Downlink {SDMA} with limited feedback in
  interference-limited wireless networks,'' \emph{IEEE Trans. Wireless Comm.},
  vol.~11, no.~8, Aug. 2012.

\bibitem{park:wcl:16}
J.~Park and R.~W. Heath, ``Multiple-antenna transmission with limited feedback
  in device-to-device networks,'' \emph{IEEE Wireless Comm. Lett.}, vol.~5,
  no.~2, pp. 200--203, Apr. 2016.

\bibitem{hsjo:2012_twc}
H.~S. Jo, Y.~J. Sang, P.~Xia, and J.~G. Andrews, ``Heterogeneous cellular
  networks with flexible cell association: {A} comprehensive downlink {SINR}
  analysis,'' \emph{IEEE Trans. Wireless Comm.}, vol.~11, no.~10, pp.
  3484--3495, Oct. 2012.

\bibitem{baccelli:inria}
F.~Baccelli and B.~Blaszczyszyn, \emph{{Stochastic Geometry and Wireless
  Networks, Volume I - Theory}}, ser. Foundations and Trends in
  Networking.\hskip 1em plus 0.5em minus 0.4em\relax {Now Publishers}, 2009,
  vol.~3.

\bibitem{zhang:tcom:15}
X.~Zhang and J.~G. Andrews, ``Downlink cellular network analysis with
  multi-slope path loss models,'' \emph{IEEE Trans. Comm.}, vol.~63, no.~5, pp.
  1881--1894, May 2015.

\bibitem{muk:tit:03}
K.~K. Mukkavilli, A.~Sabharwal, E.~Erkip, and B.~Aazhang, ``On beamforming with
  finite rate feedback in multiple-antenna systems,'' \emph{IEEE Trans. Info.
  Th.}, vol.~49, no.~10, pp. 2562--2579, Oct. 2003.

\bibitem{shen:twc:05}
S.~Zhou, Z.~Wang, and G.~B. Giannakis, ``Quantifying the power loss when
  transmit beamforming relies on finite-rate feedback,'' \emph{IEEE
  Transactions on Wireless Communications}, vol.~4, no.~4, pp. 1948--1957, Jul.
  2005.

\bibitem{gers:tit:79}
A.~Gersho, ``Asymptotically optimal block quantization,'' \emph{IEEE Trans.
  Info. Th.}, vol.~25, no.~4, pp. 373--380, Jul. 1979.

\bibitem{dhil:jsac:12}
H.~S. Dhillon, R.~K. Ganti, F.~Baccelli, and J.~G. Andrews, ``Modeling and
  analysis of {K}-tier downlink heterogeneous cellular networks,'' \emph{IEEE
  Jour. Select. Areas in Comm.}, vol.~30, no.~3, pp. 550--560, Apr. 2012.

\bibitem{hamdi:useful}
K.~A. Hamdi, ``A useful lemma for capacity analysis of fading interference
  channels,'' \emph{IEEE Trans. Comm.}, vol.~58, no.~2, pp. 411--416, Feb.
  2010.

\bibitem{park:tcom:17}
J.~Park, S.~Park, A.~Yazdan, and R.~W. Heath, ``Optimization of mixed-{ADC}
  multi-antenna systems for cloud-{RAN} deployments,'' \emph{IEEE Trans.
  Comm.}, vol.~PP, no.~99, pp. 1--1, 2017.

\bibitem{haenggi:tit:05}
M.~Haenggi, ``On distances in uniformly random networks,'' \emph{IEEE Trans.
  Info. Th.}, vol.~51, no.~10, pp. 3584--3586, Oct. 2005.

\bibitem{lee:twc:15}
N.~Lee, D.~Morales-Jimenez, A.~Lozano, and R.~W. Heath, ``Spectral efficiency
  of dynamic coordinated beamforming: {A} stochastic geometry approach,''
  \emph{IEEE Trans. Wireless Comm.}, vol.~14, no.~1, pp. 230--241, Jan. 2015.

\bibitem{khoshnevis:jsac:13}
B.~Khoshnevis, W.~Yu, and Y.~Lostanlen, ``Two-stage channel quantization for
  scheduling and beamforming in network {MIMO} systems: {F}eedback design and
  scaling laws,'' \emph{IEEE Jour. Select. Areas in Comm.}, vol.~31, no.~10,
  pp. 2028--2042, Oct. 2013.

\bibitem{hao:tcom:16}
C.~Hao and B.~Clerckx, ``Achievable sum {DoF} of the ${K}$-user {MIMO}
  interference channel with delayed {CSIT},'' \emph{IEEE Trans. Comm.},
  vol.~64, no.~10, pp. 4165--4180, Oct. 2016.

\bibitem{zaib:tcom:16}
A.~Zaib, M.~Masood, A.~Ali, W.~Xu, and T.~Y. Al-Naffouri, ``Distributed channel
  estimation and pilot contamination analysis for massive {MIMO-OFDM}
  systems,'' \emph{IEEE Trans. Comm.}, vol.~64, no.~11, pp. 4607--4621, Nov.
  2016.

\bibitem{prodgammabeta}
S.~Nadarajah and S.~Kotz, ``{On the product and ratio of Gamma and Beta random
  variables},'' \emph{Allgemeines Statistisches Archiv}, vol.~89, no.~4, pp.
  435--449, 2005.

\end{thebibliography}

\end{document}